\documentclass[sigconf]{acmart}
\usepackage[utf8]{inputenc}
\usepackage{microtype}

\title[When Can We Answer Queries Using Result-Bounded Data Interfaces?]{When Can We Answer Queries\texorpdfstring{\\}{ }Using Result-Bounded Data Interfaces?}
\author{Antoine Amarilli}
\affiliation{LTCI, T\'el\'ecom ParisTech, Universit\'e Paris--Saclay}
\author{ Michael Benedikt}
\affiliation{University of Oxford}

\usepackage{balance}
\usepackage{multirow}
\usepackage{amsmath}
\usepackage{amsthm}
\usepackage{xspace}
\usepackage{amsfonts}
\usepackage{needspace}
\usepackage{tabularx}
\usepackage{booktabs}
\usepackage{xcolor}
\newcommand{\smpr}{\kw{SMPR}}

\newcommand{\lift}{\kw{Lin}}

\newcommand{\gammasep}{\Gamma^{\kw{Sep}}}

\newcommand{\amd}{\kw{AMonDet}}
\newcommand{\canondb}{\kw{CanonDB}}
\newcommand{\ptime}{\kw{PTIME}}
\newcommand{\exptime}{\kw{EXPTIME}}
\newcommand{\expspace}{\kw{EXPSPACE}}
\newcommand{\arity}{\kw{Arity}}
\newcommand{\adom}{\kw{Adom}}
\newcommand{\dom}{\adom}
\newcommand{\gctwo}{\text{GC}^{2}}

\newcommand{\elimnd}{\kw{AxiomRB}}

\newcommand{\checkview}{R}

\newcommand{\detby}{\kw{DetBy}}
\newcommand{\relaxs}{\kw{ElimUB}}
\newcommand{\twoexp}{\kw{2EXPTIME}}
\newcommand{\return}{\kw{Return}}

\newcommand{\fds}{\Sigma_{\mathrm{FD}}}
\newcommand{\ids}{\Sigma_{\mathrm{ID}}}
\newcommand{\incd}{\text{ID}}
\newcommand{\uincd}{\text{UID}}

\newcommand{\wrt}{w.r.t.}

\newcommand{\determines}{\rightarrow}

\newcommand{\aschema}{\kw{Sch}}
\newcommand{\inmap}{\kw{InMap}}
\newcommand{\outmap}{\kw{OutMap}}
\newcommand{\accbind}{\kw{AccBind}}
\newcommand{\abind}{\accbind}

\newcommand{\accpart}{\kw{AccPart}}

\newcommand{\kw}[1]{{\mathsf{#1}}\xspace}
\newcommand{\profinfo}{ \kw{Prof}}

\newcommand{\univdirect}{\udirectory}
\newcommand{\udirectory}{\kw{Udirectory}}

\newcommand{\accessible}{\kw{accessible}}
\newcommand{\acc}[1]{\kw{\scriptscriptstyle Accessed} #1}

\newcommand{\mt}{\kw{mt}}
\newcommand{\aplan}{\kw{PL}}

\newtheorem{theorem}{Theorem}[section]
\newtheorem{claim}[theorem]{Claim}

\newtheorem{proposition}[theorem]{Proposition}
\newtheorem{corollary}[theorem]{Corollary}
\newtheorem{definition}[theorem]{Definition}

\newtheorem{example}[theorem]{Example}

\newtheorem{lemma}[theorem]{Lemma}

\newcommand{\myeat}[1]{}

\makeatletter
\def\@Opargbegintheorem#1#2#3#4{#4\trivlist
      \item[\hskip\labelsep{#3#1}]{#3#2\@thmcounterend\ }}

\newcommand{\definerep}[2]{%
\spnewtheorem*{#1rp}{#2}{\bf}{\itshape}
\newenvironment{#1rep}[2]{%
  \ifthenelse{\equal{##1}{*}}
  {\begin{#1rp}[\ref{##2}]}
  {\begin{#1}\label{##2}}}
{\ifthenelse{\equal{\@currenvir}{#1}}{\end{#1}}{\end{#1rp}}}
}

\newcommand{\card}[1]{\left|#1\right|}

\newcommand{\NN}{\mathbb{N}}

\newcommand{\gnf}{\kw{GNF}}
\newcommand{\np}{\kw{NP}}

\newcommand{\wclo}[2]{\widehat{#1}^#2}

\newcommand{\sign}{{\mathcal{S}}}
\newcommand{\sidesign}{\sign'}
\newcommand{\aselect}{\sigma}
\newcommand{\dep}{\delta}
\newcommand{\gdep}{\gamma}
\newcommand{\trig}{\tau}

\newcommand{\bounded}{\mathit{Bounded}}
\newcommand{\acyclic}{\mathit{Acyclic}}

\newcommand{\praccess}{\mathsf{pr}}
\newcommand{\udaccess}{\mathsf{ud}}

\newcommand{\attrfmt}[1]{\mathit{#1}}

\usepackage{colonequals}

\newcommand{\myparaskip}{}
\newcommand{\mylistskip}{}
\newcommand{\myparagraph}[1]{\paragraph*{#1.}}

\newcommand{\sectionarxiv}[2]{\section{#2}}

\newcommand{\myproof}{Proof }

\usepackage{paralist}
\renewenvironment{compactitem}{\begin{itemize}}{%
    \end{itemize}\ignorespacesafterend
}
\renewenvironment{compactenum}{\begin{enumerate}}{%
    \end{enumerate}\ignorespacesafterend
}

\hypersetup{
    colorlinks,
    linkcolor={red!50!black},
    citecolor={blue!50!black},
    urlcolor={blue!30!black}
}

\begin{abstract}
We consider answering queries on data 
available through \emph{access methods}, 
that provide lookup access to the tuples matching a given binding.
Such interfaces are common on the Web; further, they often have \emph{bounds} on
how many results they can return, e.g., because of pagination or rate limits.
We thus study \emph{result-bounded methods}, which 
may return only a limited number of tuples.
We study how to decide if a query is \emph{answerable} using result-bounded
  methods, i.e., how to compute a \emph{plan} that returns all answers to
  the query using the methods, assuming that the underlying data satisfies some
   integrity constraints.
  We first show how to reduce answerability to a query containment problem with
  constraints. Second, we show 
  ``schema simplification'' theorems describing when and how result bounded services can be
  used. Finally, we use these theorems to give decidability and complexity results about  answerability for
  common constraint classes.
\end{abstract}

\begin{document}
\copyrightyear{2018}
\acmYear{2018}
\setcopyright{rightsretained}
\acmConference[PODS'18]{35th ACM SIGMOD-SIGACT-SIGAI Symposium on Principles of Database Systems}{Extended version}{Appendices included}
\maketitle
\section{Introduction}

Web services 
expose programmatic interfaces to data.
Many of these services can be modeled as an
\emph{access method}:
given a set of arguments for some attributes of a relation,
the method returns all matching
tuples for the relation.

\begin{example} \label{ex:simple}
  Consider a Web service that exposes university employee information.
  The schema has a relation $\profinfo(\attrfmt{id}, \attrfmt{name}, \attrfmt{salary})$
  and an access method $\praccess$ on this relation: the input to $\praccess$ is
  the $\attrfmt{id}$
  of a professor, and an access to this method
  returns the $\attrfmt{name}$ and $\attrfmt{salary}$ of
  the professor. The schema also has a relation $\univdirect(\attrfmt{id},
  \attrfmt{address}, \allowbreak \attrfmt{phone})$, and an access method $\udaccess$: it
  has no input and returns the $\attrfmt{id}$, $\attrfmt{address}$, and
  phone number of all university employees.
\end{example}

Our goal is to answer queries using such services.
In the setting of Example~\ref{ex:simple}, the user queries are posed
on the relations $\profinfo$ and $\univdirect$, and we wish to answer
them using the methods $\praccess$ and $\udaccess$. 
To do so, we can exploit \emph{integrity constraints} that the data is known to satisfy:
for instance, the referential constraint $\tau$ that says that the
$\attrfmt{id}$ of every tuple in~$\profinfo$ is also in~$\univdirect$.

\begin{example}
  \label{exa:plani}
  Consider $Q_1(n): \exists i ~ \profinfo(i, n, 10000)$, the query
  that asks for the names of professors with salary~$10000$. 
  If we assume the integrity constraint~$\tau$,
  we can implement
  $Q_1$ as the following \emph{plan}: first access $\udaccess$ to get the set of all ids, and then
  access $\praccess$ with each id
to obtain the salary, filtering the results to return only the names 
with salary~$10000$.
This plan \emph{reformulates} $Q_1$ over the access methods:
it is equivalent to~$Q_1$ on all instances satisfying~$\tau$,
  and it only uses $\praccess$ and $\udaccess$ to access
  $\profinfo$ and $\univdirect$.
\end{example}

Prior work  (e.g.,~\cite{dln,ustods}) has 
formalized this reformulation task as an \emph{answerability} problem: given a schema
with access methods and integrity constraints, and given a query, determine if
we can answer the query using the methods. The query has to be answered in a
\emph{complete} way, i.e., without missing any results.
This prior work has led to implementations (e.g. ~\cite{usvldb14,usvldb15,bioint})
that can
determine how to evaluate a conjunctive query using
a collection of Web services, by generating a plan that makes calls to the services.

However, all these works assume that whenever we access a Web service, we will
always obtain \emph{all} tuples that match the access.
This is not a realistic assumption: to avoid wasting
resources and bandwidth, virtually all Web services impose a \emph{limit} on how
many results they will return.
For instance, the ChEBI service (chemical entities of
biological interest, see~\cite{bioint}) limits the output of lookup methods to 5000 entries,
while IMDb's web interfaces impose a limit of~10000~\cite{imdb}. 
Some services make it possible to request more
results beyond the limit, e.g., using pagination or continuation tokens,
but there is often a
\emph{rate limitation} on how many requests can be
made~\cite{fbapi,githubapi,twitterapi}, which also limits the total number of
obtainable results.
Thus, for many Web services, beyond a certain number
of results, we cannot assume that all matching tuples are returned.
In this work, we introduce \emph{result-bounded} methods to reason on these
services.

\begin{example}
  \label{exa:rbound}
The~$\udaccess$ method in 
Example~\ref{ex:simple} may have a result bound, e.g., it may return
  at most 100 entries. If this is the case, then the plan of
  Example~\ref{exa:plani} is not equivalent to~$Q_1$ as it may miss some result
  tuples.
\end{example}

Result-bounded methods make it very challenging to reformulate queries. Indeed,
they are \emph{nondeterministic}: if the number of results is more than
the result bound, then the Web service only returns a subset of results, usually according to
unknown criteria. 
For this reason, it is not even clear whether result-bounded methods can be
useful at all to answer queries in a complete way.
However, this may be the case:

\begin{example} \label{ex:existencecheck} Consider the schema of Example~\ref{ex:simple}
  and assume that $\udaccess$ has a result bound of~$100$ as in
  Example~\ref{exa:rbound}.
  Consider the query $Q_2: \exists i\, a\, p ~ \univdirect(i, a, p)$ asking if
there  is \emph{some} university employee. We can answer $Q_2$ with a plan  that
accesses
the  $\udaccess$ method and returns true if the output is non-empty.
It is not a problem that $\udaccess$ may omit
some result tuples, because we only want to know if it returns something.
This gives a first intuition: 
  result-bounded methods are useful to check for the existence of
  matching tuples.
\end{example}

Further, 
result-bounded methods can also help under integrity constraints such as keys or
functional dependencies:

\begin{example} \label{ex:fd}
  Consider the schema
 of Example~\ref{ex:simple} and the access method 
  $\udaccess_2$  on $\univdirect$
  that takes an $\attrfmt{id}$ as input and returns the 
  $\attrfmt{address}$ and phone number of tuples with this $\attrfmt{id}$.
  Assume that $\udaccess_2$ has a result bound of~$1$, i.e., returns at most one
  answer when given an $\attrfmt{id}$. Further assume 
   the functional dependency $\phi$: 
  each employee id 
  has exactly one $\attrfmt{address}$ (but possibly
  many phone numbers).
Consider the query~$Q_3$ asking for the address of 
  the employee with id~12345.
We can answer~$Q_3$ by calling $\udaccess_2$ with 12345 and projecting onto the
  $\attrfmt{address}$ field.
  Thanks to~$\phi$, we know that 
  the result will contain the 
  employee's address,
  even though 
  only one of the phone numbers will be returned.
This gives a second intuition:
  result-bounded methods are useful when there is a functional dependency that guarantees that
  some projection of the output is complete.
\end{example}

In this paper, we study how and when we can use result-bounded methods to
reformulate queries and obtain complete answers, formalizing in particular the
intuition of Examples~\ref{ex:existencecheck} and~\ref{ex:fd}.
We then show decidability and complexity results for the answerability problem.
We focus on two common classes of integrity constraints on
databases: \emph{inclusion dependencies} ($\incd$s), as in
Example~\ref{ex:existencecheck}, and \emph{functional dependencies} (FDs), as in
Example~\ref{ex:fd}. But we also show results for more expressive constraints:
see Table~\ref{tab:results} for a summary.

The first step of our study (Section~\ref{sec:reduce}) is to reduce the answerability 
problem to \emph{query containment under constraints}. Such a reduction is well-known
in the context of reformulation of queries over views \cite{NSV}, and in answering queries
with access methods without result bounds~\cite{thebook}. However, the
nondeterminism of result-bounded methods means that we cannot apply these
results directly. We nevertheless show that 
this reduction technique can still be applied in the presence of result
bounds. However, the resulting query containment problem involves complex cardinality constraints, 
so it does not immediately lead to decidability results.

Our second step (Section~\ref{sec:simplify})
is to show \emph{schema simplification results}, which explain why some
of the result bounds can be ignored for the answerability problem.
These results characterize how result-bounded methods are
useful: they capture and generalize the examples above.
For instance, 
we show that for constraints given as $\incd$s,
result-bounded methods are only useful
as an \emph{existence check} as in Example~\ref{ex:existencecheck}.
We also show that, for FD constraints,
result-bounded methods are only useful to access 
the \emph{functionally-determined part of the
output}, as in
Example~\ref{ex:fd}.
The proofs
introduce  a technique of \emph{blowing up models}, i.e., we enlarge them to increase
the number of outputs of an access, without violating constraints or changing query answers.

Third, in Section~\ref{sec:complexity}, we use the simplification results to
deduce that answerability is decidable for these constraint classes, and
show tight complexity bounds: we show that the problem is $\np$-complete for FDs,
and $\exptime$-complete for $\incd$s. We refine the latter result to show that
answerability is $\np$-complete for \emph{bounded-width} $\incd$s,
which export only a constant number of variables. This refinement is proved using
ideas of Johnson and Klug~\cite{johnsonklug}, along with a
\emph{linearization} technique of potentially independent interest: we show how
the constraints used to reason about answerability can be ``simulated'' with
 restricted  inclusion dependencies.

In Section~\ref{sec:simplifychoice}, we study more expressive constraint
classes, beyond $\incd$s and FDs. We do so using a weaker form of simplification, called
\emph{choice simplification}, which replaces all result bounds by~$1$:
this intuitively
implies that the number of results does not matter.  
We show that it suffices to consider the choice
simplification
for a huge class of  constraints, including all
TGDs, and also constraints consisting of FDs and
U$\incd$s. 
In Section~\ref{sec:complexitychoice}, we use this 
technique to show that decidability of  answerability holds much more broadly:
in particular it holds for a wide range of classes where query containment is decidable.
We conclude the paper by giving some limits to schema simplification and decidability of answerability
(Section~\ref{sec:general}), followed by conclusions (Section~\ref{sec:conc}).

This is the full version of the conference paper~\cite{confpaper}. Most proofs
are deferred to the appendix.

\myparaskip
\myparagraph{Related work}
Our paper relates
to a line of work about finding plans to answer
queries using access methods. The initial line of work considered finding
equivalent ``executable rewritings''~--- conjunctive queries
where the atoms are ordered in a way compatible with the access patterns. This was studied
first without integrity
constraints~\cite{access1,access2},
and then
with disjunctive TGD constraints \cite{dln}.
Later \cite{ustods,thebook} formulated the problem of finding a \emph{plan}
that answers the query over the access patterns, distinguishing two notions of plans with access methods:
one with arbitrary relational operators in middleware and another without the difference
operator. They
studied the problem of getting plans of both types in the presence of integrity
constraints:
following~\cite{dln},
they reduced the search for executable rewritings to 
query containment under constraints.
Further, \cite{ustods,thebook} also related the reduction to a semantic notion of 
determinacy, originating
from the work of Nash, Segoufin, and Vianu \cite{NSV} in the context of views.
Our paper extends the reduction to query containment 
in the presence of result bounds,
relying heavily on the techniques of \cite{dln,NSV,ustods,thebook}.

Non-determinism in query
 languages has been studied in other contexts \cite{av,asv}. 
However, the topic of this work, namely, using non-deterministic Web services to
implement deterministic queries, has not been studied.
Result bounds are reminiscent of \emph{cardinality constraints}, for
which
the answerability problem has been studied~\cite{wenfeifloris1}. However, the
two are different: whereas cardinality constraints restrict the 
\emph{underlying data},
result bounds concern the \emph{access methods} to the data, and makes them
\emph{non-deterministic}:
this has not been studied in the past.
In fact, surprisingly, 
our schema simplification results (in Sections~\ref{sec:simplify}
and~\ref{sec:simplifychoice}) imply that answerability with result bounds can be
decided \emph{without} reasoning about cardinality constraints at all.

To study our new setting with result-bounded methods, we introduce several 
specific techniques
to reduce to a decidable query containment problem, e.g., 
we give determinacy notions for non-deterministic services and present the technique of
``blowing up models''.
The additional technical  tools needed for the complexity analysis revolve around analysis of the chase.
While many components of this analysis are specific to the constraints produced by our problem,
the analysis includes 
a \emph{linearization} method, which we believe may be more generally applicable.
This method relates to the Datalog$^{\pm}$ agenda of
getting bounds for query answering with restricted classes of constraints
\cite{taming,shy,sticky}, because
our own method deals with guarded rules as in \cite{taming}.
Linearization can
thus be understood as a refinement of a technique
from~\cite{gmp}: we  isolate classes that can be reduced
to well-behaved classes of linear TGDs, where more specialized bounds \cite{johnsonklug} can be applied.

\section{Preliminaries} \label{sec:prelims}

\myparagraph{Data and queries}
We consider a \emph{relational signature} $\sign$ that consists of
a set of \emph{relations} with an associated
\emph{arity} (a positive integer).
The \emph{positions} of a relation~$R$ of~$\sign$ are $1, \ldots, n$ where
$n$ is the arity of~$R$.
An \emph{instance} of~$R$ is a
set of~$n$-tuples (finite or infinite), and an \emph{instance} $I$ of~$\sign$
consists of instances
for each relation of~$\sign$. We equivalently see~$I$ as a set
of \emph{facts} $R(a_1 \ldots a_n)$ for each tuple 
$(a_1 \ldots a_n)$ in the instance of each relation~$R$.
A \emph{subinstance} $I'$ of~$I$ is an instance that contains a subset of the
facts of~$I$.
The \emph{active domain} of~$I$, denoted $\adom(I)$, is the set
of all the values that occur in facts of~$I$.

We study \emph{conjunctive queries} (CQs) which are expressions of the form
$\exists x_1 \ldots x_k  ~ (A_1 \wedge \cdots \wedge A_m)$, where
the~$A_i$ are \emph{relational atoms} of the form
$R(x_1 \ldots x_n)$, with~$R$ being a relation of arity~$n$ and $x_1 \ldots x_n$
being variables or constants.
A CQ is \emph{Boolean} if it has no free variables.
A Boolean CQ $Q$ \emph{holds} in an instance $I$ exactly when there
is a \emph{homomorphism} of~$Q$ to~$I$: a mapping~$h$ from the variables and
constants of~$Q$ to~$\adom(I)$ which is the identity on constants and which
ensures that,
for every atom $R(x_1 \ldots x_n)$ in~$Q$, the atom $R(h(x_1) \ldots h(x_n))$ is a fact of~$I$.
We let $Q(I)$ be the \emph{output} of $Q$ on $I$, defined in the usual way: if $Q$ is Boolean,
the output is true if the query holds and false
otherwise.
A \emph{union of conjunctive queries} (UCQ) is a disjunction of CQs.

\vspace{-2pt}
\myparaskip
\myparagraph{Integrity constraints}
To express restrictions on instances, we will use fragments of first-order logic
(FO), with the active-domain semantics, and where we disallow constants.
We will focus on \emph{dependencies}, especially
on \emph{tuple-generating dependencies} (TGDs) and on
\emph{functional dependencies} (FDs).

A \emph{tuple-generating dependency} (TGD) is an FO sentence~$\tau$ of the form:
$\forall \vec x ~ (\phi(\vec x) \rightarrow \exists \vec y ~ \psi(\vec x, \vec
y))$
where $\phi$ and $\psi$ are conjunctions of relational atoms:
$\phi$ is the \emph{body} of~$\tau$ while $\psi$ is the \emph{head}.
For brevity, in the sequel, we will omit outermost universal quantifications in TGDs.
The
\emph{exported variables} of~$\tau$ are the variables of~$\vec x$ which occur in the head.
A \emph{full TGD} is one with no existential quantifiers in the head.
A \emph{guarded TGD} (GTGD) is a TGD where $\phi$ is of the form
$A(\vec x) \wedge \phi'(\vec x)$ where $A$ is a relational atom containing 
all free variables of~$\phi'$. 
An \emph{inclusion dependency} ($\incd$) is a GTGD where both $\phi$ and
$\psi$
consist of a single
atom  with no
repeated variables.
The \emph{width} of an $\incd$ is the number
of exported variables, and an ID is  \emph{unary}
(written $\uincd$) if it has width $1$.
For example,
$R(x, y) \rightarrow \exists z \, w ~ S(z, y, w)$ is a $\uincd$.

A \emph{functional dependency} (FD) is an FO sentence $\phi$ written as
$\forall \vec x \, \vec{x}' ~ (R(x_1 \ldots x_n) \wedge R(x'_1 \ldots x'_n) \wedge  \left(\bigwedge_{i \in D}
x_i=x'_i\right) \rightarrow  x_j=x'_j)$,
with $D\subseteq\{1 \ldots n\}$ and
$j \in \{1 \ldots n\}$,
Intuitively, $\phi$ asserts that
position~$j$ is \emph{determined} by the positions of~$D$, i.e., when two
$R$-facts match on the positions of~$D$, they must match on 
position~$j$ as well. We write $\phi$ as $D \determines j$ for brevity.

\myparaskip
\myparagraph{Query and access model}
We model a collection of Web services as a service schema $\aschema$,
which we simply call a 
\emph{schema}.
It consists of  
\begin{inparaenum}[(1.)]
  \item a relational signature $\sign$;
  \item a set of integrity constraints $\Sigma$ given as
    FO sentences; and 
  \item a set of
    \emph{access methods} (or simply \emph{methods}).
\end{inparaenum}
Each access method $\mt$ is associated with a relation~$R$ and
a subset of positions of~$R$ called 
the \emph{input positions} of~$\mt$.
The other positions of~$R$ are called
\emph{output positions} of~$\mt$.

In this work, we allow each access method to have an optional \emph{result bound}. 
If $\mt$ has a result bound, then $\mt$ is further associated to a positive
integer $k\in\NN$; we call $\mt$ a
\emph{result-bounded method}. Informally, the result bound on~$\mt$
asserts
two things: (i) $\mt$ returns at most $k$ matching tuples; (ii)
if there are no more than $k$ matching tuples, then $\mt$ returns all of them,
otherwise it returns some subset of $k$ matching tuples.
We also allow access methods to have a \emph{result lower bound}, which only
imposes point (ii).

An \emph{access} on an instance $I$ consists of a method $\mt$ on some relation
$R$ and of a \emph{binding} $\accbind$ for~$I$: the binding
is a mapping from the input positions of
$\mt$ to values in~$\adom(I)$.
The \emph{matching tuples}~$M$ of the access $(\mt, \accbind)$ are the tuples
for relation~$R$ in~$I$ 
that match $\accbind$ on the input positions
of~$R$, and an \emph{output} of the access is a subset $J \subseteq M$.
If there is no result bound or result lower bound on~$\mt$, then
there is only one \emph{valid output} to the access, namely, the 
output $J \colonequals M$
that contains all matching
tuples of~$I$.
If there is a  result bound~$k$ on~$\mt$, then a
\emph{valid output} to the access
is any subset $J \subseteq M$ such that:
\mylistskip
\begin{compactenum}[(i)]
\item $J$ has size at most $k$
\item for any $j \leq k$, if $I$ has
$\geq j$ matching  tuples, then
  $J$ has size~$\geq j$. Formally, if $\card{M} \geq j$ then $\card{J} \geq j$.
\end{compactenum}
\mylistskip
If there is a result lower bound of~$k$ on~$\mt$, then a \emph{valid output}
is any subset $J \subseteq M$ satisfying point (ii) above.

We give specific names to two kinds of methods.
First, a method is \emph{input-free} if it has no input positions.
Second, a method is \emph{Boolean} if
all positions are input positions. Note that accessing a Boolean method
with a binding $\accbind$ just checks if $\accbind$
is in the relation associated to the method
(and result bounds have no effect).

\myparaskip
\myparagraph{Plans}
We use \emph{plans} to describe programs that use the access methods, formalizing
them using the terminology of ~\cite{ustods,thebook}. 
A \emph{monotone plan}  $\aplan$
is a 
sequence of \emph{commands} that produce \emph{temporary tables}. There are two
types of commands:

\mylistskip
\begin{compactitem}
\item \emph{Query middleware commands}, of the form $T~\colonequals~E$, with~$T$ a temporary table
  and $E$ a monotone relational algebra expression over the temporary tables
  produced by previous commands. By \emph{monotone}, we
  mean that $E$ does not use the relational
  difference operator; equivalently, it is expressed in monotone first-order
  logic.
\item \emph{Access commands}, of the form $T \Leftarrow_\outmap \mt \Leftarrow_\inmap E$,
  where $E$ is a monotone relational algebra expression over
  previously-produced temporary tables, $\inmap$ is an \emph{input
  mapping} from the output attributes of~$E$ to the input positions of~$\mt$,
  $\mt$ is a \emph{method} on some relation~$R$, $\outmap$ is an \emph{output
  mapping} from the positions of~$R$ to those of~$T$, and $T$ is a temporary
  table.
We often omit the mappings for brevity.
\end{compactitem}
\mylistskip
The \emph{output table} $T_0$ of~$\aplan$ is indicated by a special command
$\return~T_0$ at the end, with $T_0$ being a temporary table.

We must now define the semantics of~$\aplan$ on an instance~$I$.
Because of the non-determinism of result-bounded methods,
we will do so relative to an
\emph{access selection} for~$\aschema$ on~$I$, i.e.,
a function $\aselect$ mapping each access $(\mt, \accbind)$ on~$I$ to a
set of facts $J \colonequals \aselect(\mt, \accbind)$ that match the access.
We say that the access selection is \emph{valid} if it
maps every access to a valid output: intuitively, the access selection
describes which
valid output is chosen when
an access to a result-bounded method matches more tuples than
the bound. Note that the definition implies that performing the same access
twice must return the same result;
however, \emph{all our
results still hold without this assumption} (see Appendix~\ref{app:idempotent}
for details).

For every valid access selection~$\aselect$, we can now define
the semantics of each command of $\aplan$ for~$\aselect$ by considering them
in order.
For an access command $T \Leftarrow_\outmap \mt \Leftarrow_\inmap E$ in~$\aplan$,
we evaluate $E$ to get a collection $C$ of tuples. For each tuple $\vec t$
of~$C$, we use
$\inmap$ to turn it into a binding $\accbind$, and we perform the access
on~$\mt$ to obtain $J_{\vec t} \colonequals \aselect(\mt, \accbind)$. We then
take the union $\bigcup_{\vec t \in C} J_{\vec t}$ of all outputs, rename it
according to~$\outmap$, and write it in~$T$.
For a middleware query command $T \colonequals E$, we evaluate $E$ 
and write the result in~$T$.
The \emph{output} of~$\aplan$ on~$\aselect$ is then the set of tuples that are
written to
the output table~$T_0$.

The \emph{possible outputs} of~$\aplan$ on~$I$ are the outputs that can be
obtained with some valid access selection~$\aselect$. Intuitively, when we evaluate
$\aplan$, we can obtain any of these outputs, depending on which valid access
selection $\aselect$ is used.

\begin{example}
  \label{ex:existencecheckplan}
The plan $\aplan$ of Example~\ref{ex:existencecheck}
is as follows:\\[.3em]
\null\hfill$
T \Leftarrow \udaccess \Leftarrow  \emptyset; \qquad
T_0 \colonequals \pi_\emptyset T;
\qquad
\return~ T_0;
$\hfill\null\\[.2em]
The first command runs the relational algebra expression $E = \emptyset$
returning the empty set, giving a trivial binding for~$\udaccess$. The result of
accessing $\udaccess$ is stored in a temporary table~$T$. The second command 
projects~$T$ to the empty set of attributes, and the third command
  returns the result. For every instance~$I$, 
  the plan $\aplan$ has only one possible output (no matter the access
  selection), describing if $\udirectory$ is empty. We will say that $\aplan$ \emph{answers} the query~$Q_2$ of 
  Example~\ref{ex:existencecheck}.
\end{example}

\myparaskip
\myparagraph{Answerability}
Let $\aschema$ be a schema consisting of a relational signature, integrity
constraints, and access
methods, 
and let $Q$  be a CQ over the relational signature  of~$\aschema$.
A monotone plan $\aplan$ 
\emph{answers}~$Q$ 
under~$\aschema$ if the following holds: for all instances $I$ 
satisfying the constraints, $\aplan$ on~$I$ has exactly one possible output,
which is the query output $Q(I)$. 
In other words, no matter which
valid access selection~$\aselect$ is used to return tuples, the output of~$\aplan$ 
evaluated under~$\aselect$ on~$I$ is equal to~$Q(I)$.
Of course, $\aplan$ can have a single possible output (and answer~$Q$) even if
some intermediate command of~$\aplan$ has multiple possible outputs. 

We say that~$Q$ is
\emph{monotone answerable} under schema $\aschema$ 
if there is a monotone plan that answers it.
Monotone answerability generalizes notions of reformulation that have been
previously studied.
In particular, in the absence of constraints and result bounds, it reduces
to the notion of a query having an \emph{executable rewriting with respect to access methods}, 
studied in 
work on access-restricted querying~\cite{access1,access2}.  In the setting where the limited
interfaces simply expose views, monotone answerability
corresponds to the well-known notion of \emph{UCQ rewriting} with respect to views~\cite{lmss}.

\myparaskip
\myparagraph{Query containment and chase proofs}
We will reduce answerability to
\emph{query containment under constraints}, i.e.,
checking whether a Boolean CQ~$Q'$ follows from another Boolean CQ~$Q$ 
and some constraints~$\Sigma$. Formally, the problem asks if
any instance that satisfies~$Q$ and $\Sigma$ also satisfies $Q'$,
which we denote as~$Q
\subseteq_\Sigma Q'$.
There are well-known reductions between query containment with TGDs and
 the  problem of 
\emph{certain answers}~\cite{fagindataex,taming}
 under TGDs.
We will not
need the definition of certain answers, but we will use some
existing upper and lower bounds from this line of work (e.g., from
\cite{taming,baget2010walking}), rephrased to query containment under
constraints.

When $\Sigma$ consists of dependencies,
query containment under constraints
can be solved
by searching for a 
\emph{chase proof}~\cite{fagindataex}. Such a proof starts with
an instance called the \emph{canonical
database of~$Q$} and denoted $\canondb(Q)$: it consists of facts for each atom
of~$Q$, and its elements are the variables and constants of~$Q$.
The proof then proceeds by \emph{firing dependencies}, as we explain next.

A homomorphism $\trig$ from the body of a dependency $\dep$ into an instance
$I$ is called a \emph{trigger} for~$\dep$.
We say that~$\trig$ is an \emph{active trigger} if
$\trig$ cannot be extended to a homomorphism from the
head of~$\dep$ to~$I$. In other words, an active trigger $\trig$ witnesses the fact
that~$\dep$ does
not hold in~$I$. We can solve this by \emph{firing} the dependency $\dep$ on
the active trigger $\trig$, which we also call performing a \emph{chase step},
in the following way. If $\dep$ is a TGD, the result of the chase step on~$\trig$ 
for~$\dep$ in~$I$ 
is the superinstance
$I'$ of~$I$ obtained by adding new facts
corresponding to an extension of~$\trig$ to the head of~$\dep$, using fresh elements  to instantiate the
existentially quantified variables of the head: we call these elements \emph{nulls}.
If $\dep$ is an FD with~$x_i=x_j$ in the head, then a chase step yields $I'$
which is the result of identifying $\trig(x_i)$ and $\trig(x_j)$ in~$I$.
A  \emph{chase sequence}  is a sequence of chase steps, and it is a \emph{chase proof}
of~$Q \subseteq_\Sigma Q'$ 
if it produces an instance where $Q'$ holds.  

It can be shown~\cite{fagindataex} that whenever
$Q \subseteq_\Sigma Q'$ there is a chase proof that witnesses this.
If all chase sequences are finite we say the \emph{chase with~$\Sigma$ on
$Q$ terminates}. In this case, we can use the chase to decide containment
under constraints.

\myparaskip
\myparagraph{Variations of answerability}
So far, we have defined monotone answerability. An alternative 
notion
is \emph{RA-answerability}, defined using \emph{RA-plans} that
allow arbitrary relational algebra expressions in commands.
In the body of the paper we focus on monotone
answerability, because we think it is the more natural notion 
for CQs and for the class of constraints that we consider.
Indeed, CQs are monotone: if facts are added to an instance, the output of a
CQ cannot decrease.
Thus the bulk of prior work on implementing CQs over restricted limited interfaces, both in theory
\cite{lmss,dln,access1,access2} and in practice \cite{candbsigmod14,candbsigrecord}, has focused on monotone 
implementations.
However, \emph{many of our results extend to answerability with RA-plans} (see
Appendix~\ref{apx:ra}). Indeed, we can sometimes show that monotone answerability
and RA-answerability coincide.

As a second variation, note that we have defined monotone
answerability by requiring that the query and plan agree on all instances,
finite and infinite. An alternative is to consider equivalence over finite instances only. We say that
a plan $\aplan$ \emph{finitely answers} $Q$,
if for any finite instance $I$ satisfying the integrity constraints of~$\aplan$,
the only possible output of~$\aplan$s is~$Q(I)$; the notion of a query being 
\emph{finitely monotone answerable}
is  defined in the obvious way. Both finite and unrestricted answerability
have been studied
in past work on access methods \cite{ustods,thebook}, just as finite and unrestricted variants
of other static analysis problems (e.g.,
query containment)
have long been investigated in database theory  (e.g., \cite{johnsonklug}).  
The unrestricted variants usually
provide a cleaner theory, while the finite
variants can be more precise. In this work our goal is to  investigate both
variants, 
leaving a discussion of  the trade-off between finite and unrestricted answerability for
future work. 
As it turns out,
for the database-style dependencies that we consider, 
the finite variant can be reduced to the unrestricted one.
In particular, this reduction holds
for constraints $\Sigma$ that are \emph{finitely controllable}, by which we mean
that for all Boolean UCQs $Q$ and~$Q'$,
the containment $Q \subseteq_\Sigma Q'$ holds if and only if,
whenever a
finite instance $I$ satisfies $Q$, then it
also satisfies $Q'$. For such constraints~$\Sigma$, there is no distinction between the finite and unrestricted
versions:
\begin{proposition} \label{prop:finitecontrol}
  If $\aschema$ is a schema whose constraints 
are finitely controllable,
then any CQ $Q$ that is finitely monotone answerable with respect to $\aschema$ is monotone
answerable with respect to $\aschema$.
\end{proposition}
\begin{proof}
If $Q$ is finitely monotone answerable there is a monotone plan $\aplan$ that is equivalent to $Q$ over
all finite instances. $\aplan$ can be rewritten as a UCQ. Thus finite controllability 
  implies 
that $\aplan$  is equivalent to $Q$ over all instances, and thus $Q$ is monotone answerable.
\end{proof}

Many of the well-studied classes of dependencies with decidable static analysis problems are finitely controllable.
An exception are dependencies consisting of a mix of $\uincd$s and FDs. However, these are known to be finitely
controllable once certain dependencies are added, and thus the finite controllability
technique can also be applied in this case (see Section~\ref{sec:complexitychoice}).

Finally, for simplicity we also look \emph{only at Boolean CQs from here on}. But our results extend straightforwardly
to the non-Boolean case.

\sectionarxiv{$\!\!\!$Reducing to Query Containment}{Reducing to Query Containment} \label{sec:reduce}

We start our study of the monotone answerability problem by
reducing it to \emph{query
containment under constraints},  defined in the previous section.
We explain in this section how this reduction is done. It 
extends the approach
of~\cite{dln,ustods,thebook} to result bounds, and follows the connection
between answerability and determinacy notions of~\cite{NSV,thebook}.

The 
query containment problem corresponding to monotone answerability will capture
the idea that \emph{if 
an instance $I_1$
satisfies a query $Q$ and another instance $I_2$ has more ``accessible data'' than~$I_1$, then $I_2$
should satisfy\/ $Q$ as well}.
We will first define accessible data via the notion of
\emph{accessible part}. We
use this to 
formalize the previous idea as the property of \emph{access monotonic-determinacy},
and show it to be
equivalent to monotone answerability. 
Using access monotonic-determinacy we
show that we can simplify the result bounds of arbitrary schemas, and
restrict to \emph{result lower bounds} throughout this work.
Last, we close the section by showing how to rephrase
access monotonic-determinacy with result lower bounds to
query containment under constraints.

\myparaskip
\myparagraph{Accessible parts}
We first formalize the notion of ``accessible
data''. 
Given a schema $\aschema$ with result-bounded methods and an instance~$I$, an \emph{accessible part} of $I$
is any subinstance obtained by
iteratively making accesses until we
reach a fixpoint.
Formally, we define an accessible part
by choosing a valid access selection $\aselect$ and
 inductively
defining sets of facts $\accpart_i(\aselect, I)$
 and sets of values $\accessible_i(\aselect, I)$ by:
\begin{align*}
  \accpart_0(\aselect, I) & \colonequals \emptyset \mathrm{~~and~~}
  \accessible_0(\aselect,I) \colonequals \emptyset\\
  \accpart_{i+1}(\aselect, I)& \colonequals \!\!\!\!\!\!\!\!\!\!\!\!\!\!\!\!\!\!\!\!\!\!\!\!\!\!\!\!\!\!\!\!\!\!\!\bigcup_{\substack{\mt \mbox{ method},\\[.2em] \accbind
  \mbox{ binding with values in~} \accessible_i(\aselect,I)}} \!\!\!\!\!\!\!\!\!\!\!\!\!\!\!\!\!\!\!\!\!\!\!\!\!\!\!\!\!\!\!\!\!\!\!\!\!\!\!\!\aselect(\mt, \accbind)\qquad \\
   \accessible_{i+1}(\aselect,I) & \colonequals  \adom(\accpart_{i+1}(\aselect, I))
\end{align*}
Above we abuse notation by considering $\aselect(\mt, \accbind)$ as a set of facts, rather
than a set of tuples. These equations define by mutual induction the
set of values
($\accessible$)
that we can retrieve by iterating accesses 
and the set of facts
($\accpart$)
that we can retrieve using those values.

The \emph{accessible part} under~$\sigma$, written $\accpart(\aselect,I)$,
is then defined as~$\bigcup_i \accpart_i(\aselect,I)$.
As the equations are monotone, this fixpoint is reached
after finitely many iterations if $I$ is finite,
or as the union of all finite iterations if $I$ is infinite.
When there are no result bounds, there is only one valid access selection~$\aselect$, so
only one accessible part: it intuitively
corresponds to the data that can be accessed using the methods.
In the presence of result bounds, 
there can be many 
accessible parts, depending on~$\aselect$.

\myparaskip
\myparagraph{Access monotonic-determinacy}
We now formalize the idea that a query~$Q$ is ``monotone under accessible
parts''.
Let $\Sigma$ be the integrity constraints of~$\aschema$.
We call~$Q$ \emph{access monotonically-determined} in~$\aschema$
(or $\amd$, for short),
 if for any two instances $I_1$, $I_2$ satisfying~$\Sigma$,
 if there is an accessible part of $I_1$
 that is a subset of an accessible part of  $I_2$, then
$Q(I_1) \subseteq Q(I_2)$.
Note that when there are no result bounds, there is a unique accessible part of $I_1$ and of $I_2$, and
$\amd$ says that when the accessible part grows, then $Q$ grows.
The definition of~$\amd$ is justified by the following result:

\newcommand{\thmequiv}{
$Q$ is monotone answerable \wrt\ $\aschema$ if and only if
$Q$ is $\amd$ over~$\aschema$.
}
\begin{theorem}
  \label{thm:equiv}
  \thmequiv
\end{theorem}
Without result bounds, this equivalence of monotone answerability and
access monotone determinacy is proven in \cite{ustods,thebook}, using a variant of Craig's interpolation theorem.
Theorem~\ref{thm:equiv} shows that the equivalence extends to schemas with
result bounds (see Appendix~\ref{app:equiv} for the proof).

In the sequel, it will be more convenient to use an alternative definition
of~$\amd$, based on the notion of \emph{access-valid} subinstances.
A subinstance $I_\acc$ of~$I_1$ is \emph{access-valid in~$I_1$} for~$\aschema$ if,
for any access $(\mt, \accbind)$ performed with a method $\mt$ of~$\aschema$
and with a binding $\accbind$ whose values are in~$I_\acc$,
there is a set $J$ of matching
tuples in $I_\acc$ such that $J$ is a valid output to the access $(\mt,
\accbind)$ in~$I_1$.
In other words, for any access performed on $I_\acc$, we can choose an output
in~$I_\acc$ which is also valid in~$I_1$. We can use this notion to rephrase the
definition of~$\amd$ to talk about 
a common subinstance of $I_1$ and $I_2$ that is access-valid:

\newcommand{\propaltdef}{
For any schema $\aschema$ with constraints~$\Sigma$  and result-bounded methods,
a CQ $Q$ is $\amd$
 if and only if the following implication holds: for any two instances $I_1$, $I_2$ satisfying $\Sigma$, if $I_1$
 and $I_2$ have a common subinstance $I_\acc$ that
is access-valid in~$I_1$, then
$Q(I_1) \subseteq Q(I_2)$.
}

\begin{proposition}  \label{prop:altdef}
  \propaltdef
\end{proposition}

The proof, given in Appendix \ref{app:determinacy}, follows from the definitions.
The alternative definition of~$\amd$ is more convenient, because it only deals
with a subinstance of~$I_1$ and not with accessible parts. Thus, 
\emph{we will use this characterization of monotone answerability in  the rest of
this paper}.

\myparaskip
\myparagraph{Elimination of result upper bounds}
The characterization of monotone answerability in terms of~$\amd$ allows
us to prove a key simplification in the analysis of result bounds.
Recall that a result bound of~$k$
declares both an \emph{upper bound} of~$k$ on the number of returned results, and
a \emph{lower bound} on them: for all $j \leq k$,
if there are $j$ matches, then
$j$ must be returned. 
We can show that the upper bound makes no difference for monotone answerability.
Formally, for a schema $\aschema$  with
integrity constraints and access methods, some of which
may be result-bounded, we define
the schema $\relaxs(\aschema)$. It has the same vocabulary, constraints,
and access methods as in~$\aschema$.  For each access method
$\mt$ in~$\aschema$ with
result bound of~$k$, $\mt$ has instead a \emph{result lower bound} of~$k$ in
$\relaxs(\aschema)$, i.e., $\mt$ does not impose the upper bound.
We can then show:

\newcommand{\elimupper}{
 Let $\aschema$ be a schema with
arbitrary constraints and with access methods which may be result-bounded.
A CQ~$Q$ is monotone answerable in~$\aschema$ if and only if it is monotone answerable
in~$\relaxs(\aschema)$.
}
\begin{proposition} \label{prop:elimupper}
  \elimupper
\end{proposition}

\begin{proof}
  We show the result for $\amd$ instead of monotone answerability,
  thanks to Theorem~\ref{thm:equiv}, and use Proposition~\ref{prop:altdef}.
  Consider arbitrary instances $I_1$ and
  $I_2$ that satisfy the constraints, and let us show that any common subinstance
  $I_\acc$ of~$I_1$ and $I_2$ is access-valid in~$I_1$ for~$\aschema$ iff
  it is access-valid in~$I_1$ for~$\relaxs(\aschema)$: this implies the claimed
  result.

  In the forward direction, if $I_\acc$ is access-valid in~$I_1$ for
  $\aschema$, then clearly it is access-valid in~$I_1$
  for~$\relaxs(\aschema)$, as any output of an access on~$I_\acc$ which is
  valid in~$I_1$ for~$\aschema$ is also valid for~$\relaxs(\aschema)$.

  In the backward direction, assume  $I_\acc$ is access-valid in~$I_1$
  for~$\relaxs(\aschema)$, and consider an access $(\mt, \allowbreak\accbind)$
  with values of~$I_\acc$.
  If $\mt$ has no result lower bound, then there is only one possible output for the
  access, and it is also valid for~$\aschema$.
  Likewise, if $\mt$ has a result lower bound of~$k$ and there are $\leq k$ matching
  tuples for the access, then the definition of a result lower bound ensures
  that there is only one possible output, which is again valid for~$\aschema$.
  Last, if there are $>k$ matching tuples for the access,
  we let $J$ be a set of tuples in~$I_\acc$ which is is a valid output
  to the access in~$\relax(\aschema)$, and take any subset $J'$ of~$J$ with
  $k$ tuples; it is clearly a valid output to the access for~$\aschema$. This
  establishes the backward direction, concluding the proof.
\end{proof}

Thanks to this, in our study of monotone answerability in the rest of the paper, 
\emph{we only consider result lower bounds}.

\myparaskip
\myparagraph{Reducing to query containment}
Now that we have reduced our monotone answerability problem to~$\amd$, and
eliminated result upper bounds, we explain how to
restate $\amd$ as a query containment
problem.
To do so, we will expand the relational signature:
we let $\accessible$ be a new unary predicate, and for
each relation~$R$ of the original signature,
we introduce two copies $R_\acc$ and $R'$ with the same
arity as~$R$.
Letting $\Sigma$  be the integrity constraints
in the original schema, we let $\Sigma'$ 
be formed
by replacing every relation~$R$ with~$R'$. For any CQ $Q$, we define
$Q'$ from~$Q$ in the same way.
The \emph{$\amd$ containment for~$Q$ and $\aschema$} is
then the CQ containment $Q \subseteq_\Gamma Q'$ where
the constraints $\Gamma$ are defined as follows: they include the original constraints~$\Sigma$, the
constraints $\Sigma'$ on the relations $R'$, and the following
\emph{accessibility axioms} (with implicit universal quantification):
\mylistskip
\begin{compactitem}
\item For each method $\mt$ that is not result-bounded, letting $R$ be
  the relation accessed by~$\mt$:
  \[\Big(\bigwedge_i \accessible(x_i)\Big) \wedge
  R(\vec x, \vec y) \rightarrow R_\acc(\vec x, \vec y)\]
where $\vec x$ denotes the input positions of~$\mt$ in~$R$.
\item For each method $\mt$ with a result lower bound of~$k$, 
  letting $R$ be the relation accessed by~$\mt$, 
  for all $j \leq k$:
  \[\Big(\bigwedge_i \accessible(x_i)\Big) \wedge \exists^{\geq j} \vec y ~ 
  R(\vec x, \vec y)
  \rightarrow \exists^{\geq j} \vec z ~ R_\acc(\vec x, \vec z)\]
where $\vec x$ denotes the input positions of~$\mt$ in~$R$.
Note that we write
  $\exists^{\geq j} \vec y ~ \phi(\vec x, \vec y)$ for a subformula $\phi$
  to mean that there exist at least $j$ different values of~$\vec y$
  such that $\phi(\vec x, \vec y)$ holds.
\item For every relation $R$ of the original signature:
  \[R_\acc(\vec w) \rightarrow R(\vec w) \wedge R'(\vec w) \wedge \bigwedge_i
  \accessible(w_i)\]
\end{compactitem}
\mylistskip
The $\amd$ containment above simply formalizes the definition of~$\amd$, via
Proposition~\ref{prop:altdef}.
Intuitively, $R$ and $R'$ represent the interpretations of the relation~$R$ in~$I_1$ and $I_2$;
$R_\acc$ represents the interpretation of~$R$ in~$I_\acc$; and $\accessible$ represents the active  domain
of~$I_\acc$.
The constraints $\Gamma$ include $\Sigma$ and $\Sigma'$, which means that~$I_1$ and $I_2$ both satisfy $\Sigma$.
The first two accessibility axioms enforce that $I_\acc$ is access-valid in~$I_1$: for non-result-bounded
methods, accesses to a method $\mt$ on a relation~$R$ return all the results, while for result-bounded methods it respects the lower bounds.
The last accessibility axiom enforces that~$I_\acc$ is a common subinstance
of~$I_1$ and $I_2$ and
that~$\accessible$ includes the active domain of~$I_\acc$.
Hence, from the definitions, we have:

\begin{proposition} \label{prop:reduce} $Q$ is monotone answerable with respect
  to a schema $\aschema$
iff the $\amd$ containment for~$Q$ and $\aschema$ holds.
\end{proposition}
Note that, for a schema without result bounds,  the accessibility axioms above
can be rewritten
as follows (as in~\cite{ustods,thebook}): for each method $\mt$, letting $R$
be the relation accessed by~$\mt$ and $\vec x$ be the input positions
of~$\mt$ in~$R$, we have the axiom:
\[
\Big(\bigwedge_i \accessible(x_i)\Big) \wedge
  R(\vec x, \vec y) \rightarrow  R'(\vec x, \vec y) \wedge \bigwedge_i
  \accessible(y_i)\]

\begin{example} \label{ex:reduce} Let us apply the reduction above to the schema of Example~\ref{ex:simple}
  with the result bound of 100 from Example~\ref{exa:rbound}.
We  see that monotone answerability of a CQ $Q$  is equivalent to $Q
  \subseteq_\Gamma Q'$, for~$\Gamma$ containing:
\begin{compactitem}
\item  the referential constraint from $\univdirect$ into $\profinfo$  and from
$\univdirect'$ into $\profinfo'$

\item  $\accessible(i) \wedge \profinfo(i, n, s) \rightarrow
  \profinfo_\acc(i, n, s)$,
\item for all $1 \leq j \leq 100$:\\
$\exists \vec y_1 \cdots \vec y_j
    (\bigwedge_{1 \leq p < q \leq j} \vec y_p \neq \vec y_q \wedge
    \univdirect(\vec y_p))$ \\
    $\rightarrow \exists \vec y_1' \cdots \vec y_j'
    (\bigwedge_{1 \leq p < q \leq j} ~  \vec y_p' \neq \vec y_q' \wedge
    \univdirect_\acc(\vec y_p'))$
\item 
$\profinfo\!_\acc(\vec w) \!\rightarrow\! \profinfo(\vec w) \wedge \profinfo'(\vec w) \wedge 
  \!\bigwedge_i \accessible(w_i)$
and similarly for $\univdirect$.
\end{compactitem}
Note that the  constraint in the third item is  quite complex; it contains
  inequalities and also disjunction, since we write
$\vec y \neq \vec z$ to abbreviate a disjunction $\bigvee_{i \leq |\vec y|} y_i
  \neq z_i$. This makes it challenging to decide if $Q
  \subseteq_\Gamma Q'$ holds.
Hence, our goal in the next section will be to simplify result bounds to avoid
  such complex constraints.
\end{example}

\begin{table*}
  \centering
  {
\caption{Summary of results on simplifiability and complexity of monotone answerability}
\label{tab:results}
  \begin{tabular}{lll}
\toprule
{\bfseries Fragment} &
{\bfseries Simplification} & 
    {\bfseries Complexity } \\
\midrule
  IDs & Existence-check  (Theorem~\ref{thm:simplifyidsexistence}) & $\exptime$-complete  (Theorem~\ref{thm:decidids})\\
  Bounded-width IDs & Existence-check (see above) & $\np$-complete  (Theorem~\ref{thm:npidsbounds}) \\
  FDs & FD  (Theorem~\ref{thm:fdsimplify})
    & $\np$-complete  (Theorem~\ref{thm:decidfd}) \\
  FDs and UIDs & Choice  (Theorem~\ref{thm:simplifychoiceuidfd})
    &  $\np$-hard (see above) and in~$\exptime$  (Theorem~\ref{thm:deciduidfd}) \\
  Equality-free FO & Choice  (Theorem~\ref{thm:simplifychoice})
    & Undecidable  (Proposition~\ref{prp:undec}) \\
  Frontier-guarded TGDs & Choice  (see above)
    & $\twoexp$-complete  (Theorem~\ref{thm:decidegf}) \\
\bottomrule
\end{tabular}
  }
\end{table*}

\sectionarxiv{Simplifying result bounds}{Simplifying result bounds}
\label{sec:simplify}
The results in Section~\ref{sec:reduce} allow us to reduce the monotone answerability problem to
a query containment problem. However, for result bounds greater than~$1$, the
containment problem involves complex cardinality constraints, as illustrated in Example \ref{ex:reduce},
and thus we cannot apply standard results or algorithms on query containment under
 constraints to get decidability ``out of the box''.
To address this difficulty, we must \emph{simplify} result-bounded schemas,
i.e., change or remove the result bounds.
We do so in this section, with \emph{simplification} results of the following
form: if we can find a plan for a query on a result-bounded schema,
then we can find a plan in a \emph{simplification} of the schema, i.e.,
a schema with simpler result bounds or
no result bounds at all.

These simplification results have two benefits.
First, they give insight about the use of result bounds, following the examples
in the introduction. For instance, our
results will
show that for most common classes of  constraints used in databases, the actual
numbers in the result bounds never matter for answerability.
Secondly, they help us to show the decidability of monotone answerability.

\myparaskip
\myparagraph{Existence-check simplification} \label{subsec:exist}
The simplest way to use result-bounded methods is to 
check if some tuples exist, as in Example~\ref{ex:existencecheck}.
We will formalize this as the \emph{existence-check simplification}, where we
replace result-bounded methods by Boolean methods that can only do such
existence checks.

Given a schema $\aschema$ with result-bounded methods,
its \emph{existence-check simplification}
$\aschema'$ is formed as follows:
\mylistskip
\begin{compactitem}
\item The signature of~$\aschema'$ is that of~$\aschema$ plus some new relations:
  for each result-bounded method $\mt$, letting $R$ be the relation accessed
  by~$\mt$, we add 
a relation~$\checkview_\mt$ whose arity is the number of input positions
  of~$\mt$.
\item The integrity constraints of~$\aschema'$ are those of~$\aschema$ plus, for
  each result-bounded method $\mt$ of~$\aschema$,
  a new constraint (expressible as two~$\incd$s):
$
\checkview_\mt(\vec x) ~ \leftrightarrow ~ \exists
\vec y ~ R(\vec x, \vec y)
$,\\
where $\vec x$  denotes the input positions of~$\mt$ in~$R$.
\item The methods of~$\aschema'$ are the  methods of~$\aschema$ that have
  no result bounds, plus one new Boolean method $\mt'$ on each new relation
  $\checkview_\mt$, that has no result bounds either.
\end{compactitem}
\begin{example}
\label{ex:existssimplify}
Recall the schema $\aschema$ of Example~\ref{ex:simple} having the method
  $\praccess$ and the result-bounded method $\udaccess_2$ of
  Example~\ref{ex:fd}.
  The existence-check simplification of~$\aschema$ 
  has a signature with relations $\udirectory$, $\profinfo$, and a new relation
  $\udirectory_{\udaccess_2}$ of arity~$1$.
It has two access methods without result bounds: the method $\praccess$
  on~$\profinfo$ like in~$\aschema$, and a Boolean method $\udaccess_2'$ on
  $\udirectory_{\udaccess_2}$. Its constraints are those of~$\aschema$, plus the
  following IDs:\mylistskip\begin{compactitem}
  \item $\udirectory(i, a, p)
  \rightarrow \udirectory_{\udaccess_2}(i)$; and
\item $\udirectory_{\udaccess_2}(i) \rightarrow \exists a\,p ~ \udirectory(i, a, p)$.
  \end{compactitem}
\end{example}

Clearly, every plan that uses the existence-check simplification
$\aschema'$ of a schema
$\aschema$ can be converted into a plan using $\aschema$, by simply replacing the
accesses on the Boolean method of~$\checkview_\mt$ to non-deterministic accesses
with~$\mt$, and only checking whether the result of these accesses is empty.
We want to understand when the converse is true. That is, when a plan on~$\aschema$ can
be converted to a plan on~$\aschema'$. For instance,
recalling the plan of Example~\ref{ex:existencecheck} that tests
whether~$\udirectory$ is empty, we could implement it in the existence-check
simplification of this schema.
More generally, we want to identify schemas $\aschema$ for which \emph{any} 
CQ having a monotone plan over~$\aschema$ has a plan 
on the existence-check
simplification~$\aschema'$.
We say that 
$\aschema$
is \emph{existence-check simplifiable} when this holds:
this intuitively means that 
``result bounded methods of~$\aschema$ are only useful for existence checks''.

\myparaskip
\myparagraph{Showing existence-check simplifiability}
We first show that this notion of existence-check simplifiability holds for schemas like Example~\ref{exa:plani} whose constraints consist of
inclusion dependencies:

\newcommand{\thmsimplifyidsexistence}{
  Let $\aschema$ be a schema whose constraints are
$\incd$s, and let $Q$ be a CQ that is
monotone answerable 
  in~$\aschema$.
Then $Q$ is monotone answerable in the existence-check simplification
   of~$\aschema$.
}
\begin{theorem} \label{thm:simplifyidsexistence}
  \thmsimplifyidsexistence
\end{theorem}

This existence-check simplifiability result implies in particular that 
\emph{for schemas with~$\incd$s, monotone answerability  is decidable even
with result bounds}. This is because the existence-check simplification of the
schema features only $\incd$s and no result bounds,
so the query containment problem for~$\amd$ 
only features guarded TGDs, which implies decidability. We will show a finer
complexity bound in the next section.

To prove Theorem~\ref{thm:simplifyidsexistence}, we show that if $Q$ is not $\amd$ 
in the existence-check simplification~$\aschema'$ of~$\aschema$, then it cannot be $\amd$ 
in~$\aschema$. This
suffices to prove the contrapositive of the result, because 
$\amd$ is equivalent
to monotone answerability  (Theorem~\ref{thm:equiv}).
This claim is shown with a general method of
\emph{blowing up models} that we will reuse in all subsequent simplifiability
results. We assume that~$\amd$ does not hold in the
simplification~$\aschema'$, and consider
a \emph{counterexample to~$\amd$} for~$\aschema'$:
two instances $I_1,I_2$ both satisfying the schema constraints, 
such that~$I_1$ satisfies $Q$ while $I_2$ satisfies $\neg Q$, and
$I_1$ and $I_2$ have a
common subinstance $I_\acc$ which is access-valid in~$I_1$.
We use them to build a counterexample to~$\amd$
for the original schema~$\aschema$: we will always do so by adding more facts
to~$I_1$ and~$I_2$
and then restricting to the relations of~$\aschema$.
We formalize this method in the following immediate lemma:

\begin{lemma}
  \label{lem:enlarge}
  Let $\aschema$ and $\aschema'$ be 
  schemas and $Q$ a CQ on the common relations of~$\aschema$ and~$\aschema'$
  such that~$Q$ is not $\amd$ in~$\aschema'$.
  Suppose that for some counterexample $I_1, I_2$ to~$\amd$ for~$Q$ in~$\aschema'$
  we can build instances $I_1^+$ and $I_2^+$ that satisfy the constraints
  of~$\aschema$ with a common subinstance
  $I_\acc$ that is access-valid in~$I_1^+$ for~$\aschema$, and such that for
  each $p \in \{1, 2\}$, the instance $I_p^+$ has a homomorphism to~$I_p$, and
  the restriction of $I_p$ to the relations of~$\aschema$ is a subinstance
  of~$I_p^+$.
Then $Q$ is not $\amd$ in~$\aschema$.
\end{lemma}

Let us sketch how the blowing-up process of the lemma is used to prove
our existence-check simplification result:

\begin{proof}[\myproof sketch for Theorem~\ref{thm:simplifyidsexistence}]
Assume we have  a counterexample $I_1, I_2$ to
$\amd$ for~$Q$ in the simplification~$\aschema'$.
We will ``blow up'' $I_1$ and $I_2$
  to~$I_1^+$ and $I_2^+$ as explained in Lemma~\ref{lem:enlarge},
  ensuring that $I_1^+$ and $I_2^+$ have a common subinstance $I^{+}_\acc$ that
  is access-valid in~$I_1^+$
  for the original schema $\aschema$. For this, we must ensure that each access
in~$I^{+}_\acc$ 
to a result-bounded method 
returns either no tuples or more tuples than the bound.

Intuitively, we form $I^{+}_\acc$ in two steps.
First, we
consider all $\incd$s of the form $\checkview_\mt(\vec x) \rightarrow \exists \vec y
~ R(\vec x, \vec y)$ in~$\aschema'$, and we chase them ``obliviously''; i.e., for every method
$\mt$ and value for~$\vec x$, we create infinitely many facts to instantiate the
head, with infinitely many nulls for~$\vec y$. We even do this when the trigger
is not active, i.e., when witnesses for the head already exist.
    Let $I_\acc^*$ be the result of this.

    In a second step, we solve the constraint violations that may have been
    added by creating these new facts. We do so by
  applying the chase to~$I_\acc^*$ 
  in the usual way
  with all $\incd$
  constraints of~$\Sigma$. This yields $I_\acc^{+}$, from which we remove all
  relations not in~$\aschema$, i.e., all $R_\mt$ facts.

We form $I_1^+$ by unioning $I^{+}_\acc$ with~$I_1$ and restricting again to the
  relations of $\aschema$.
We similarly form $I_2^+$ from ~$I^{+}_\acc$ and $I_2$. As the
constraints are $\incd$s, we can argue that~$I_1^+$ and $I_2^+$ satisfy
$\Sigma$, because $I_1$, $I_2$, and $I^{+}_\acc$ do.
We can also construct homomorphisms of~$I_1^+$ back to~$I_1$ and $I_2^+$ back to~$I_2$,
and we can use $I^{+}_\acc$ as the
common access-valid subinstance. This proves
Theorem~\ref{thm:simplifyidsexistence}.
\end{proof}

\newcommand{\simplifyfddef}{
\item The signature of~$\aschema'$ is that of~$\aschema$ plus some new
relations: for each result-bounded
method $\mt$, letting $R$ be the relation accessed by~$\mt$,
we add a relation~$R_\mt$ whose arity is
$\left|\detby(\mt)\right|$.
\item The integrity constraints of~$\aschema'$ are those of~$\aschema$ plus, 
for each result-bounded method~$\mt$ of~$\aschema$,
a new 
constraint (expressible as two~$\incd$s):
 $R_\mt(\vec x, \vec y) \!\leftrightarrow\! \exists \vec z ~ R(\vec x, \vec y, \vec z)$,\\
where $\vec x$ denotes the input positions of~$\mt$ and $\vec y$ denotes the
other positions of~$\detby(\mt)$.
\item The methods of~$\aschema'$ are the methods of~$\aschema$ that have
no result bounds, 
 plus the following: for each result-bounded method
$\mt$ on relation~$R$ in~$\aschema$,
a method $\mt'$ on~$R_\mt$ that has no result bounds and whose input positions
are the positions of~$R_\mt$ corresponding to input positions of~$\mt$.
}

\myparaskip
\myparagraph{FD simplification} 
When our constraints include functional dependencies, we can hope for another kind of simplification,
generalizing
the idea of Example~\ref{ex:fd}:
an FD can force the output of a result-bounded method to be deterministic on a
projection of the output positions.
We will define the \emph{FD simplification} to formalize this intuition.

Given a set of constraints $\Sigma$, a relation~$R$ that occurs in~$\Sigma$, and a
subset $P$
of the positions of
$R$, we write $\detby(R,P)$ for the set of positions \emph{determined} by $P$, i.e.,
the set of positions~$i$ of~$R$ such that~$\Sigma$ implies the FD
$P \determines i$. In particular, we have $P \subseteq \detby(R,P)$.
For any access method $\mt$, letting $R$ be the relation that it accesses,
we let $\detby(\mt)$ denote $\detby(R,P)$ where $P$ is the set of input
positions of~$\mt$. 
Given a schema $\aschema$ with result-bounded methods, we can now define its \emph{FD simplification} 
$\aschema'$ as follows:
\mylistskip
\begin{compactitem}
  \simplifyfddef
\end{compactitem}
\mylistskip
Note that the FD simplification is the same as the existence check
simplification when the integrity constraints~$\Sigma$ do not imply any FD.
Further observe that, even though the methods of~$\aschema'$ have no result
bounds, any access to a new method $\mt'$ of~$\aschema'$
is guaranteed to return at most one result. This is thanks to the FD on the
corresponding relation~$R$, and thanks to the constraints that relate
$R_\mt$ and~$R$.

\begin{example} \label{ex:fdsimplify}
Recall the schema $\aschema$ of Example~\ref{ex:fd} and the FD $\phi$
  on~$\univdirect$. In the FD simplification of~$\aschema$, we
  add a relation $\univdirect_{\udaccess2}(\attrfmt{id}, \attrfmt{address})$,
  we replace $\udaccess_2$ by a method $\udaccess_2$ on~$\univdirect_{\udaccess2}$ whose
  input attribute is~$\attrfmt{id}$, and we add the $\incd$s 
  $\univdirect(i, a, p) \rightarrow \univdirect_{\udaccess2}(i, a)$ and
  $\univdirect_{\udaccess2}(i, a) \rightarrow \exists p ~ \univdirect(i, a, p)$.
  The method $\udaccess_2'$ has no result bound, but the $\incd$s above and the
  FD $\phi$ 
  ensure that it always returns at most one result.

  Since the FD simplification has no result-bounded methods, the query containment problem for the simplification
will not  
  use any  complex cardinality constraints, in contrast to Example \ref{ex:reduce}.
\end{example}

A schema $\aschema$ is \emph{FD simplifiable} if every CQ having a monotone plan
over~$\aschema$
has one over the FD simplification of~$\aschema$. As for existence-check,
if a schema is FD simplifiable, 
 we can decide  monotone answerability by reducing to the same  problem in 
a schema
without result bounds.

We use a variant of our ``blowing-up process'' to show that all schemas 
with only FD constraints are FD simplifiable:

\newcommand{\fdsimplify}{
  Let $\aschema$ be a schema whose constraints are FDs, and let $Q$ be a CQ
  that is
   monotone answerable in~$\aschema$. Then $Q$ is monotone answerable in the FD simplification of~$\aschema$.
}
\begin{theorem} \label{thm:fdsimplify}
  \fdsimplify
\end{theorem}

\begin{proof}[\myproof sketch]
  We start by considering a counterexample $I_1,I_2$ to~$\amd$ for the FD simplification
  of~$\aschema$, i.e., $Q$ holds in~$I_1$ but not in~$I_2$, and $I_1$ and $I_2$
  have a common subinstance $I_\acc$ which is access-valid in~$I_1$. We blow up
  the accesses on~$I_1$ one after the other, to enlarge $I_1$ and~$I_2$ to
  $I_1^+$ and~$I_2^+$ satisfying the requirements of Lemma~\ref{lem:enlarge}.

  We blow up each access by adding tuples to~$I_1$ and~$I_2$, to ensure that the
  access has enough common matching tuples in~$I_1$ and~$I_2$ to define
  a valid output. It
  suffices to do this for accesses with result-bounded methods~$\mt$,in the
case where
  some matching tuples in~$I_1$ are not in~$I_2$. In this case, the definition
  of $\detby(\mt)$ ensures that matching tuples in~$I_1$ and~$I_2$ agree on
  positions of~$\detby(\mt)$. But the assumption that there are some
matching tuples in $I_1$ that are not in $I_2$ implies that not all positions
  are  in $\detby(\mt)$. Thus, we can add enough tuples to~$I_1$ and~$I_2$ by defining
  them on~$\detby(\mt)$ like the existing tuples, and putting fresh values in the
  other positions. This can be shown to satisfy the FD constraints of~$\aschema$.
  By performing this blow-up process on each access, we obtain
  $I_1^+$ and~$I_2^+$ with the access-valid subinstance~$I^+_\acc$, we restrict
  to the relations of~$\aschema$, and we
  conclude using Lemma~\ref{lem:enlarge}.
\end{proof}

\section{Decidability of monotone\\answerability}
\label{sec:complexity}

Thus far we have seen a general way to  reduce monotone answerability problems with result bounds to query containment
problems (Section~\ref{sec:reduce}). We have also seen schema  simplification
results for both FDs and $\incd$s, which give us insight
into how result-bounded methods can be used (Section~\ref{sec:simplify}).
We now show that for these two classes of constraints,
the reduction to containment and simplification results combine to
give decidability results, along with tight complexity bounds.
Note that both of these classes are well-known to be finitely
controllable~\cite{cosm,rosati};
hence, thanks to Proposition \ref{prop:finitecontrol}, 
\emph{all bounds on monotone answerability in this section
also apply to finite monotone answerability}.

\myparaskip
\myparagraph{Decidability for FDs}
We first consider schemas in which all the integrity constraints are FDs. 
We start with an analysis of monotone answerability in the case \emph{without result
bounds}:

\newcommand{\fdclassic}{
We can decide whether a CQ is monotone answerable with respect to a
  schema without result bounds whose constraints
  are FDs. The problem is $\np$-complete.
}
\begin{proposition} \label{prop:decidmdetfdclassic}
\fdclassic
\end{proposition}
\begin{proof}[\myproof sketch]
  The lower bound already holds
  without result bounds or constraints \cite{access2}, so we show the upper
  bound.
  By  Theorem \ref{thm:equiv}
  and Proposition~\ref{prop:reduce},
the problem reduces to the $\amd$ query
  containment problem $Q \subseteq_\Gamma Q'$ for~$\aschema$. As $\aschema$ has
  no result bounds, we can define $\Gamma$ using the rewriting of the accessibility axioms given
  after Proposition~\ref{prop:reduce}. This ensures that $\Gamma$
  only contains FDs and full TGDs from~$R$ and~$\accessible$
  to~$R'$ and~$\accessible$. 
  We can then show that the chase with~$\Gamma$ terminates in polynomially many
  rounds. Hence, we can decide
  containment by checking in~$\np$ if $Q'$ holds on the chase result,
  which concludes.
\end{proof}

We now return to the situation \emph{with result bounds}. We know that
schemas with FDs are FD simplifiable.
 From this we get a reduction to query containment with no result bounds, but introducing
new axioms.
We can show that the additional axioms involving $R_{\mt}$ and $R$ do not harm chase
termination,
so that~$\amd$ 
is decidable; in fact, it is NP-complete, i.e., no harder than CQ evaluation:

\newcommand{\thmdecidfd}{
  We can decide whether a CQ is monotone answerable with respect to a schema with
  result bounds whose constraints are FDs. The problem is $\np$-complete.
}
\begin{theorem} \label{thm:decidfd} 
  \thmdecidfd
\end{theorem}

\myparaskip
\myparagraph{Decidability for $\incd$s}
Second, we consider schemas whose constraints consist of~$\incd$s.
As we already mentioned, Theorem~\ref{thm:simplifyidsexistence} implies
decidability for such schemas. We now give the precise complexity bound:

\begin{theorem} \label{thm:decidids}
  We can decide whether
a CQ is monotone answerable
with respect to  a schema with result bounds whose  constraints are
$\incd$s. Further, the problem is $\exptime$-complete.
\end{theorem}
\begin{proof}
  Hardness already holds without result bounds~\cite{bbbicdt}, so we focus on
  the upper bound.
  By Theorem~\ref{thm:simplifyidsexistence}, we can equivalently replace
  the schema $\aschema$ with
  its existence-check simplification $\aschema'$, and $\aschema'$ does not
  have result bounds. Further, we can check that $\aschema'$ consists only of
  $\incd$s, namely, those of~$\aschema$ plus the $\incd$s added in the simplification. 
Note that the resulting query containment problem only involves guarded TGDs, and thus we can conclude
that the problem is in $\twoexp$ from~\cite{datalogpm}. However, we can do better:
\cite{bbbicdt} showed that the monotone answerability problem for schemas where
the constraints are $\incd$s is in $\exptime$, and thus we conclude the proof.
 Note that the result in \cite{bbbicdt} is based on a finer analysis of
   the query containment  problems associated with answerability. We will refine
this analysis in obtaining bounds for more restrictive classes of constraints.
\end{proof}

\myparaskip
\myparagraph{Complexity for bounded-width $\incd$s}
An important practical case for $\incd$s are those whose \emph{width} --- the number of
exported variables --- is bounded by a constant. This 
includes  $\uincd$s. For bounded-width $\incd$s, it was shown by Johnson and Klug~\cite{johnsonklug}
that query containment under constraints is $\np$-complete.
A natural question is whether the same holds for monotone
answerability.
We accordingly conclude the section by showing
the following, which is new even in the setting without result bounds:

\newcommand{\npidsbounds}{
  It is $\np$-complete to decide whether a CQ is monotone answerable with respect to a schema
  with result bounds whose constraints are bounded-width $\incd$s.
}
\begin{theorem} \label{thm:npidsbounds}
  \npidsbounds
\end{theorem}

To show this result, we will again use the fact that $\incd$ constraints are
existence-check simplifiable (Theorem~\ref{thm:simplifyidsexistence}).
Using  Proposition \ref{prop:reduce} we 
reduce to a query containment problem with guarded TGDs. But this is not enough
to get an $\np$ bound.
The reason is that the query containment problem includes accessibility axioms, which
are not~$\incd$s.
So we cannot hope to conclude directly 
using~\cite{johnsonklug}. In the rest of this section, we sketch the proof of
Theorem~\ref{thm:npidsbounds} in the case \emph{without result bounds},
explaining in particular how we handle this problem.
See Appendix~\ref{apx:npidsbounds} for the complete proof of
Theorem~\ref{thm:npidsbounds}.

 In the absence of result bounds, the~$\amd$ query containment problem
$Q \subseteq_\Gamma Q'$ can be expressed as follows: 
$\Gamma$ contains the bounded-width $\incd$s $\Sigma$ of the schema, their
primed copy $\Sigma'$, and for
each access method $\mt$ accessing relation~$R$ with input positions~$\vec x$
there is an accessibility axiom:
\[\Big(\bigwedge_i \accessible(x_i)\Big) \wedge
  R(\vec x, \vec y) \rightarrow R'(\vec x, \vec y) \wedge \bigwedge_i
  \accessible(y_i)\]
For each method~$\mt$, we can rewrite the axiom above by splitting its head, and
obtain the following pair
of axioms:
\mylistskip
\begin{compactitem}
\item (Truncated Accessibility):\\$\left(\bigwedge_i \accessible(x_i)\right) \wedge
  R(\vec x, \vec y) \rightarrow \bigwedge_i
  \accessible(y_i)$
\item (Transfer): $\left(\bigwedge_i \accessible(x_i)\right) \wedge R(\vec x, \vec y)
  \rightarrow  R'(\vec x, \vec y)$
\end{compactitem}
\mylistskip
We let $\Delta$ be the set of the Truncated Accessibility axioms and
Transfer axioms that we obtain for all the methods~$\mt$.

The constraints of~$\Delta$ are TGDs but not IDs. However, we will explain how we can take
advantage of their structure to \emph{linearize} $\Delta$ together with
$\Sigma$, i.e., construct a set~$\Sigma^\lift$
of $\incd$s
that ``simulate''
the chase by $\Sigma$ and~$\Delta$.
To define $\Sigma^\lift$ formally, we will change the signature.
Let $\sign$ be the signature of the relations used in~$\Sigma$, not including
the special unary relation $\accessible$ used in~$\Delta$; and let $w\in\NN$ be
the constant bound on the width of the $\incd$s in~$\Sigma$.
We expand~$\sign$ to the signature $\sign^{\lift}$ as follows.
For each relation~$R$ of arity~$n$ in~$\sign$,
we consider each subset $P$ of the positions
of~$R$ of size at most~$w$.
For each such subset $P$,
we add a relation
$R_{P}$ of arity $n$ to~$\sign^{\lift}$. 
Intuitively, an $R_P$-fact denotes an $R$-fact where the elements
in the positions of~$P$ are 
accessible.

Remember that our goal is to \emph{linearize}~$\Sigma$ and~$\Delta$ to a set of
$\incd$s $\Sigma^\lift$ which emulates the chase by~$\Sigma$ and~$\Delta$.
If we could ensure that $\Sigma^\lift$ has bounded width, we could then conclude
using the result of~\cite{johnsonklug}. We will not be able to enforce this, but $\Sigma^\lift$
will instead satisfy a notion of \emph{bounded semi-width} that we now define.
The \emph{basic position graph} of~$\Sigma^\lift$ is the directed graph whose nodes are the positions of relations
in~$\Sigma^\lift$ with an edge 
from position~$i$ of a relation~$T$ to position~$j$ of a relation~$U$
if and only
if the following is true:
there is an $\incd$ $\dep \in \Sigma$ whose body atom~$A$ uses relation $T$,
whose head atom~$A'$ uses relation $U$, and there is an exported variable $x$
that occurs at position~$i$ of~$A$ and at position $j$ of~$A'$.
We say that
$\Sigma^\lift$ has \emph{semi-width} bounded
by $w$ if it can be decomposed into~$\Sigma^\lift_1 \cup \Sigma^\lift_2$
where $\Sigma_1^\lift$ has width bounded by~$w$ and
the position graph of~$\Sigma_2^\lift$ is acyclic.

We can now state our linearization result:

\newcommand{\linearizeaccessids}{
  Given the set $\Sigma$ of $\incd$s of width~$w$ and the set
  $\Delta$ of Truncated Accessibility and Transfer axioms, and given
  a set of facts~$I_0$,
  we can compute in $\ptime$ a set of $\incd$s $\Sigma^\lift$ of semi-width~$w$ 
  and a set of facts $I_0^\lift$ satisfying the following:
  for any set of primed facts $I$ derivable from~$I_0$ by chasing with~$\Sigma$
and~$\Delta$, 
we can derive the same set of primed facts from~$I_0^\lift$ by chasing
with~$\Sigma^\lift$.}
\begin{proposition} \label{prop:linearizeaccessids}
  \linearizeaccessids
\end{proposition}
\begin{proof}[\myproof sketch]
  We can easily translate each  of the truncated axioms into $\incd$s on the expanded signature, but we also need to account for the propagation of
  accessibility facts via
$\incd$s in the chase. We do this by incorporating to~$\Sigma^\lift$
  some new~$\incd$s in  the
  extended signature~$\sign^\lift$ that are implied by~$\Sigma$ and~$\Delta$.
  A saturation algorithm can compute them in polynomial time, thanks to the
  polynomial bound on the number of subsets~$P$ considered in
  $\sign^\lift$.
\end{proof}

The bound on the semi-width of $\Sigma^{\lift}$ then implies an
$\np$ bound on query containment, thanks 
to the following easy generalization of the result of Johnson and Klug
~\cite{johnsonklug}:
\newcommand{\semiwidthclassic}{%
For any fixed  $w \in \NN$,
there is an $\np$ algorithm for
containment under $\incd$s of
semi-width at most~$w$.
}%
\begin{proposition} \label{prop:semiwidthclassic}
  \semiwidthclassic
\end{proposition}

This allows us to conclude the proof of Theorem~\ref{thm:npidsbounds}:
  \begin{proof}[\myproof sketch of Theorem~\ref{thm:npidsbounds}]
$\np$-hardness already holds in the case without constraints
  or result bounds \cite{access2}, so we focus on $\np$-membership.
  In the case without result bounds, 
    we have explained how to reduce to a query containment problem $Q
    \subseteq_{\Gamma} Q'$ with $\Gamma = \Sigma \cup \Sigma' \cup \Delta$. Now, we have shown in
 Proposition \ref{prop:linearizeaccessids} how $\Sigma$ and~$\Delta$ can be
    simulated by a set $\Sigma^\lift$ of $\incd$s of bounded semi-width.
    Further, $\Sigma'$ consists of bounded-width $\incd$s, so we can
    modify~$\Sigma^\lift$ to incorporate~$\Sigma'$. This allows us to
    decide the problem $Q \subseteq_{\Gamma} Q'$ in~$\np$ using Proposition~\ref{prop:semiwidthclassic}.
  The details of this argument, and its extension to the case with
    result bounds, are in Appendix~\ref{apx:npidsbounds}.
  \end{proof}

\sectionarxiv{Schema simplification for\\expressive constraints}{Schema
simplification for expressive constraints}
\label{sec:simplifychoice}

We have presented in Section~\ref{sec:simplify} the 
two kinds of simplifications anticipated in the introduction:
\emph{existence-check simplification} (using
result-bounded methods to check for the existence of tuples, as
in Example~\ref{ex:existencecheck}); and \emph{FD
simplification} (using them to retrieve
functionally determined information, as in 
in Example~\ref{ex:fd}). A natural question is then to understand
whether these simplifications capture \emph{all} the ways in which result-bounded
methods can be useful, for integrity constraints expressed in more general
constraint languages. It turns out that this is not the case when we move even
slightly beyond $\incd$s:
\begin{example} \label{ex:noexistsimplify}
Consider a schema $\aschema$ with TGD
  constraints
$T(y) \wedge S(x) \rightarrow T(x)$ and
$T(y) \rightarrow  \exists x ~ S(x) $.
We have an input-free access method $\mt_S$ on~$S$ with result bound~$1$
and a Boolean access method $\mt_T$ on~$T$.
Consider the query $Q= \exists y ~ T(y)$.
The following monotone plan answers $Q$:\\[.3em]
$T_1 \Leftarrow {\mt_S} \Leftarrow \emptyset$; \hfill
$T_2 \Leftarrow {\mt_T} \Leftarrow T_1$;  \hfill
  $T_3 \colonequals \pi_{\emptyset} T_2$; \hfill
  $\return~ T_3$;
  \\[.3em]
That is, we access $S$
and return true if the result is in~$T$.

On the other hand, consider the existence-check simplification $\aschema'$
of~$\aschema$. It has an existence-check method on~$S$, 
but
  we can only test if $S$ is non-empty, giving
 no indication whether $Q$ holds. So
  $Q$ is not answerable in $\aschema'$.
  The same holds for the FD simplification $\aschema''$ of~$\aschema$, because
  $\aschema$ implies no FDs, so $\aschema'$ and $\aschema''$ are the same.
\end{example}

Thus, existence-check simplification and FD simplification no longer suffice for
more expressive constraints. In this section, we introduce a new notion of
simplification, called \emph{choice simplification}. We will show that it allows
us to simplify
schemas with very general constraint classes, in particular TGDs as in
Example~\ref{ex:noexistsimplify}. In the next section, we will 
combine this with our query containment reduction (Proposition~\ref{prop:reduce})
to show that monotone answerability is decidable for much more expressive constraints.
Intuitively, choice simplification changes the
\emph{value} of all result bounds, replacing them by one; this
means that the number of tuples returned by result-bounded methods is not
important, provided that we obtain at least one if some exist.
We formalize the definition in this section, and show choice simplifiability for
two constraint classes: equality-free first-order logics
(which includes in particular TGDs), and $\uincd$s and FDs.
We study the decidability and complexity consequences of these results in the
next section.

\myparaskip
\myparagraph{Choice simplification}
Given a schema $\aschema$ with result-bounded methods,
its \emph{choice simplification} $\aschema'$ is defined by keeping the relations and constraints
of~$\aschema$, but changing every result-bounded  method  to have
bound~$1$. That is, every result-bounded method of~$\aschema'$ returns~$\emptyset$ if there are no matching
tuples for the access, and otherwise selects and returns one matching tuple.
We call~$\aschema$ \emph{choice simplifiable}
if any CQ having a monotone plan over~$\aschema$ has one over~$\aschema'$.
This implies that the value of the result bounds never matters.

Choice simplifiability is weaker than existence check or FD simplifiability,
but it still has a dramatic impact on the resulting query containment problem:

\begin{example} \label{ex:choice}
  Consider the schema $\aschema$ in Example~\ref{ex:simple}
and its na\"ive axiomatization in Example~\ref{ex:reduce}.
  As $\aschema$ is choice simplifiable, we can axiomatize its
  choice simplification instead, and 
the problematic axiom in the third bullet item becomes a simple $\incd$, namely
  the following:
    $\univdirect(\vec y) \rightarrow \exists \vec y' ~ \univdirect_\acc(\vec y')$.
\end{example}

\myparaskip
\myparagraph{Showing choice simplifiability}
We now give a  result showing that choice simplification holds for a huge class of constraints: all first-order
constraints that do not involve equality. This result implies, for instance,
that choice simplification holds for integrity constraints expressed as
TGDs:

\newcommand{\thmsimplifychoice}{
 Let $\aschema$ be a schema with constraints in
  equality-free first-order logic (e.g., TGDs),  and 
  let $Q$ be a CQ that is monotone answerable in~$\aschema$.
Then $Q$ is monotone answerable in the choice simplification of~$\aschema$.
}
\begin{theorem} \label{thm:simplifychoice}
  \thmsimplifychoice
\end{theorem}

\begin{proof}[\myproof sketch]
 The result is shown using 
  a simpler variant of the ``blow-up'' method of
Theorem~\ref{thm:simplifyidsexistence}.
We start with counterexample models to~$\amd$ in the choice simplification, and blow
them up by cloning the output tuples of each result-bounded access, including all facts
that hold about these output tuples.
\end{proof}

\myparaskip
\myparagraph{Choice simplifiability with $\uincd$s and FDs}
The previous result does not cover FDs.
However, we can also show a choice simplifiability result for 
FDs and $\uincd$s:

\newcommand{\thmsimplifychoiceuidfd}{
Let $\aschema$ be a schema whose constraints are
$\uincd$s and arbitrary FDs, and $Q$ be a CQ that is monotone answerable in~$\aschema$.
Then $Q$ is monotone answerable in the choice simplification of~$\aschema$.
}
\begin{theorem} \label{thm:simplifychoiceuidfd}
\thmsimplifychoiceuidfd
\end{theorem}

\begin{proof}[\myproof sketch]
We use a
strengthening of the
enlargement process of Lemma~\ref{lem:enlarge} which 
constructs $I_1^+$ and $I_2^+$ from~$I_1$ and
$I_2$ in successive steps, to fix accesses one after the other.
The construction that performs the blow-up is more complex (see Appendix~\ref{app:simplifychoiceuidfd}):
it involves copying access outputs and chasing with $\uincd$s in such a way as to avoid
violating the FDs.
\end{proof}

\sectionarxiv{Decidability using choice\\simplification}{Decidability using choice simplification}
\label{sec:complexitychoice}

In this section, we present the consequences of the choice simplifiability
results of the previous section,
in terms of decidability for expressive constraint languages. Again, these will
apply to both monotone answerability and finite monotone answerability.

\myparaskip
\myparagraph{Decidable equality-free constraints}
Thanks to Theorem~\ref{thm:simplifychoice}, we know that monotone answerability is decidable 
for a wide variety of schemas. The approach applies to constraints that do not involve equality 
and have decidable query containment.
We state here one complexity result for the class of  \emph{frontier-guarded TGDs}.
These are TGDs whose body contains a single atom  including
all exported variables.
But the same approach  applies  to  
extensions of FGTGDs
with disjunction and negation \cite{bourhis2016guarded,gnfj}.

\begin{theorem} \label{thm:decidegf}
  We can decide whether a CQ is 
monotone answerable with respect to a schema with
  result bounds whose constraints are frontier-guarded TGDs.
  The problem is $\twoexp$-complete.
\end{theorem}
\begin{proof}
  Hardness already holds by a reduction from query containment with
  frontier-guarded TGDs (see, e.g., Prop.~3.16 in \cite{thebook}), already in
  the absence of result bounds, so we will only show on $\twoexp$-membership.
By  Theorem~\ref{thm:simplifychoice} we can assume that all result bounds
are one, and by Proposition~\ref{prop:elimupper} we can replace
the schema with the relaxed version that contains
 only result lower bounds. Now, a result lower bound of~$1$ can be expressed
  as an $\incd$. Thus, Proposition~\ref{prop:reduce} allows us to reduce
  monotone answerability to a query containment problem 
  with additional frontier-guarded TGDs, and
this is decidable in~$\twoexp$ (see, e.g.,~\cite{bgo}).
\end{proof}

\myparaskip
\myparagraph{Complexity with $\uincd$s and FDs}
We now turn to constraints that consist of $\uincd$s and FDs,
and use the choice simplifiability result of
Theorem~\ref{thm:simplifychoiceuidfd} to derive complexity results for monotone
answerability with result-bounded access methods:

\newcommand{\deciduidfd}{
  We can decide monotone answerability with respect to a schema
  with result bounds whose constraints are $\uincd$s and FDs. The problem is in $\exptime$.
}
\begin{theorem} \label{thm:deciduidfd}
  \deciduidfd
\end{theorem}

Compared to Theorem~\ref{thm:npidsbounds},
this result restricts to $\uincd$s rather than $\incd$s,
and has a higher complexity, but it allows FD constraints.
To the best of our knowledge, this result is new even in the setting without result bounds.

\newcommand{\separableconstraints}{
\item for each non-result-bounded method~$\mt$ accessing relation $R$ with
  input positions $\vec x$, $\left(\bigwedge_i \accessible(x_i)\right) \wedge
  R(\vec x, \vec y) \rightarrow R'(\vec x, \vec y)
\wedge \bigwedge_i \accessible(y_i)$
\item for each result-bounded method $\mt$ accessing relation $R$ with input
  positions $\vec x$,  $\left(\bigwedge_i \accessible(x_i)\right) \wedge R(\vec x, \vec y)
  \rightarrow \exists \vec z ~ R(\vec x, \vec z) \wedge R'(\vec x, \vec z)
    \wedge \bigwedge_i \accessible(z_i)$
}
\begin{proof}[\myproof sketch of Theorem~\ref{thm:deciduidfd}]
We prove only decidability in $\twoexp$. The finer bound is in
  Appendix~\ref{app:deciduidfd}, and uses a more
involved variant of our linearization method.

  We use choice simplifiability (Theorem~\ref{thm:simplifychoiceuidfd}) to assume that all result bounds
  are one, use Proposition~\ref{prop:elimupper} to replace them by
  result lower bounds, and use Proposition~\ref{prop:reduce} to reduce to a query
  containment problem 
  $Q \subseteq_\Gamma Q'$.
The constraints $\Gamma$ include $\Sigma$, its copy $\Sigma'$,
and accessibility axioms:
  \mylistskip
\begin{compactitem}
\item $\left(\bigwedge_i \accessible(x_i)\right) \wedge
  R(\vec x, \vec y) \rightarrow R_\acc(\vec x, \vec y)$ for each
  non-result-bounded method $\mt$ accessing relation~$R$ and having input
  positions $\vec x$;
\item $\left(\bigwedge_i \accessible(x_i)\right) \wedge \exists \vec y ~
  R(\vec x, \vec y)
\rightarrow \exists \vec z ~ R_\acc(\vec x, \vec z)$ for each result-bounded
  method $\mt$ accessing relation~$R$ and having input positions $\vec x$;
\item $R_\acc(\vec w) \rightarrow R(\vec w) \wedge R'(\vec w) \wedge \bigwedge_i
  \accessible(w_i)$ for each relation~$R$.
\end{compactitem}
\mylistskip
Note that $\Gamma$ includes FDs and non-unary IDs; containment for these in general is undecidable
\cite{mitchell1983implication}.
To show decidability, we will
explain how to rewrite these axioms in a way that makes $\Gamma$
  \emph{separable}~\cite{cali2003decidability}. That is, we will be able to
  drop the FDs of~$\Sigma$ and~$\Sigma'$ without impacting containment.
First, by inlining~$R_\acc$, we can rewrite the axioms as follows:
\mylistskip
\begin{compactitem}
  \separableconstraints
\end{compactitem}
\mylistskip
We then modify the second type of axiom so that,
in addition to the variables $\vec x$ at input positions of~$\mt$ in~$R$,
they also export the variables at positions of~$R$ that are determined by the
  input positions.
This rewriting does not impact the soundness of the chase, because each chase step
  with a rewritten axiom can be
mimicked by a step with an original axiom followed by FD applications.

  After this rewriting, a simple induction on proof length (see Appendix~\ref{app:deciduidfd})
  shows that
  firing TGD triggers in the chase
  never creates a
  violation on~$R$ or~$R'$ of the FDs of~$\Sigma$ and~$\Sigma'$.
  Hence, after having applied these FDs to~$Q$,
  we know that we can drop them
  without impacting query containment.
  Let  $Q^*$ be the minimization of~$Q$ under the FDs, and
  let $\gammasep$ denote the rewritten constraints without the FDs.
  We have shown that monotone answerability is equivalent to \mbox{$Q^*
  \subseteq_{\gammasep} Q'$}.
As~$\gammasep$ contains only GTGDs, we can infer
decidability in~$\twoexp$ using~\cite{taming}, which concludes.
\end{proof}

\myparaskip\myparagraph{Extending to finite monotone answerability}
Our decidability results for choice simplifiable constraints 
extend to monotone
answerability over finite instances. For Theorem~\ref{thm:decidegf},
we simply use Proposition~\ref{prop:finitecontrol}.
As for Theorem~\ref{thm:deciduidfd}, we can also show that
it extends to the finite variant:
\newcommand{\uidfdfinite}{
  We can decide whether a CQ is finitely monotone answerable with respect to a
  schema with result bounds whose constraints are $\uincd$s and FDs. The problem is
  in $\exptime$.
}
\begin{corollary} \label{cor:uidfdfinite}
  \uidfdfinite
\end{corollary}

However, constraints that mix UIDs and FDs are not finitely controllable, so we
cannot simply use Proposition~\ref{prop:finitecontrol}. 
Instead, we will consider the \emph{finite closure} of the set of
$\uincd$s and FDs $\Sigma$.
This is the set $\Sigma^*$ of FDs and $\uincd$s that are implied
by~$\Sigma$ over finite instances.
The finite closure of~$\Sigma$
is computable (see~\cite{cosm}), and query
containment over finite instances with $\Sigma$ is equivalent to query containment
over all instances with $\Sigma^*$:

\newcommand{\thmfiniteopenucq}{
For any Boolean UCQs $Q$ and~$Q'$,  the following are equivalent:
  \begin{inparaenum}[(i.)]
\item for any finite instance $I$ satisfying~$\Sigma$,  if $Q$ holds on~$I$ then
  $Q'$ holds on $I$;
\item for any instance $I$ satisfying~$\Sigma^*$,  if $Q$ holds on $I$ then $Q'$
  holds on $I$.
\end{inparaenum}
}
\begin{theorem}[\cite{amarilli2015finite}]\label{thm:finiteopenucq} 
  \thmfiniteopenucq
\end{theorem}

This allows us to prove Corollary~\ref{cor:uidfdfinite}:

\begin{proof}[\myproof sketch of Corollary~\ref{cor:uidfdfinite}]
  We argue only for decidability: the details and $\exptime$ bound are in
  Appendix~\ref{app:uidfdfinite}.
  Let $\aschema$ be a
  schema whose constraints $\Sigma$ are UIDs and FDs, and let $\aschema^*$ be
  the same schema as~$\aschema$ but with constraints $\Sigma^*$.
  We will show that any CQ $Q$ is finitely monotone
  answerable over~$\aschema$ iff $Q$ is monotone answerable over~$\aschema^*$.
  We can decide the latter
  by Theorem~\ref{thm:deciduidfd}, so it suffices to show the equivalence.
  For the forward direction, given a monotone plan $\aplan$
  that answers~$Q$ over finite instances satisfying~$\Sigma$, we can convert it to a UCQ $Q_\aplan$. Now, since $Q$ and $Q_\aplan$ are
  equivalent over finite instances
satisfying~$\Sigma$, they are equivalent over all instances
  satisfying~$\Sigma^*$, by Theorem~\ref{thm:finiteopenucq}.
Thus $Q$ is monotone answerable over~$\aschema^*$. Conversely,
  if $Q$ is monotone answerable over
$\aschema^*$, it is finitely monotone answerable over all finite instances
  satisfying~$\Sigma^*$, but Theorem~\ref{thm:finiteopenucq} says that 
finite instances that satisfy~$\Sigma$ also satisfy~$\Sigma^*$,  which
  concludes.
\end{proof}

\section{General FO constraints} \label{sec:general}

We have shown that, for many expressive
constraint classes, the value of result
bounds does not matter, and monotone answerability is decidable.
A natural question is to understand what happens with schema simplification
and decidability for general FO constraints.
In this case, choice simplifiability no longer holds:

\begin{example} \label{ex:complex}
Consider a schema $\aschema$ with two relations $P$ and~$U$ of arity~$1$.
There is an input-free method $\mt_P$ on~$P$ with result bound~$5$,
 and an input-free method $\mt_U$ on~$U$ with no result bound.
The first-order constraints $\Sigma$ say that~$P$ has exactly $7$ tuples,
and if one of the tuples is in~$U$,
then $4$ of these tuples must be in~$U$.
  Consider the query $Q: \exists x ~ P(x) \land U(x)$.
The query is monotone answerable on~$\aschema$: the plan simply accesses $P$
  with~$\mt_P$ and intersects the result with~$U$ using~$\mt_U$.
  Thanks to~$\Sigma$, this will always return the correct result.

In the choice simplification $\aschema'$ of~$\aschema$, 
all we can do 
is access $\mt_U$, returning all of~$U$, and access $\mt_P$, returning a single
  tuple. If this tuple is not in~$U$, we have no information on  whether or not
$Q$ holds. Hence, we can easily see that $Q$ is not answerable on~$\aschema'$.
\end{example}

The constraints in the previous example still lie
in a decidable language, namely, two-variable logic with counting quantifiers~\cite{pratt2009data}. 
We can 
still decide monotone answerability for this language even without any schema simplification;
see Appendix~\ref{app:aritytwo}.
Unsurprisingly, if we move to  constraints 
where containment is undecidable, 
then the monotone answerability problem is also undecidable, even in cases such as equality-free FO
which are choice simplifiable:

\newcommand{\foundecid}{It is undecidable to check if $Q$ is monotone answerable
with respect to equality-free FO constraints.}

\begin{proposition} \label{prp:undec} 
\foundecid
\end{proposition}
The same holds for other constraint languages where query containment is undecidable, such as general TGDs.

\section{Summary and Conclusion} \label{sec:conc}
We formalized the problem of answering queries in a complete way 
by accessing Web services that only return a bounded number of answers to each
access, assuming integrity constraints on the data.
We showed how to reduce this to a standard reasoning problem, query containment
with constraints. 
We have further shown simplification results for many classes of  constraints, 
limiting the ways in which  a query can be answered using result-bounded plans, 
thus simplifying the corresponding query containment problem.
By coupling these results with an analysis of query containment,
we have derived complexity bounds
for monotone answerability under several
classes of constraints.
 Table~\ref{tab:results} summarizes which simplifiability result
 holds for each constraint class, 
as well as the decidability and
complexity results.  
We leave open the complexity of monotone answerability
with result bounds for some important
cases: Full TGDs, and more generally weakly-acyclic TGDs.
Our choice approximation result applies here, but we do not know
how to analyze the chase even for the simplified containment problem.

We have  restricted to \emph{monotone} plans 
throughout the paper. As explained in
Appendix~\ref{apx:ra},
the reduction to query containment still applies to plans  that can use
negation.
Our schema simplification results 
also extend easily to 
answerability with such plans, but lead to a more involved query containment
problem. Hence,
we do not know how to show decidability of the
answerability problem  for $\uincd$s and FDs with such plans.
We also leave open the question of whether choice simplifiability holds
for general FDs and $\incd$s.

In our study of the answerability problem,
we have also introduced technical tools which could be useful in a wider
context. One example is the blowing-up method that we use in schema
simplification results; a second example is 
linearization, for which we intend to study further applications.

\myparagraph{Acknowledgments}
The work  was funded by EPSRC grants
PDQ (EP/M005852/1), ED$^3$ (EP/N014359/1), and DBOnto (EP/L012138/1).

\pagebreak
\bibliographystyle{abbrv}
\balance
\pagebreak
\bibliography{algs}
\pagebreak
\appendix
\section{Details for Section~\lowercase{\ref{sec:prelims}}: Preliminaries}
\label{app:idempotent}

In the body of the paper we defined a semantics for plans using valid access
selections, which assumed that 
multiple accesses with a result-bounded method always return the same output.
We also claimed that all our results held without this assumption.
We now show the alternative semantics where this assumption does not hold, and
show that indeed the choice of semantics makes no difference. We will call
\emph{idempotent semantics} in this appendix the one that we use in the main
body of the paper, and \emph{non-idempotent semantics} the one that we now
define.

Intuitively, the idempotent semantics, as used in the body, assumes that 
the access selection function is chosen for the entire plan,
so that all calls with the same input to the same
 access method return the same output. The non-idempotent
semantics makes no such assumption, and can choose a different valid access selection
for each access.
In both cases, the semantics is a function taking an instance $I$ for the input schema and
the input tables of the plan, returning as output a set of \emph{possible outputs} for each output table of the plan.

Formally, given a schema $\aschema$ and instance $I$,
an \emph{access selection} is a function mapping each access on~$I$ to an
output to the
access, as defined in the main text, and it is valid if every output returned by
the access selection is a valid output to the corresponding access.
Given a valid access selection~$\aselect$, we can associate to each instance $I$ and each plan $\aplan$
an \emph{output} by induction on the number of commands. The general
scheme for both semantics is the same: for an
access command $T \Leftarrow_\outmap \mt \Leftarrow_\inmap E$
the output is obtained by evaluating $E$ to get a collection of tuples and then performing an access with~$\mt$ using each tuple, putting
the union of the corresponding output selected by~$\aselect$ into~$T$.   The semantics of
middleware query commands is the usual semantics for relational algebra. The semantics
of concatenation of plans is via composition.

The difference between the two semantics is: for the
\emph{idempotent semantics}, given $I$ we take  the union
over all valid access selections~$\aselect$ of the output of the \emph{entire} plan
for~$I$ and $\aselect$;
for the \emph{non-idempotent semantics}, we calculate the possible outputs of
\emph{each individual
access command} as the union of the outputs for all~$\aselect$, we calculate  the output
of a query middleware command as usual, and then we calculate the possible outputs for a plan
via composition.

\begin{example} \label{ex:twosem}
Consider a schema with a input-free access method $\mt$ with result bound~$5$ on relation~$R$.
Let $\aplan$ be the plan that accesses $\mt$ twice and then determines
  whether the intersection of the results is non-empty:\\[.3em]
$T_1 \Leftarrow \mt \Leftarrow \emptyset; \hfill
T_2 \Leftarrow \mt \Leftarrow \emptyset; \hfill
T_0 \colonequals \pi_{\emptyset} (T_1 \cap T_2); \hfill
  \return~T_0$\\[.3em]
As $T_1$ and $T_2$ are identical under the idempotent semantics,
  $\aplan$ just tests if $R$ is non-empty.
  Under the non-idempotent semantics, $\aplan$ is non-deterministic,
since it can return empty or non-empty when $R$ contains at least $10$ tuples.
\end{example}

Note that, in both semantics, when we use multiple access methods on the same relation,
there is no requirement that an access selection be ``consistent'':
if an instance $I$ includes a fact $R(a,b)$ and we have result-bounded access methods $\mt_1$ on the
first position of~$R$ and $\mt_2$ on the second position of~$R$, then an
access to~$\mt_1$ on~$a$ might return $(a,b)$ even if an access to~$\mt_2$ on
$b$ does not return
$(a,b)$. This captures the typical situation where distinct access methods use unrelated criteria to
determine which tuples to return.

It is clear that if a query that has a plan that answers it under the non-idempotent
semantics, then the same plan works under the idempotent semantics.
Conversely, Example~\ref{ex:twosem} shows that that a given plan may answer a query
under the idempotent semantics, while it does  not answer any query under the non-idempotent semantics.
However, if a query $Q$ has  some plan that answers it under the idempotent
semantics, we can show that it also does under the non-idempotent semantics:

\begin{proposition} \label{prop:idempsuffices}
For any CQ  $Q$ over schema $\aschema$, there is a  monotone plan that
answers $Q$ under  the idempotent semantics with respect to~$\aschema$ iff
there is a  monotone plan that answers $Q$ under the non-idempotent semantics.
Likewise, there is an RA-plan that answers $Q$ under the idempotent semantics
  with respect to~$\aschema$ iff there is an RA-plan that answers $Q$ under the
  non-idempotent semantics.
\end{proposition}

We first give  the argument for \emph{RA-plans} (i.e., non-monotone plans, which allow arbitrary relational algebra expressions).
If there is a plan $\aplan$ that answers $Q$ under the non-idempotent semantics, then clearly $\aplan$
also answers $Q$ under the idempotent semantics, because there are less
  possible outputs. 

In the other direction, suppose $\aplan$ answers $Q$ under the idempotent semantics.
Let $\kw{cached}(\aplan)$ be the function that executes $\aplan$, but whenever it encounters
an access $\mt$ on a binding $\accbind$ that has already been performed in a previous command, it uses
the values output by the prior command rather than making a new access, i.e., it
uses ``cached values''.
Executing $\kw{cached}(\aplan)$  under the non-idempotent semantics gives
  exactly the same outputs as executing $\aplan$ under the idempotent semantics,
  because $\kw{cached}(\aplan)$ never performs the same access twice.
  Further we can implement $\kw{cached}(\aplan)$ as an RA-plan $\aplan'$: 
  for each access command
$T \Leftarrow \mt \Leftarrow E$  in~$\aplan$, we pre-process it in~$\aplan'$ by
removing from the output of~$E$ any tuples previously accessed in~$\mt$,
  using a middleware query command with the relational difference
  operator. We then perform an access to~$\mt$
with the remaining tuples, cache the output for further accesses, and
  post-process the output with a middleware query command to add back the
  output tuples cached from previous accesses.
Thus $\aplan'$ answers $Q$ under the idempotent semantics as required.

\medskip

Let us now give the argument for  \emph{monotone} plans (i.e., USPJ-plans), which are the plans
used throughout the body of the paper. Of course the forward direction
is proven in the same way, so we focus on the backward direction.
Contrary to plans that can use negation, we can no
longer avoid making accesses that were previously performed, because we can no
longer remove input tuples that we do not wish to query. However, we can still
cache the output of each access, and union it back when performing further
accesses.

Let $\aplan$ be a plan that answers $Q$ under the idempotent semantics.
We use Proposition~\ref{prop:elimupper} about the elimination of result upper
bounds to assume without loss of generality that~$\aplan$ answers the query~$Q$
on the schema $\relaxs(\aschema)$, where all result bounds of~$\aschema$ 
are replaced with result lower bounds only.

We define the plan $\aplan'$ from~$\aplan$, where access commands are modified
in the following way: whenever we perform an access for a method $\mt$ in an
access command $i$, we cache the input of access command $i$ in a special
intermediate table $\text{Inp}_{\mt,i}$ and its output in another table
$\text{Out}_{\mt,i}$, and then we add to the output of access command $i$ the
result of unioning, over all previously performed accesses with~$\mt$ for~$j <
i$, the intersection $\text{Inp}_{\mt,i} \cap \text{Inp}_{\mt,j}$ joined
with~$\text{Out}_{\mt,j}$. Informally, whenever we perform an access with a set
of input tuples, we add to its output the previous outputs of the accesses with
the same tuples on the same methods earlier in the plan. This can be implemented
using USPJ operators. For each table defined on the left-hand side of an access
or middleware command in~$\aplan$, we define its \emph{corresponding table} as
the table in~$\aplan'$ where the same result is defined: for middleware
commands, the correspondence is obvious because they are not changed from
$\aplan$ to~$\aplan'$; for access commands, the corresponding table is the one
where we have performed the postprocessing to incorporate the previous tuple
results.

We now make the following claim: 
\begin{claim} 
  \label{clm:sameres}
  Every possible output of~$\aplan'$ in the non-idempotent
semantics is a subset of a possible output of~$\aplan$ in the idempotent
  semantics, and is a superset of a possible output of~$\aplan$ in the
  idempotent semantics.
  \end{claim}
 
 This
suffices to establish that~$\aplan'$ answers the query $Q$ in the non-idempotent
semantics, because, as~$\aplan$ answers $Q$ in the idempotent semantics, its
only possible output on an instance $I$ in the idempotent semantics is $Q(I)$,
so Claim~\ref{clm:sameres} implies that the only possible output of~$\aplan'$
on~$I$ is also~$Q(I)$, so $\aplan'$ answers $Q$ under the non-idempotent
semantics, concluding the proof. So it suffices to prove
Claim~\ref{clm:sameres}. We now do so:

\begin{proof}
  Letting $O$ be a result of~$\aplan'$ under
  the non-idempotent semantics on an instance~$I$, and letting 
  $\aselect_1, \ldots, \aselect_n$ be the choice of valid access selections used for each
  access command of~$\aplan'$ to obtain~$O$, we first show that $O$ is a
  superset of a possible output of~$\aplan$ in the idempotent semantics, and
  then show that $O$ is a subset of a possible output of~$\aplan$ in the
  idempotent semantics.

  To show the first inclusion, let us first consider the access selection
  $\aselect^-$ on~$I$ defined in
  the following way: for each access binding $\accbind$ on a method $\mt$,
  letting $\aselect_i$
  be the access selection for the first access command of~$\aplan$ where the
  access on~$\accbind$ is performed on~$\mt$, we define $\aselect^-(\mt,
  \accbind) \colonequals \aselect_i(\mt, \accbind)$; if the access is never
  performed, define $\aselect$ according to one of the~$\aselect_i$ (chosen arbitrarily). 
  We see that~$\aselect^-$ is a valid access selection for~$I$, because each
  $\aselect_i$
  is a valid access selection for~$i$, and for each access $\aselect^-$ returns the
  output of one of the~$\aselect_i$, which is valid.
  Now, by
  induction on the length of the plan, it is clear that for every table in the
  execution of~$\aplan$ on~$I$ with~$\aselect^-$, its contents are a \emph{subset}
  of the contents of the corresponding table in the execution of~$\aplan'$ on~$I$ with
  $\aselect_1, \ldots, \aselect_n$. Indeed, the base case is trivial. The induction case
  for middleware commands is by monotonicity of the USPJ operators. The
  induction case on access commands is simply because we perform an access with
  a subset of bindings: for each binding $\accbind$, if this is the first time we
  perform the access for this method on~$\accbind$, we obtain the same output in
  $\aplan$ as in~$\aplan'$, and if this is not the first time, in~$\aplan$ we
  obtain the output as we did the first time, and in~$\aplan'$ we still obtain it
  because we retrieve it from the cached copy. The conclusion of the induction
  is that the output of~$\aplan$ on~$I$ under $\aselect^-$ is a subset of the output
  $O$ of~$\aplan'$ on~$I$ under $\aselect_1, \ldots, \aselect_n$.

  Let us now show the second inclusion by considering the access selection
  $\aselect^+$ on~$I$ defined in the following way: for each access binding
  $\accbind$ and method $\mt$, we define $\aselect^+(\mt, \accbind) \colonequals
  \bigcup_{1 \leq i \leq n} \aselect_i(\mt, \accbind)$. That is, $\aselect^+$ returns
  all outputs that are returned in the execution of~$\aplan'$ on~$I$ in the
  non-idempotent semantics with~$\aselect_1, \ldots, \aselect_n$. This is a valid
  access selection, because for each access and binding it returns a superset
  of a valid output, so we are still obeying the result lower bounds, and there
  are no result upper bounds because we we are working with the schema
  $\relaxs(\aschema)$ where result upper bounds have been eliminated.
  Now, by induction
  on the length of the plan, analogously to the case above, we see that for
  every table in the execution of~$\aplan$ on~$I$ with~$\aselect^+$, its contents are
  a \emph{superset} of that of the corresponding table in the execution of
  $\aplan'$ on~$I$ with~$\aselect_1, \ldots, \aselect_n$: the induction case is because
  each access on a binding in~$\aplan'$ cannot return more than the outputs  of
  this access in all the~$\aselect_i$, and this is the output obtained with
  $\aselect^+$. So we have shown that~$O$ is a subset of a possible output of
  $\aplan$, and that it is a superset of a possible output of~$\aplan$,
  concluding the proof of the claim.
\end{proof}

This concludes the proof of Proposition~\ref{prop:idempsuffices}.

\section{Finite controllability of common classes of constraints} \label{app:finitecontrol}

Recall that we defined
a set of constraints $\Sigma$ to be finitely controllable
if for every Boolean UCQ $Q$ and~$Q'$,
the following are equivalent:
\begin{itemize}
\item  $Q \subseteq_\Sigma Q'$ 
\item 
if a finite instance $I$ satisfies $Q$, then it
also satisfies $Q'$
\end{itemize}
That is, the finite and unrestricted versions of query containment coincide.
Note that the first item necessarily implies the second.

In the body we claimed that common classes of constraints are finitely controllable. 
For frontier-guarded TGDs, this follows immediately from \cite{gnfj} which shows
that the guarded negation fragment of first-order logic has the finite model property.
Indeed,
$\gnf$ can express any sentence of the form $Q \wedge \Sigma \wedge \neg Q'$ where $\Sigma$ is a set
of frontier-guarded TGDs and $Q,Q'$ are Boolean UCQs. Hence, we can express in
$\gnf$ that there is a counterexample to the query containment problem
$Q \subseteq_\Sigma Q'$. Now, the finite model property of~$\gnf$ implies that there
is a counterexample to $Q \subseteq_\Sigma Q'$ iff there is a finite
counterexample, and this is precisely what finite controllability says.
The same argument applies to disjunctive Guarded TGDs~\cite{bourhis2016guarded}.

For classes of constraints where the chase terminates, which includes FDs as well as weakly-acyclic TGDs, finite controllability
in the sense above is also easy to see. Suppose the first item fails for some $Q'$, then there is an
instance $I$ satisfying the constraints $\Sigma$, and satisfying $Q \wedge \neg
Q'$, so that $I$ satisfies some disjunct $Q_i$ of $Q$. Thus the containment
problem
$Q_i \subseteq_\Sigma Q'$ fails. Letting $I_i$ be the chase
of $Q_i$ by the constraints~$\Sigma$, we know that each  disjunct $Q'_j$ of $Q'$ must fail to hold in $I_i$, since the
chase is universal for containment. But then, because the chase by~$\Sigma$
terminates, we know that $I_i$ is finite, so it contradicts the second item.

\section{Proofs for Section~\lowercase{\ref{sec:reduce}}:
Reducing to Query Containment}
\subsection{Proof of Proposition~\ref{prop:altdef}: Equivalence Between
Accessible Part Notions}
\label{app:determinacy}

Recall the statement of Proposition~\ref{prop:altdef}:
\begin{quote}
  \propaltdef
\end{quote}

\begin{proof}
It suffices to show that the two definitions of ``having more accessible data'' agree.
That is, we show that
the following are equivalent:
 \begin{enumerate}[(i)]
\item $I_1$ and $I_2$ have a common subinstance $I_\acc$ that
is access-valid in~$I_1$.
\item There  are $A_1 \subseteq A_2$ such that $A_1$  is an accessible part
for~$I_1$ and~$A_2$ is an accessible part for~$I_2$.
\end{enumerate}

Suppose  $I_1$ and $I_2$ have a common subinstance $I_\acc$ that is access-valid
  in~$I_1$. 
This means that we can define a valid access selection~$\aselect_1$ that takes any
  access performed with values of~$I_\acc$ and a method of~$\aschema$, and maps it
  to a set of matching tuples in~$I_\acc$ that is valid  in $I_1$. 
We can
  extend $\aselect_1$ to a function~$\aselect_2$ which returns a superset of the
  tuples returned by $\aselect_1$ for accesses with values of~$I_\acc$, and returns
  an arbitrary set of tuples from~$I_2$ otherwise, such that this output to the
  access is valid in~$I_2$.
  We have $\accpart(\aselect_1,I_1) \subseteq\accpart(\aselect_2,I_2)$, and thus the first item
  implies the second.

Conversely, suppose there are $A_1 \subseteq A_2$ such that
$A_1$ is an accessible part for~$I_1$ and $A_2$ is an accessible part for~$I_2$.
Let $\sigma_1$ and $\sigma_2$ be the valid access selections used to define~$A_1$ and
$A_2$, so that 
that $\accpart(\aselect_1, I_1) \subseteq \accpart(\aselect_2, I_2)$.
Let $I_\acc \colonequals \accpart(\aselect_1, I_1)$, and let us show that
$I_\acc$ is a common subinstance of~$I_1$ and~$I_2$ that is access-valid
in~$I_1$.
By definition, we know that $I_\acc$ is a subinstance of~$I_1$, and by
  assumption we have $I_\acc \subseteq A_2 \subseteq I_2$, so indeed $I_\acc$ is
  a common subinstance of~$I_1$ and $I_2$. Now, to show that it is access-valid
  in~$I_1$, consider any access $\abind, \mt$
  with values in~$I_\acc$. We know that there is $i$ such
that~$\abind$ is in~$\accpart_i(\aselect_1, I_1)$, so by definition of the
  fixpoint process and of the access selection~$\aselect_1$ there is a valid
  output  in~$\accpart_{i+1}(\aselect_1, I_1)$, hence in~$I_\acc$.
  Thus, $I_\acc$ is access-valid. This shows the converse
  implication, and concludes the proof.
\end{proof}

\subsection{Proof of Theorem~\ref{thm:equiv}: Equivalence Between Answerability and $\amd$}
\label{app:equiv}

Recall the statement of the theorem:
\begin{quote}
  \thmequiv
\end{quote}

As in our other results involving $\amd$, we will use the definition in terms of access-valid
subinstances, i.e., we use Proposition~\ref{prop:altdef}.

We first prove the ``easy direction'':

\begin{proposition} \label{prp:plantoproof}
If $Q$ has a  (monotone) plan $\aplan$ that answers it \wrt\ $\aschema$, then
$Q$ is $\amd$ over~$\aschema$.
\end{proposition}
\begin{proof}
  We use the definition of $\amd$ given in Proposition~\ref{prop:altdef}.
Assume that there are two instances $I_1, I_2$ satisfy the constraints
  of~$\aschema$ and that there is a common subinstance $I_\acc$ 
  that is access-valid in~$I_1$. Let us show that $Q(I_1) \subseteq Q(I_2)$.
  As $I_\acc$ is access-valid, let $\aselect_1$ be a valid access selection
  for~$I_\acc$: for any access with values in~$I_\acc$, the access selection
  $\aselect_1$ returns an output which is valid in~$I_\acc$. We extend $\aselect_1$
  to a valid access selection for~$I_2$ as in the proof of
  Proposition~\ref{prop:altdef}: for accesses in~$I_\acc$, the access selection
  $\aselect_2$ returns a superset of~$\aselect_1$, which is possible because $I_\acc
  \subseteq I_2$, and for other accesses it returns some valid subset of tuples
  of~$I_2$.
We  argue that  
for each temporary table of~$\aplan$, its value
when  evaluated on~$I_1$ with~$\aselect_1$, is contained
in its value when evaluated on~$I_2$ with~$\aselect_2$.
 We prove this by induction on~$\aplan$. As the plan is monotone, the property
  is preserved by query middleware commands, so inductively it suffices to look at an
access command  $T \Leftarrow \mt \Leftarrow E$ with~$\mt$ an access method on
some relation. Let $E_1$ be the value of~$E$ when evaluated on~$I_1$
  with~$\aselect_1$,
and let $E_2$ be the value when evaluated on~$I_2$ with~$\aselect_2$. Then by the monotonicity
of the query $E$ and the induction hypothesis, we have $E_1 \subseteq E_2$.
Now, given a tuple $\vec t$ in~$E_1$,  let  $M^1_{\vec t}$ be the
set of ``matching tuples''
(tuples for the relation~$R$ extending $\vec t$) in 
$I_1$ selected by $\aselect_1$. Similarly let $M^2_{\vec t}$ be the set
selected by $\aselect_2$ in~$I_2$.  By construction of~$\aselect_2$, we have $M^1_{\vec t} \subseteq M^2_{\vec t}$,
and thus $\bigcup_{\vec t \in E_1} M^1_{\vec t} \subseteq \bigcup_{\vec t \in
  E_1} M^2_{\vec t}$,
which  completes the induction.
  Thanks to our
  induction proof, we know that the output of~$\aplan$ on~$I_1$ with~$\aselect_1$ is a
  subset of the output of~$\aplan$ on~$I_2$ with~$\aselect_2$. As we have assumed
  that $\aplan$ answers~$Q$ on~$\aschema$, this means that $Q(I_1) \subseteq
  Q(I_2)$, which is what we wanted to show.
\end{proof}

For the other direction, we first use the corresponding result in the case without result bounds:

\begin{theorem}[\cite{thebook,ustods}]
  \label{thm:mdetermandplansclassic}
  For any CQ~$Q$ and schema $\aschema$ (with no result bounds)
  whose constraints $\Sigma$ are expressible
  in active-domain first-order logic, the following are equivalent:
\begin{compactenum}
\item $Q$ has a  monotone plan that answers it over~$\aschema$
\item
$Q$ is $\amd$ over~$\aschema$.
\end{compactenum}
\end{theorem}

Thus, for schemas without result-bounded  methods, existence
of a monotone plan is the same as~$\amd$, and both can be expressed as a query containment problem.
It is further shown in~\cite{ustods} that a monotone plan can be extracted from
any proof of the query containment for~$\amd$.
This reduction to query containment is what we will now extend to the setting
with result-bounded methods in the main text.

We adapt the above result to the setting of result-bounded methods with a simple
construction that allows us
to rewrite away the result-bounded methods (expressing them in the constraints):

\myparaskip
\myparagraph{Axiomatization of result-bounded methods}
Given a schema $\aschema$ with constraints and access methods, possibly
with result bounds, we will define an auxiliary schema $\elimnd(\aschema)$ without
result bounds.
In the schema $\elimnd(\aschema)$, for
every method $\mt$ with  result bound~$k$ on relation~$R$  we have a new relation
$R_\mt$ whose arity agrees with that of~$R$. Informally, $R_\mt$ stores
only up to~$k$ result tuples for each input.
The constraints include all the constraints of~$\aschema$ (on the original
relation names).
In addition, we have for every method $\mt$ with input positions $i_1 \ldots i_m$
and result bound~$k$, the following axioms:

\begin{itemize}
\item A \emph{soundness of selection} axiom stating that~$R_\mt$ is a subset of
  $R$.
\item An axiom stating that for any binding of the input positions,
$R_\mt$ has at most $k$ distinct matching tuples
\item For each $1 \leq j \leq k$, 
  a \emph{result lower bound axiom} stating that,
for any values $c_{i_1} \ldots c_{i_m}$, if
$R$
    contains at least $j$ matching tuples (i.e., tuples $\vec c$ that extend
    $c_{i_1} \ldots c_{i_m}$), then $R_\mt$ contains at least $j$  such tuples. 
\end{itemize}

In this schema we have the same access methods, except that \emph{any
$\mt$ with a  result bound over~$R$ is removed, and in its place we add an access method
with no  result bound over~$R_\mt$}.

Given a query $Q$ over  $\aschema$, we can consider it as a query  over
$\elimnd(\aschema)$ instances by simply ignoring the additional relations.

We claim that, in considering $Q$ over~$\elimnd(\aschema)$ rather than $\aschema$,
we do not change monotone answerability.

\begin{proposition} \label{prop:elimresultboundplan}
For any query $Q$ over~$\aschema$, there is a monotone plan that answers~$Q$
  over~$\aschema$ iff there is a monotone plan that answers~$Q$ over~$\elimnd(\aschema)$.
\end{proposition}

Proposition~\ref{prop:elimresultboundplan} thus shows that we can axiomatize 
result bounds, at the cost of including new constraints. 

\begin{proof}
Suppose that there is a monotone plan $\aplan$ over~$\aschema$ that answers $Q$. Let
$\aplan'$ be formed from~$\aplan$ by replacing every access
with method $\mt$ on relation~$R$ with an access to~$R_\mt$ with the
  corresponding method.
We claim that~$\aplan'$ answers $Q$ over~$\elimnd(\aschema)$.
Indeed, given an instance $I'$ for~$\elimnd(\aschema)$, we can drop the relations
$R_\mt$ to get an instance $I$ for~$Q$, and use the relations $R_\mt$ to define
a valid access selection~$\aselect$ for each method of~$\aschema$, and we can
show that~$\aplan$ evaluated with~$\aselect$ over~$I$ gives the same output as
$\aplan'$ over~$I$. Since the former evaluates to~$Q(I)$, so must the latter.

Conversely,
suppose that there is a monotone plan $\aplan'$ that answers $Q$ over~$\elimnd(\aschema)$.
Construct $\aplan$  from~$\aplan'$ by replacing accesses to~$R_\mt$ with accesses to~$R$.
We claim that~$\aplan$ answers $Q$ over~$\aschema$.
To show this, consider
an instance $I$ for~$\aschema$, and a particular valid access selection~$\aselect$,
  and let us show that the evaluation of~$\aplan$ on~$I$ following~$\aselect$
  correctly answers~$Q$.
We 
build an instance $I'$ of~$\elimnd(\aschema)$ by interpreting $R_\mt$ as follows:
  for each tuple $\vec{t}$ such that $R(\vec{t})$ holds in~$I$, project
  $\vec{t}$ on the input positions $i_1 \ldots i_m$
of~$\mt$, and include all of the outputs of this access according to~$\aselect$
  in~$R_\mt$. As the outputs of accesses according to~$\aselect$ are must be, $I'$
  must satisfy the
  constraints of~$\elimnd(\aschema)$.
  We define a valid access selection~$\aselect'$
  from~$\aselect$ so that every access on~$R_\mt$ returns the output of the
  corresponding access on~$R$ according to~$\aselect$.
Since $\aplan'$ answers $Q$,
  we know that evaluating $\aplan'$ on~$I'$ with~$\aselect'$ yields the output~$Q(I')$ of~$Q$
  on~$I'$.
Now, the definition of~$\aselect'$ ensures that the accesses made by $\aplan'$ on~$I'$
  under~$\aselect'$ are exactly the same
as those made by $\aplan$ on~$I$ under~$\aselect$, and that the output of these
accesses are the same.
Thus $\aplan$ evaluated 
on~$I$ under~$\aselect$ gives the same result as~$\aplan'$ does on~$I'$
  under~$\aselect'$, namely, $Q(I')$.
Now, $Q$ only uses the original relations of~$\aschema$, so the definition
  of~$I'$ clearly implies that $Q(I') = Q(I)$, so indeed the evaluation
  of~$\aplan$ on~$I$ under~$\aselect$ returns~$Q(I)$. As this holds for any
  valid access selection~$\aselect$, we have shown that~$\aplan$ answers~$Q$
  over~$\aschema$, the desired result.
\end{proof}

The equivalence of a schema $\aschema$ with
result bounds and its variant $\elimnd(\aschema)$ easily
extends  to~$\amd$.

\begin{proposition} \label{prop:elimresultbounddet}
For any query $Q$ over~$\aschema$, the corresponding
query is $\amd$ over~$\elimnd(\aschema)$ if and only if
$Q$ is $\amd$ over~$\aschema$.
\end{proposition}

\begin{proof}
  For the forward direction,
  assume $Q$ that is $\amd$ over~$\elimnd(\aschema)$, and let us show that $Q$ is
  $\amd$ over~$\aschema$.
  We use the characterization of $\amd$ in terms of access-valid
subinstances given in Proposition~\ref{prop:altdef}.
  Let 
  $I_1$ and $I_2$ be instances satisfying the constraints of~$\aschema$, 
  and let $I_\acc$ be a common subinstance of~$I_1$ and~$I_2$ which is
  access-valid in~$I_1$ for~$\aschema$. Let $\aselect_1$ be a valid access selection
  for~$I_\acc$. We can extend it to an access selection $\aselect_2$ for~$I_2$ 
that ensures that every access
  with $\aselect_2$ returns a superset of the tuples obtained with~$\aselect_1$.
  We now extend~$I_1$ into an instance~$I_1'$ for~$\elimnd(\aschema)$ by
  interpreting each~$R_\mt$ as the union of the
outputs given by~$\aselect_1$ over every
  possible access with~$\mt$ on~$I_\acc$. 
We define $I_2'$ from $I_2$ and
  $\aselect_2$ in the same way. As the access outputs given by~$\aselect_1$ and
  $\aselect_2$ must be valid, we know that $I_1'$ and $I_2'$ satisfy the
  new constraints of~$\elimnd(\aschema)$, and clearly they still satisfy the
  constraints of~$\aschema$. Now extend $I_\acc$ to~$I_\acc'$ by adding all
  $R_\mt$-facts of~$I_1'$ for all~$\mt$. Clearly $I_\acc'$ is a subinstance
  of~$I_1'$.  It is access-valid because $I_\acc$ was access-valid. 
It is a subinstance
  of~$I_2'$ because $I_\acc$ is a subinstance of~$I_2'$ and because the
  $R_\mt$-facts in~$I_1'$ also occur in~$I_2'$ by construction of~$\aselect_2$.
  Thus, because $Q$ is $\amd$ over~$\elimnd(\aschema)$, we know that $Q(I_1') \subseteq
  Q(I_2')$. Now, as $Q$ only uses the relations in~$\aschema$, we have 
  $Q(I_1)=Q(I'_1)$ and $Q(I_2)=Q'(I'_2)$, so we have shown that $Q(I_1) \subseteq
  Q(I_2)$, concluding the forward direction.

Conversely, suppose $Q$ is $\amd$ over~$\aschema$
and consider instances  $I'_1$ and $I'_2$ for $\elimnd(\aschema)$
with valid access selections $\aselect'_1$ and $\aselect'_2$ giving accessible parts $A'_1 \subseteq A'_2$. We create  an instance $I_1$ for
$\aschema$ from~$I'_1$ by dropping the relations $R_\mt$, and similarly
create $I_2$ from~$I'_2$. Clearly both satisfy the constraints of~$\aschema$.
We modify $\aselect'_1$ to obtain an access selection~$\aselect_1$ for~$I_1$:
for every access on~$I_1$ with a method~$\mt$, the output is that of the
corresponding access with~$\aselect_1'$ on~$R_\mt$; we do the same to build
$\aselect_2$ from~$\aselect_2'$. It is clear that these access selections are
valid, i.e., that they return valid outputs to any access; and letting
$A_1$ and~$A_2$ be the corresponding accessible parts of~$I_1$ and~$I_2$, it is
clear that $A_1 \subseteq A_2$. Thus, because $Q$ is $\amd$ over~$\aschema$, we
know that $Q(I_1) \subseteq Q(I_2)$, and again we have $Q(I_1) = Q(I_1')$ and
$Q(I_2) = Q(I_2')$ so we have $Q(I_1') \subseteq Q(I_2')$, which concludes.
\end{proof}

Putting together Proposition \ref{prop:elimresultboundplan}, Proposition
\ref{prop:elimresultbounddet} and Theorem \ref{thm:mdetermandplansclassic},
 we complete the proof of Theorem \ref{thm:equiv}.

\section{Proofs for Section~\lowercase{\ref{sec:simplify}}: 
Simplifying Result Bounds}
\subsection{Proof of Theorem~\ref{thm:simplifyidsexistence}: Existence-Check
Simplification for~IDs}
\label{app:simplifyidsexistence}

In this appendix, we show Theorem~\ref{thm:simplifyidsexistence}. Recall its
statement:

\begin{quote}
  \thmsimplifyidsexistence
\end{quote}

We will show the contrapositive of this statement. Let us assume that~$Q$ does
not have a monotone plan in the existence-check simplification~$\aschema'$ of
$\aschema$.
We will show that this implies that~$Q$ is not $\amd$ in~$\aschema$: this allows us to conclude
because, by Theorem~\ref{thm:equiv}, this implies that~$Q$ has no
monotone plan in~$\aschema$. Throughout the proof, we will use the definition of
$\amd$ given by Proposition~\ref{prop:altdef}.

Thus it suffices to show:

\begin{quote} Consider a schema $\aschema$ whose  constraints are 
$\incd$s,  and
let $Q$ be a CQ that is $\amd$ with respect to~$\aschema$.
Then $Q$ is also $\amd$ in the existence-check simplification of~$\aschema$.
\end{quote}

We now prove the theorem, using Lemma~\ref{lem:enlarge}:

\begin{proof}
Let $\aschema$ be the original
schema and $\aschema'$ be the existence-check simplification.
Notice that the query $Q$ is indeed posed on the common relations of~$\aschema$
and~$\aschema'$, i.e., it does not involve the $R_\mt$ relations added
  in~$\aschema'$.
  To use Lemma~\ref{lem:enlarge}, suppose that we have  a counterexample $(I_1, I_2)$ to
$\amd$ for~$Q$ and the  simplification~$\aschema'$, i.e.,
the instances $I_1$ and $I_2$ satisfy the constraints $\Sigma'$ of~$\aschema'$, the instance $I_1$ satisfies $Q$ and
the instance $I_2$ violates $Q$, and $I_1$ and $I_2$ have a common subinstance $I_\acc$ that is
access-valid in~$I_1$.
As mentioned in the main text of the paper, we will show how to ``blow up'' each instance
to~$I_1^+$ and $I_2^+$ which have a common subinstance which is access-valid in~$I_1^+$,
  i.e., we must ensure that each access to a result-bounded method with
a result bound in~$I_1^+$ returns either no tuples or more tuples than the bound.
In the blowing-up process we will
preserve the constraints $\Sigma'$ and the properties of the~$I_i$ with respect
to the CQ~$Q$.

We now explain how $I_1^+$ and $I_2^+$ are formed.
The first step is ``obliviously chasing with the existence-check constraints'':
for any existence-check constraint $\dep$
of the form
\[
  \forall x_1 \ldots x_m ~ \checkview_\mt(\vec x) \rightarrow ~ \exists y_1 \ldots y_n ~
  R(\vec x, \vec y)
\]
and any homomorphism $h$ of the variables $x_1 \ldots x_n$ to~$I_\acc$,
we extend the  mapping by choosing infinitely many fresh witnesses for~$y_1 \ldots y_n$, naming the~$j^{th}$ value for~$y_i$
 in some canonical way
depending on~$(h(x_1), \ldots h(x_m), \dep, j, i)$, and creating the
  corresponding facts.
We let $I_\acc^*$ be $I_\acc$ extended with these facts.

The  second step is ``standard chasing with the original dependencies'': we chase $I_\acc^*$ in a standard way  in rounds with all
dependencies of~$\Sigma$, yielding a possibly infinite result. We let
$I_\acc^{+}$ be the result of extending $I_\acc^*$ by this chasing process. Note
that~$I_\acc^{+}$ then satisfies $\Sigma$ by definition of the chase.

We now construct $I_1^+ \colonequals I_1 \cup I_\acc^{+}$
and similarly define $I_2^+ \colonequals I_2 \cup I_\acc^*$.
  First observe that, for all $p \in \{1, 2\}$, we have ~$I_p
\subseteq I_p^+$, so that~$I_1^+$ still satisfies $Q$.
  Further,
  we argue that for all $p \in \{1, 2\}$, the instance~$I_p^+$ satisfies~$\Sigma$. 
As $\Sigma$ consists only of $\incd$s,
 its triggers consist of single facts, so it suffices to check this
on~$I_p$ and on~$I_\acc^+$ separately. For $I_\acc^+$, we know that it satisfies $\Sigma$ by
definition of the chase. For $I_p$, we know it satisfies $\Sigma'$, which is a superset
of~$\Sigma$, hence it satisfies $\Sigma$.
  
  We must now justify that, for all $p \in \{1, 2\}$, the instance~$I_p^+$
has a homomorphism $h$ to~$I_p$, which will imply that~$I_2^+$ still does not satisfy~$Q$. 
We
first define $h$ to be the identify on~$I_p$.
It then suffices to define $h$ as a homomorphism from~$I_\acc^{+}$ to~$I_p$
which is the identity on~$I_\acc$, because $I_\acc^{+} \cap I_p = I_\acc$. 
We next define $h$ on~$I_\acc^* \setminus I_\acc$. Consider a fact
  $F = R(\vec a)$ of~$I_\acc^* \setminus I_\acc$ created by obliviously chasing
  a trigger on an existence-check constraint $\dep$ on~$I_\acc$. 
Let $F' = S(\vec b)$ be the
fact of~$I_\acc$ in the image of the trigger: that is, the fact that matches the body of
$\dep$.
   We know that~$\dep$ holds in~$I_p$ and thus there is some
 fact  $F'' \colonequals R(\vec c)$ in~$I_p$
  that serves as a witness for this. Writing $\arity(R)$ to denote the
  arity of~$R$,
  we define $h(a_i)$
  for each $1 \leq i \leq \arity(R)$ as~$h(a_i) \colonequals c_i$. In this way, the image
of the fact $F$ under $h$ is  $F''$.
  This is consistent with the stipulation that~$h$ is the identity on~$I_\acc$,
  because whenever $a_i \in I_\acc$ then it must be exported between $F'$ and
  $F$, hence $a_i$ is also exported between $F'$ and $F''$ so we have $c_i =
  a_i$. Further, all these assignments are consistent across the facts of
  $I_\acc^* \setminus I_\acc$ because all elements of~$I_\acc^* \setminus
  I_\acc$ which do not occur in~$\dom(I_\acc)$ occur at exactly one position in
  one fact of~$I_\acc^* \setminus I_\acc$.

We now define $h$ on facts of~$I_\acc^{+} \setminus I_\acc^*$ by extending it
on the new elements introduced throughout the chase.
Whenever we create a fact $F = R(\vec a)$ in~$I_\acc^{+}$ for
a trigger mapping to~$F' = S(\vec b)$ for an $\incd$ $\dep$ in~$I_\acc^{+}$, we explain how to
extend $h$ to the nulls introduced in~$F$. Consider the fact $h(F')
= S(h(\vec
b))$ in~$I_p$. The body of~$\dep$ also matches this fact, and as~$I_p$ satisfies $\ids$
there must be a fact $F'' = R(\vec c)$ in~$I_p$ which extends this match to the
head of~$\dep$,
since  $\dep$ holds in~$I_p$.
 We define $h(a_i) \colonequals c_i$ for all $1 \leq i
\leq \arity(R)$. We show that this is consistent with the current definition of~$h$.
Whenever an element $a_i$ of~$\vec a$ already occurred in~$I_\acc^{+}$, it must
have been exported between $F'$ and $F$, so $h(a_i)$ was also exported between
$h(F')$ and $F''$, so we already have $h(a_i) = c_i$. Further, this assignment
is well-defined for the nulls introduced in~$F$, because each null occurs only at
one position.  The resulting $h$ is a homomorphism because the
image of previous facts is unchanged, and because the fact $R(h(\vec a)) = F''$ is
a fact of~$I_p$ as required.

This
concludes the proof of the fact that there is a homomorphism from~$I_\acc^{+}$ to
$I_\acc$ which is the identity on~$I_p$.

It remains to justify that the common subinstance $I_\acc^{+}$ in~$I_1^+$ and
$I_2^+$ is
access-valid in~$I_1^+$. Consider one access in~$I_1^+$ performed with some
method $\mt$ of a relation~$R$, with a binding $\accbind$ of values
in~$I_\acc^{+}$, and let us show that we can define a valid output to this
access in~$I_\acc^+$.
It is clear by definition of~$I_\acc^{+}$ that, if some value of~$\accbind$ is not
in the domain of~$I_\acc$, it must be a null introduced in the chase to create
$I_\acc^{+}$, in the first or in the second step.  In this case the only possible matching facts in~$I_1^+$ are
in~$I_\acc^{+} \setminus I_\acc$ and there is nothing to show. Hence,
we focus on  the case when all values of~$\accbind$ are in~$I_\acc$. If $\mt$ is not a
result-bounded access, then we can simply use the fact that~$I_\acc$ is access-valid
in~$I_1$ to know that all matching tuples in~$I_1$ were in~$I_\acc$, so the
matching tuples in~$I_1^+$ must be in~$I_\acc \cup (I_1^+ \setminus I_1)$, hence
in~$I_\acc^{+}$. If $\mt$ is a bounded access, then consider the access
on~$\mt'$ with the same binding. Either this access returns nothing or it tells
us that there is a fact $\checkview_\mt$ containing the values of~$\accbind$.
In the first case, as $I_\acc$ is access-valid in~$I_1$, we know that~$I_1$
contains no matching tuple, hence the constraints of $\aschema'$ imply
that~$I_1$ does not contain any $R$-fact which matches~$\accbind$ in the input
positions of~$\mt$. This means that any matching tuple in~$I_1^+$ for the access
on~$\mt$ must be in~$I_1^+ \setminus I_1$, so they are in~$I_\acc^+$ and we can
define a valid output to the access in~$I_\acc^+$. This covers the first case.

In the second case, the $\checkview_\mt$ fact of~$I_1$ implies by construction
that~$I_\acc^*$, hence $I_\acc^{+}$, contains infinitely many suitable
facts $R(\vec x, \vec y)$ with~$\vec x = \vec y$. Letting $k$ be the bound
of~$\mt$, we choose $k$ facts among those, and obtain a valid output to the
access with~$\accbind$ on~$\mt$ in~$I_1^+$. Hence, we have shown that~$I_\acc^{+}$ is
access-valid in~$I_1^+$.

The only remaining difficulty is that $I_1^+$, $I_2^+$, and $I_\acc^+$ are not
instances on the relational signature of~$\aschema$, because they still contain
$R_\mt$-facts. We simply build the final counterexample superinstance by
removing all facts that are not on a relation of~$\aschema$. The constraints
of~$\Sigma$ still hold, because they only mention relations of~$\aschema$. There
is still a homomorphism from $I_p^+$ to~$I_p$ for every $p \in \{1, 2\}$ as we
are removing facts from the left-hand side of the homomorphism. Further, it is
now the case that for all $p \in \{1, 2\}$, the restriction of~$I_p$ to
the relations of~$\aschema$ is a subinstance of~$I_p^+$, as claimed in the lemma
statement. Last, it is still the case that $I_\acc^+$ is a common subinstance
of~$I_1^+$ and~$I_2^+$ which is access-valid in~$I_1^+$ for~$\aschema$, as
$\aschema$ only has accesses on relations of its signature. Hence, the result of
this modification satisfies the conditions of Lemma~\ref{lem:enlarge}. Using
this lemma, we have completed the proof of Theorem~\ref{thm:simplifyidsexistence}.
\end{proof}

\subsection{Proof of Theorem~\ref{thm:fdsimplify}: FD
Simplification for~FDs}
\label{app:simplifyfds}

In this appendix, we show Theorem~\ref{thm:fdsimplify}. Recall its
statement:

\begin{quote}
  \fdsimplify
\end{quote}

We will again show the contrapositive of the statement.
Assume that we have a counterexample $I_1,I_2$ to
$\amd$ for
the FD simplification of $\aschema$, with $Q$ holding in $I_1$, with $Q$ not holding in $I_2$,
and with $I_1$ and $I_2$ having a common subinstance $I_\acc$ that is access-valid
  in~$I_1$ under the FD simplification of~$\aschema$.
We will  upgrade
these to $I^+_1, I^+_2, I_\acc^+$ having the same property for $\aschema$, by
blowing up accesses one after the other. To do so, we initially set $I_1^+
\colonequals I_1$, set $I_2^+ \colonequals I_2$, set $I_\acc^+ \colonequals
I_\acc$, and then we consider accesses one after the other.

  Consider each access $(\mt, \abind)$ using a method $\mt$ on relation $R$ with
  binding $\abind$ having values in $I_\acc$.
  Let $M_1$ be the matching tuples for $(\mt, \abind)$ in~$I_1$,
  and $M_2$ the matching tuples in~$I_2$. The definition of~$I_\acc$ and the
  constraints added in the FD simplification ensure that
  $M_1$ and $M_2$ must either intersect or be both empty. If $M_1$ and $M_2$ are
  both empty or if~$M_1$ is a singleton, then we do nothing for the access $(\mt,
  \abind)$: intuitively, we can already define a valid output to this access
  in~$I_1$ for~$\aschema$.
  Otherwise, we know that $M_1$ and $M_2$ are both non-empty.
  Let $k$ be the result bound
  of~$\mt$.
Recall that $\detby(\mt)$ denotes the positions determined under the FDs by
  the input positions of~$\mt$: the tuples of~$M_1$ and
  of~$M_2$ must agree on $\detby(\mt)$. Let~$X$
  be the other positions of~$R$ that are \emph{not} in~$\detby(\mt)$:
  the set $X$ must be non-empty, since otherwise $M_1$ and $M_2$ would both be
  singletons, contradicting our assumption.

  Let us blow up the output of $(\mt, \abind)$ in $I_1$ and $I_2$ by
  constructing
  $k$ tuples
  with all positions in $\detby(\mt)$ agreeing with the common value of the
  tuples of~$M_1$ and~$M_2$, and with all positions in~$X$ filled using fresh values
  that are disjoint from each other and from other values in $I_1 \cup I_2$.
  We then add these $k$ tuples to~$I_1^+$, to~$I_2^+$, and to~$I_\acc^+$.
  Performing this process for all accesses in~$I_\acc$ defines the
  final~$I_1^+$, $I_2^+$, and $I_\acc^+$.

  It is clear that $I_1 \subseteq I_1^+$ and that $I_2 \subseteq I_2^+$. We see that
  $I_\acc^+ \subseteq I_1^+$ and $I_\acc^+ \subseteq I_2^+$, because these two
  last inclusions are true initially and all tuples added to~$I_\acc^+$ are also
  added to~$I_1^+$ and~$I_2^+$. Further, for every $p \in \{1, 2\}$, we
  see that $I_p^+$ has a homomorphism back to~$I_p$, by defining it as the
  identity on~$I_2$, and mapping the fresh elements of every new tuple of~$I_p^+ \setminus I_p$ to
  the corresponding elements in some tuple of the (non-empty) set~$M_p$ considered at the point where
  the new tuple was added. 
This defines a homomorphism because the new tuple
  matches the tuples of~$M_2$ at the positions of~$\detby(\mt)$, and at other
  positions the new tuple contains fresh values occurring only at one position.

  We must justify that $I_1$ and~$I_2$ still satisfy the FD constraints
  of~$\aschema$. To do so, it suffices to consider each FD $\phi$ on relation~$R$, and to
  consider violations of~$\phi$ that involve the new tuples. If the
  left-hand-side of~$\phi$ contains a position of~$X$, then the freshness of
  the new values ensures that we have not added a violation of~$\phi$.
  Otherwise, the left-hand-side of~$\phi$ is contained in~$\detby(\mt)$, and the
  new tuples agree on these positions with existing tuples from~$M_1$ and~$M_2$,
  so we cannot have introduced a violation either. Hence, the FD constraints
  of~$\aschema$ still hold after adding the new tuples.

  We then explain why~$I^+_\acc$ is access-valid in~$I_1^+$. To do so, we will
  first define the notion of an access $(\mt', \accbind')$ \emph{extending}
  another access $(\mt, \accbind)$ if $\detby(\mt')$ is a superset of
  the input positions of~$\mt$, and if the restriction on~$\accbind'$ to the
  positions of~$\mt$ is exactly~$\accbind$.
  We make two claims:

  \begin{claim}
    \label{clm:extends1}
    Assume that, in the construction, we blow up an access $(\mt', \accbind')$
    that extends an access $(\mt, \accbind)$. Then we also blow up the
    access~$(\mt, \accbind)$ in the construction.
  \end{claim}

  \begin{proof}
    If we blew up $(\mt', \accbind')$ then it had more than one matching tuple in~$I_1$, and they
    are easily seen to be matching tuples for~$(\mt, \accbind)$ as well, so we
    also blow up $(\mt, \accbind)$.
  \end{proof}

  \begin{claim}
    \label{clm:extends2}
    Assume that, in the construction, when blowing up an access~$(\mt',
    \accbind')$, we add to~$I_1^+$ or to~$I_2^+$ some tuples that are also matching tuples for
    a different access $(\mt, \accbind)$. Then the access~$(\mt', \accbind')$
    extends the access~$(\mt,
    \accbind)$.
  \end{claim}

  \begin{proof}
    Consider a matching tuple $\vec{t}$ for $(\mt', \accbind')$ in~$I_1$: one must exist,
    because we are blowing up this access. The new tuples added in the blowup
    match $\vec{t}$ on the positions of~$\detby(\mt')$, and they contain fresh
    values at the other positions. Hence, for these tuples to be matching tuples
    for~$(\mt, \accbind)$ in~$I_1^+$ or in~$I_2^+$, then the input positions of~$\mt$ must be a subset of
    $\detby(\mt)$, and $\accbind$ must be the restriction of~$\accbind'$ to the
    input positions of~$\mt$, establishing the result.
  \end{proof}
    
  We can now prove that $I^+_\acc$ is access-valid in~$I_1^+$. Let us consider a method
  $\mt$ and binding $\accbind$. If $\accbind$ contains values
  from~$\dom(I^+_\acc) \setminus \dom(I_\acc)$. Thus we know that these values
  occur only in tuples from~$I^+_\acc\setminus I_\acc$, so we know that the
  matching tuples in~$I_1^+$ are all in~$I^+_\acc$ and there is nothing to show.
  Hence, we focus on the case where $\accbind$ consists of values
  of~$\dom(I_\acc)$. In this case, when we considered the access $(\mt,
  \accbind)$ in the blow-up process above, letting $M_1$ and $M_2$ be the matching
  tuples for the access in~$I_1$ and~$I_2$, either we did not perform the blowup
  or we did. If we performed the blowup, then we can define a
  valid result to the access in~$I_\acc^+$ using the $k$ tuples that we added
  in the blowup. If we did not, then by Claim~\ref{clm:extends1}, the
  construction has not blown up any access that extends $(\mt, \accbind)$
  either, so by the contrapositive of Claim~\ref{clm:extends2} we know that the
  matching tuples~$M_1^+$ for~$(\mt, \accbind)$ in~$I_1^+$  are exactly~$M_1$,
  and likewise $M_2^+ = M_2$ when defining~$M_2^+$ analogously.
  Now, if we did not perform the blowup for~$(\mt, \accbind)$, then either $M_1$
  and $M_2$ are both empty, or $M_1$ is a singleton.  We now know that the same
  is true of~$M_1^+$ and~$M_2^+$. If both $M_1$ and $M_2$
  are empty, then there are no matching tuples and the empty set is a valid
  output to the access. If $M_1^+$ is a singleton, then the single matching
  tuple is also a matching tuple in~$I_1$ for the access, so it must be part
  of~$I_\acc$ because~$I_\acc$ is access-valid in~$I_1$, and this defines a
  valid output to the access in~$I_1^+$. Hence, we have shown that $I_\acc^+$ is
  access-valid in~$I_1$ for~$\aschema$.

  The last step is to remove from $I_1^+$, $I_2^+$, and $I_\acc^+$ all facts of
  relations that are not in $\aschema$, i.e., the $R_\mt$ relations. It is still
  the case that $I_p^+$ has a homomorphism to~$I_p$ for all $p \in \{1, 2\}$, it
  is now the case that the restriction of $I_p$ to relations of~$\aschema$ is a
  subinstance of~$I_p^+$ for all $p \in \{1, 2\}$, the constraints of~$\aschema$
  are still satisfied because they do not mention the $R_\mt$ relations, and
  $I_\acc^+$ is still a common subinstance which is access-valid. Thus,
  Lemma~\ref{lem:enlarge} implies that $Q$ is not $\amd$ in~$\aschema$,
  concluding the proof.

\section{Proofs for Section~\lowercase{\ref{sec:complexity}}: Decidability of Monotone Answerability}
\subsection{Proof of Proposition~\ref{prop:decidmdetfdclassic}}
\label{app:decidmdetfdclassic}

Recall the statement of Proposition~\ref{prop:decidmdetfdclassic}:

\begin{quote}
\fdclassic
\end{quote}

  As mentioned in the body of the paper,  the lower bound already holds
  without result bounds or constraints \cite{access2}, so it suffices to show the upper
  bound.
   We also mentioned in the paper that  by  Theorem~\ref{thm:equiv}
  and Proposition~\ref{prop:reduce},
the problem reduces to the $\amd$ query
  containment problem $Q \subseteq_\Gamma Q'$ for~$\aschema$. As $\aschema$ has
  no result bounds, we can define $\Gamma$ using the rewriting of the accessibility axioms given
  after Proposition~\ref{prop:reduce}. 
  The constraints $\Gamma$
  thus consist of FDs and of full TGDs of the form:
\[\left(\bigwedge_i \accessible(x_i)\right) \wedge
  R(\vec x, \vec y) \rightarrow  R'(\vec x, \vec y) \wedge \bigwedge_i
  \accessible(y_i)\]

  As the TGDs are full, we know that
  we do not create fresh values when chasing.
Further, because there are no TGD constraints with primed relations in their
body, once $\accessible$ does not change within a chase round, the entire chase process has
terminated.
Besides, when adding values to $\accessible$ we must reach a fixpoint in linearly many rounds,
since $\accessible$ is unary. 
  Thus chasing with~$\Gamma$ terminates in linearly many rounds. Thus, we can decide
  containment by checking in~$\np$ whether $Q'$ holds on the chase result,
  concluding the proof.

\subsection{Proof of Theorem~\ref{thm:decidfd}: Complexity of Monotone Answerability for FDs}
\label{app:decidfd}

Recall the statement of Theorem~\ref{thm:decidfd}:

\begin{quote} 
  \thmdecidfd
\end{quote}

By Theorem~\ref{thm:fdsimplify} it suffices to deal with the  FD-simplification, meaning that we can 
reduce to a schema of the following form:

\begin{itemize}
  \simplifyfddef
\end{itemize}
By Proposition~\ref{prop:reduce}, we then reduce $\amd$ to query containment.
The resulting query containment problem
involves two copies of the constraints above, on primed and unprimed
copies of the schema, along with accessibility axioms for each
access method (including the new methods~$R_\mt$).
We can observe a few obvious simplifications of these constraints:
\begin{itemize}
\item In the chase, the constraint $R_\mt(\vec x, \vec y) \rightarrow \exists \vec z ~
  R(\vec x, \vec y, \vec z)$ will never fire, since a fact $R_\mt(\vec a, \vec b)$
is always generated by a corresponding fact  $R(\vec a, \vec b, \vec c)$. 
\item  In the chase, constraints of the form $R'(\vec x, \vec y, \vec z) \rightarrow R'_\mt(\vec x, \vec y)$
can fire, since it is possible that an $R'$-fact is created by one access method 
$\mt_1$ (result-bounded or not), but then an axiom of the above form is fired by
    a different access method $\mt_2$ on the same relation.
    However, such an $R'_\mt$-fact will not generate any further
    rule firings, and will not help make the query true (as it does not mention
    relations of the form $R'_\mt$), so we can disregard these constraints.
  \end{itemize}
If we consider the chase with the remaining constraints, we can
see that the only non-full TGDs are the primed copies of constraints in the
first bullet point above, namely constraints of the form:
\[
R'_\mt(\vec x, \vec y)
\rightarrow \exists \vec z ~
  R'(\vec x, \vec y, \vec z)
\]
Hence, these are the only rules that create new values, and these values
will never propagate back to the unprimed relations. Further,
whenever a primed fact $F$ is created containing  a null using the rule above,
the only further chase  steps that can apply to $F$ are FDs, and these will only
merge elements in $F$. 
Thus the chase will terminate in polynomially many parallel rounds 
as in the proof of Proposition~\ref{prop:decidmdetfdclassic} in Appendix~\ref{app:decidmdetfdclassic},
which establishes the NP upper bound and concludes the proof for unrestricted
monotone answerability.

\subsection{Proof of Proposition~\lowercase{\ref{prop:linearizeaccessids}}:
Linearization for Bounded-Width IDs and Truncated Accessibility Axioms}
\label{apx:linearizeaccessids}

In this appendix, we prove Proposition~\ref{prop:linearizeaccessids}. Recall the
statement:

\begin{quote}
  \linearizeaccessids
\end{quote}

To prove this, we will need to introduce some technical tools. 
First, we will need some details about \emph{truncated accessibility axioms},
and give a $\ptime$ implication algorithm for them assuming bounded \emph{breadth}.
Second, we will
present a notion of \emph{truncated chase proof}, which studies more closely the
structure of the chase by bounded-width IDs and truncated accessibility axioms,
and show that we can enforce a \emph{well-orderedness} property that specifies
in which order the dependencies are fired.
Third, we will present \emph{short-cut chase proofs}, where these dependencies
are fired in an even more specific order, and show that this definition of the
chase is still complete.
Last, we will use these tools to prove
Proposition~\ref{prop:linearizeaccessids}.

\subsubsection{Details about Truncated Accessibility Axioms} \label{subsubsec:detailstruncated}

We call \emph{truncated accessibility axiom} any TGD of the
following form on $\sign \cup \{\accessible\}$:
\[\left(\bigwedge_{i \in P} \accessible(x_i)\right) \wedge R(\vec x) \rightarrow
\accessible(x_j)\]
where $R$ is a relation and~$P$ is a subset of the positions of~$R$. 
Notice the similarity with axioms of the form (Truncated Accessibility) as
introduced in the main text: the only difference is that we have rewritten them
further to ensure that the head always contains a single accessibility fact.

Intuitively, such an axiom tells us that, when a subset of the elements of
an~$R$-fact are accessible, then another element of the fact becomes accessible
(by performing an access). An \emph{original truncated accessibility axiom}
is a truncated accessibility axiom which is in the set~$\Delta$ that we obtained in
the reduction to query containment. For these axioms, the set $P$ is the set of input
positions of some method~$\mt$ on~$R$.  We will study  truncated
accessibility axioms that 
are \emph{implied} by the original
truncated accessibility axioms in~$\Delta$ and by the constraints in~$\Sigma$.
We call them the \emph{derived truncated accessibility axioms}.

There can be exponentially many truncated accessibility axioms, but we will not
need to compute all of them: it will suffice to compute those of small
\emph{breadth}. Formally,
the \emph{breadth} of a truncated accessibility axiom is  the size of~$P$.
Note that the  number of possible truncated accessibility axioms of  breadth
$b$ is at most $r \cdot a^{b+1}$, where $r$ is the number of relations
in the signature and  $a$ is the maximal arity of a relation. We show that we
can efficiently
compute the derived truncated accessibility axioms of a given breadth, by
introducing a \emph{truncated accessibility axiom saturation algorithm}.

The algorithm iteratively builds up a set $O$ of triples $(R,\vec p,j)$ with~$\vec p$ a
set of positions of~$R$ of size at most $w$ and $j$ a position of~$R$. Each such
triple represents the following truncated accessibility axiom of breadth $\leq
w$:
\[
  \left(\bigwedge_{i \in \vec{p}} \accessible(x_i)\right) \wedge R(\vec x) \rightarrow
\accessible(x_j)
\]

We first set $O \colonequals \{(R, \vec p, j) \mid j \in \vec p\}$, representing
trivial axioms.
We then repeat the
 steps below:
\begin{itemize}
  \item ($\incd$): If we have an $\incd$ from~$R(\vec x)$ to~$S(\vec y)$, that
    exports the variables
    $x_{j_1}, \dots, x_{j_{m'}}, x_j$ to $y_{k_1}, \dots, y_{k_{m'}}, y_k$,
     and if we have $(S,({k_1} \ldots {k_{m'}}),k) \in O$ for some $k_1 \ldots k_{m'}$
     then we add the tuple $(R,(j_1 \ldots j_{m'}),j)$ to~$O$.
   \item (Transitivity): If there exists a relation $R$, a set of positions
     $\vec p$ of~$R$,  and a set
     of positions $\{t_1 \ldots t_m\}$ of~$R$ with
  $m\leq w$ such that we have $(R, \vec p, t_i) \in O$ for all $1 \leq i \leq m$,
  and we have $(R, \vec r, t') \in O$ with
$\vec r \subseteq \vec p \cup \{t_1 \ldots t_m\}$, then we add $(R, \vec p, t')$
to~$O$.
\item (Access): If we have a method $\mt$ on~$R$ with input positions $j_1 \ldots j_m$
and a set $\vec p$ of at most $w$ positions such that  $(R,\vec p, j_i) \in O$ for all
$1 \leq i \leq m$, then we add $(R,\vec p, j)$ to~$O$ for all $j$ between~$1$
  and the arity of~$R$.
\end{itemize}
We continue until we reach a fixpoint.

Note that  a fixpoint must occur after at most
$r \cdot a^{w+1}$ steps, with~$r$ the number of relations in the schema and $a$
  the maximal arity of a relation. It is clear that the algorithm runs
  in polynomial time in~$\Sigma$ and in the set of access methods. 
We will
  show that this correctly computes all derived truncated accessibility
  axioms satisfying the breadth bound:

\newcommand{\derivedaccessids}{
For any fixed $w \in \NN$,
the truncated accessibility saturation algorithm
 computes all derived truncated accessibility axioms of breadth
  at most~$w$, when given as input a set of~$\incd$s of width~$w$ and a set of truncated accessibility
axioms.
}
\begin{proposition} \label{prop:derivedaccessids}
  \derivedaccessids
\end{proposition}

We defer the proof of this result until we establish some results about the chase with
these axioms.
\subsubsection{Truncated Chase Proofs and Well-Orderedness}
\label{app:orderability-accessids}
Towards our goal of showing the correctness of the saturation algorithm,
we now present an ordering result about \emph{truncated
chase proofs}, that is, proofs using $\incd$s 
and truncated accessibility axioms. 
In any such proof, we can arrange the facts that we create in a tree.
Each node $n$ of the tree corresponds to a fact $F$ that is generated by an
$\incd$, and the parent of
$n$ is the node associated to the fact contained in the trigger that was
fired to generate~$F$.
During the proof, we also generate additional accessibility facts $A$ by firing
truncated accessibility axioms, and the trigger for the firing involves a fact
$F$ over the original schema (i.e., not an accessibility fact), as well as other
accessibility facts. We then call $F$
the \emph{birth fact} of
the accessibility fact $A$, and 
the \emph{birth constants} of~$A$ are all constants $d$ such that~$\accessible(d)$
is a hypothesis of the truncated accessibility axiom creating $A$.
Our main goal will be to normalize proofs so that the creation of accessibility facts
is ``compatible with the tree structure''.
Consider a truncated chase proof that results in a chase instance $I$. Such
a proof is  \emph{well-ordered} if it has the following property:

\begin{quote}
  For any fact $F=R(\vec c)$ generated in~$I$ by firing a trigger $\trig$
  for an $\incd$,
 if $\accessible(c_i)$ is generated in~$I$ with birth fact in the subtree of~$F$,
 and all the birth constants $c_{m_1}\ldots c_{m_k}$
  of~$c_i$ were exported when firing $\trig$,
  then each fact $\accessible(c_{m_1}) \ldots \accessible(c_{m_k})$ must  already
 have been present  in the chase at the time $F$ was generated.
\end{quote}
We now show:

\begin{lemma} \label{lem:wellaccessids}
For any chase proof from the canonical database
of~$Q$ using truncated accessibility axioms and $\incd$s,  producing
instance $I$, there is a well-ordered chase proof from the canonical database
  of~$Q$ that generates a set of 
facts isomorphic to those of~$I$.
\end{lemma}
\begin{proof}
Note that, in an arbitrary proof, it could well be
  that~$A_j=\accessible(c_{m_j})$ is generated after 
the generation of~$F$. 
 The idea of the proof is that we can ``re-generate
$F$'', re-firing the rules generating $F$ and its subtree after all such facts $A_j$
  are created.

Formally, we proceed by induction  on the number of counterexample firings.
In the inductive step, consider a non-well-ordered proof and the 
subproof~$f_1 \ldots f_k$ up through the first violation
of well-orderedness.
That is, there is a fact $F=R(\vec c)$ generated by a rule firing $f_i$ using an
  $\incd$ $\dep$ from its parent fact
$E$,  a fact  $A_j=\accessible(c_{m_j})$ that was not present in the chase
when $F$ was generated, and  $f_k$ is an accessibility axiom  using
  $\accessible(c_{m_j})$ (and possibly other
accessibility facts)
to generate $\accessible(c_i)$ with birth fact $F_B$ in the subtree of~$F$.
 We create a new proof that
  begins with~$f_1 \ldots f_{k-1}$
  and then  continues by ``copying $f_i$'',
generating a copy $F'$ from~$E$ via~$\dep$.
  Doing this cannot introduce a violation
  of well-orderedness, because it does not generate an accessibility fact, and
  there are no accessibility facts in the subtree of the new fact~$F'$. 

  We now continue the proof with a copy of
the firings $f_{i+1} \ldots \allowbreak f_{k-1}$,
  but the firings that were performed in the
  subtree of~$F$ are now performed instead on the corresponding node in the
  subtree of~$F'$.
  When we perform the copy of these firings, we know that we do not
  cause any violation of well-orderedness, because the original firings $f_{i+1}
  \ldots f_{k-1}$ did not cause such a violation (by minimality of~$f_k$).

  Last, instead of firing $f_k$ on the fact $F_B$ in the subtree of~$F$, we fire
  it on the corresponding fact $F_B'$ in the subtree of~$F'$:
  we call this rule firing $f_k'$. We argue that all the necessary accessibility
  hypotheses for~$f_k'$ have
  been generated, so that we can indeed fire~$f_k'$. Indeed, for the
  accessibility hypotheses of~$f_k$ that have been created in the subtree
  of~$F$, we know that these hypotheses had been
  generated by firing $f_{i+1} \ldots f_{k-1}$, so these
  the corresponding hypotheses of~$f_k'$ have also been generated in the subtree of~$F'$.
  Now, for the accessibility hypotheses of~$f_k$ that are on
  exported elements between $E$ and~$F$, they had been generated already
  when we wanted to fire~$f_k$, so they are generated when we want to
  fire~$f_k'$. 
  In fact, our construction has ensured that these accessibility hypotheses had
  already been generated when creating~$F'$, which ensures that we can
  fire~$f_k'$ and not cause a violation of well-orderedness.
Hence, the proof that we have obtained by this process generates
$\accessible(c_i)$ in a well-ordered way, and the number of violations of
  well-orderedness has decreased.
\end{proof}

\subsubsection{Proof of Proposition~\ref{prop:derivedaccessids}}
Using the well-ordered chase, we are now ready to complete the proof of  Proposition \ref{prop:derivedaccessids},
which stated that the Truncated accessibility axiom saturation algorithm generates
exactly the derived truncated axioms of a given breadth:

\begin{proof}
For one direction, it is straightforward to see that all rules  obtained by this process
are in fact derived truncated accessibility axioms.
Conversely, we claim that, for all derived truncated accessibility axioms of
  breadth $\leq w$ 
\[ \accessible(c_{s_1}) \wedge \ldots \wedge \accessible(c_{s_l}) \wedge
R(\vec x) \rightarrow \accessible(c_i),
\]
  then the corresponding triple $(R, (s_1 \ldots s_l), i)$ is added to~$O$.

We prove this by induction on the length of a chase proof
of the accessibility fact~$\accessible(c_i)$ from the hypotheses $R(\vec c)$ and the accessibility
  facts $\accessible(c_{s_j})$ for $1 \leq j \leq l$ (with $l \leq w$).
Note that by Lemma~\ref{lem:wellaccessids} we can assume that the proof is well-ordered.

If the proof is trivial, then clearly $(R, \vec p, i) \in O$ by the initialization of~$O$.
If it is non-trivial then some accessibility axiom fired to produce $\accessible(c_i)$,
and we can fix a guard atom $F$ and accessibility facts $F_1 \ldots F_l$ that were hypotheses
of the firing: following our earlier terminology, $F$ is the birth fact
  of~$\accessible(c_i)$ and the constants occurring in the~$F_1 \ldots F_l$ are
  the birth constants of~$\accessible(c_i)$.
  If $F=S(\vec c')$ with  $\vec c'$ a subset of~$\vec c$, then each $F_i$
is of the form $\accessible(c_{s_i})$, and by induction
$(R, \vec p, c_{s_i}) \in O$ for each $i$.  Now by (Transitivity) and (Access) we complete the
argument.

Otherwise, the guard $F=S(\vec a, \vec d)$ of the accessibility axiom firing
  was generated by firing an $\incd$ $\dep$ to some fact $E_1=T_1(\vec a, \vec b)$, with~$\vec a$  the subset
of the values in~$E_1$ that were exported when firing~$\dep$.
  By well-orderedness, we know that each accessibility fact 
used in the firing that mentions a value in~$\vec a$ was present when $\dep$
  was fired
  on~$E_1$: as the width of the $\incd$s is~$w$, this set has width at most~$w$.
  Now, we see
that there is a subproof of shorter length proving $\accessible(c_{i})$ 
  from~$F$ and this subset of~$F_1 \ldots F_l$.
Therefore by induction we have $(S, \vec p', i') \in O$ for~$\vec p'$
  corresponding to the subset above (of size at most~$w$, so matching the
  breadth bound)
and $i'$ corresponding to~$c_{i}$ in~$F$. Applying the rule ($\incd$) we have $(T_1,\vec
s'', i'') \in O$ for~$i''$ corresponding to~$c_i$ in~$E_1$ and $\vec p''$
corresponding to the subset in~$E_1$.
The fact $E_1$ may itself have been generated by a non-full $\incd$ applied to some $E_2$, and
hence may contain values that are not in the original set of constants $\vec c$.
But if so we can iterate the above process on the $\incd$ from~$E_2$ to~$E_1$, noting that~$E_2$ also must
contain~$c_i$. Hence, by iterating this process,
  we arrive at a triple $(T_n, \vec p_n, i_n)$ which is in~$O$,
  where $i_n$ corresponds to the position
of~$c_i$ in a fact $F_n$ that occurs in the original proof with no application
  of an $\incd$. In other words, we must have
$F_n=R(\vec c)$, and hence $T_n=R$ and $i_n=i$. By induction again, we have
$(R, \vec p, j) \in O$ for each $j \in \vec p_n$. Applying (Transitivity) completes the argument.
\end{proof}

We have shown that we can compute in $\ptime$ the implication closure of truncated
accessibility axioms of bounded breadth under bounded-width $\incd$s. We will use
this implication closure in the construction of~$\Sigma^\lift$ to show
Proposition~\ref{prop:linearizeaccessids}, but we first need to introduce the
notion of \emph{short-cut chase}.

\subsubsection{Short-Cut Chase and Completeness}

We now state a further proof normalization result: instead of chasing with truncated
accessibility axioms, we can create the same facts by firing derived axioms of small breadth in a ``greedy
fashion''.
Recall that $\Sigma$ consists of $\incd$s of width $w$, and let us write
$\Delta'$ for the set of truncated
accessibility axioms that we consider. Remember that we can use
Proposition~\ref{prop:derivedaccessids} to compute in PTIME all derived truncated accessibility
axioms from~$\Delta'$ and~$\Sigma$ of breath at most~$w$.

A \emph{short-cut chase proof} on an initial instance $I_0$ with~$\Sigma$
and~$\Delta'$ uses
two alternating kinds of steps:
\newcommand{\shortcutsteps}{
  \item  \emph{$\incd$ steps}, where we fire
an $\incd$ on a trigger $\trig$ to generate a fact $F$: we put $F$ in a new node~$n$
    which is a child of the node~$n'$ containing the fact of~$\trig$; and we
    copy in~$n$ all facts of the form $\accessible(c)$ that held in~$n'$
about any element $c$ that was exported when firing~$\trig$.
  \item \emph{Breadth-bounded saturation steps}, where we consider a newly created node~$n$ and 
apply all derived truncated accessibility axioms of breadth at most~$w$ on that
    node until we reach a fixpoint and there are no more violations of these
    axioms on~$n$.
}
\mylistskip
\begin{compactitem}
  \shortcutsteps
\end{compactitem}
\mylistskip
We continue
this process until a fixpoint is reached.
\newcommand{\shortcutexpl}{%
  The atoms in the proof are thus associated with a tree structure:
it is a tree of nodes that correspond to the application of $\incd$s, and each node
also contains accessibility facts that occur in the node where they were
generated and 
in the descendants of those nodes that contain
facts to  which the elements are exported.}%
\shortcutexpl{}
The name ``short-cut'' intuitively indicates that
we short-cut certain derivations that could have been performed by moving up and
down in the chase tree: instead, we apply a derived
truncated accessibility axiom.

\newcommand{\normalizationaccessids}{
  For any set $\Sigma$ of $\incd$s of width~$w$,
  given a set of facts $I_0$ and a chase proof using $\Sigma$
that produces $I$, 
  letting $I^+_0$ be the closure of~$I_0$ under the original and derived
  truncated accessibility axioms in~$\Delta'$,
there is  $I'$ produced by a short-cut
chase proof from~$I^+_0$ with~$\Sigma$ and~$\Delta'$ such that there is a homomorphism from~$I$ to~$I'$.
}
\begin{lemma} \label{lem:normalizationaccessids}
\normalizationaccessids
\end{lemma}

To prove this lemma,
we start with an observation about the closure properties of short-cut chase proofs.
\begin{lemma} \label{lem:normalizedclosureaccessids}
  Let $I_0^+$ be an initial instance closed under the derived and original truncated accessibility
  axioms, and suppose that a short-cut chase proof
has a breadth-bounded saturation step producing a fact $G=\accessible(c_i)$.
Then $c_i$ is not an element of~$I_0^+$, and the node associated with the
  breath-bounded saturation step
was created by the $\incd$-step where $c_i$ is generated.
\end{lemma}
\begin{proof}
We first consider the case where $c_i$ is not in~$I_0^+$, so it is 
  a null introduced in a fact $E=R(\vec c)$
  that was created by an $\incd$ trigger $\trig$. 
Let $n$ be the node of~$E$, and let $S=\accessible(c_{j_1})  \ldots \allowbreak \accessible(c_{j_l})$ be the set of
  accessibility facts that were true of the~$c_i$ when firing~$\trig$: the facts
  of~$S$ are
  present in~$E$.
Note that~$S$ has size at most $w$, since all but $w$ elements were fresh in~$E$
when the $\incd$ was fired. The node~$n$ must be an ancestor of the node
where $\accessible(c_i)$ is generated, because $n$ is an ancestor of all nodes
  where $c_i$ appears. Thus $G=\accessible(c_i)$  
is a consequence of~$E$ and the hypotheses $S$  under the constraints,
since it is generated via derived truncated accessibility axioms or constraints
  in~$\Sigma$. But then we know that
\[
  R(\vec x) \wedge  \left(\bigwedge_{k \leq l} \accessible(x_{j_k})\right) \rightarrow \accessible(x_i)
\]
is a derived truncated accessibility axiom
and it has breadth at most~$w$.  Hence, this axiom applied to generate
  $G$ from~$\{E\} \cup S$ when applying the breadth-bounded saturation step
  to~$E$, and indeed $G$ was created in the node $n$ where $c_i$ was generated.

  We now argue that $c_i$ cannot be in~$I_0^+$. Assuming to the contrary that
  it is, we know that the saturation
  step that produced~$G$ must have applied to a node which is not the root,
  as~$I_0^+$ is closed under the derived and original truncated accessibility
  axioms.
  We can assume that the depth
of the node~$n$ where $G$ is generated is minimal among all such counterexamples.
Then $G$ is generated at a node~$n$ corresponding to the firing of an $\incd$ from 
a node~$E$ to a node~$F$. But then arguing as above, $G$ must already follow  from
$E$ and the accessibility hypotheses that were present when the $\incd$ was fired, of which
  there are at most~$w$. Thus $G$
would have been derived in the breadth-bounded saturation step that followed $E$, which
  contradicts the minimality of~$n$.
\end{proof}

We now are ready to complete the proof of Lemma~\ref{lem:normalizationaccessids}:
\begin{proof}[\myproof of Lemma~\ref{lem:normalizationaccessids}]
We can extend $I$ to a full chase instance (possibly infinite), denoted $I_\infty$. Likewise,
we can continue the short-cut chase process indefinitely, letting $I'_\infty$ be the resulting facts.
It is clear that $I'_\infty$ satisfies the constraints of~$\Sigma$, and
   we claim that~$I'_\infty$ also satisfies the constraints of~$\Delta'$. 
Assume by contradiction that there is an
  active trigger in~$I'_\infty$: it is
  a trigger for an original truncated accessibility axiom in~$\Delta'$, with
  facts $\left(\bigwedge \accessible(c_{m_j})\right) \wedge R(\vec c)$, whose firing would have produced fact $\accessible(c_i)$.
Consider the node~$n$ where $R(\vec c)$ occurs in the short-cut chase proof. 
If $n$ is the root node corresponding to~$I_0^+$, then we know by
  Lemma~\ref{lem:normalizedclosureaccessids} that any accessibility facts on elements
  of~$I_0^+$ must have been generated in~$I_0^+$, i.e., must have been
  already present there, because $I_0^+$ is already saturated; hence, we
  conclude that the trigger is in~$I_0^+$, hence it is not active because
  $I_0^+$ is closed under the original truncated accessibility axioms. Hence,
  $n$ is not the root node.

  Now, if the node~$n$ is not the root, then by
  Lemma~\ref{lem:normalizedclosureaccessids}, each fact
$\accessible(c_{m_j})$ must have been present at the time $R(\vec c)$ was generated
  Hence, the breadth-bounded saturation step at~$n$ should
  have resolved the trigger, so we have a contradiction.

Since instance $I'_\infty$ satisfies  the constraints, there is
a homomorphism $h$ from the full infinite chase $I_\infty$ to that instance, by universality
of the chase \cite{fagindataex}. Letting $I'$ be the image of~$I$, we get the
  desired conclusion.
\end{proof}

\subsubsection{Concluding the Proof of Proposition~\ref{prop:linearizeaccessids}}
\label{app:prflinearizeaccessids}

We now present our definition of the set of $\incd$s $\Sigma^{\lift}$ that will simulate the
chase by~$\Sigma$ and~$\Delta$. Thanks to what precedes
(Lemma~\ref{lem:normalizationaccessids}), we know that it suffices to simulate
the short-cut chase.

We start by calling $\Delta^+$ the set of derived
truncated accessibility axioms calculated using
Proposition~\ref{prop:derivedaccessids} on~$\Sigma$ and~$\Delta$.
To define these axioms, when considering a relation $R$, a subset $P$ of the
positions of~$R$, and a position $j$ of~$R$, we will say that $P$
\emph{transfers} $j$ if $\Delta^+$ contains the following derived truncated
accessibility axiom: \[\left(\bigwedge_{i\in P} \accessible(x_i)\right) \wedge
R(\vec x) \rightarrow \accessible(x_j).\]
The set of positions $P'$ of~$R$ transferred by~$P$ is then the set of positions~$j$
such that $P$ transfers $j$. In particular, note that we always have $P'
\supseteq P$.

We now define $\Sigma^{\lift}$ as follows:

\begin{itemize}
  \item (Transfer): Consider a relation $R$, and a subset $P$ of positions of~$R$ of
    size at most~$w$. Let $P'$ be the set of positions transferred by~$P$.
    If $P'$ contains the set of input positions of some access method on~$R$,
    then 
    we add the full $\incd$:
    \[
      R_P(\vec x) \rightarrow R'(\vec x)
    \]

  \item (Lift):
  Consider an $\incd$  $\dep$ of~$\Sigma$, 
  \[R(\vec u) \rightarrow \exists \vec z ~ S(\vec z,\vec
  u),\]
    For every subset $P$ of positions of~$R$ of size at most~$P$, we let $P'$
    be the set of positions transferred by $P$.
We let $P''$ be the intersection
    of~$P'$ with the exported positions in the body of~$\dep$, and we let $P'''$
    be the subset of the exported positions in the head of~$\dep$ that
    corresponds to~$P''$.  Then we add the dependency:
    \[R_P(\vec u) \rightarrow \exists \vec z ~ S_{P''}(\vec z,\vec u)\]
\end{itemize}
We also need to describe the effect of~$\Sigma$ and $\Delta$ when we start the
chase.
We recall that $\sign$ denotes the signature of the schema, and that the
constraints of $\Sigma$ are expressed on~$\sign$,
that the constraints $\Sigma'$ are expressed on a primed copy $\sign'$
of~$\sign$, and that $\Delta$ is expressed on~$\sign$, $\sign'$, and the unary
relation~$\accessible$.
Given a CQ~$Q$, let $I_0 \colonequals \canondb(Q)$ be its canonical database, and
let $I_0^\lift$ be formed by adding atoms to $I_0$ as follows. 

\begin{itemize}
\item Apply all of the truncated accessibility axioms of~$\Delta^+$
  to~$I_0$ to obtain~$I_0'$.
\item  For any relation~$R$ of the signature~$\sign$, and
  for every fact $R(a_1 \ldots a_n)$ of~$I_0'$, let $P$ be the
  set of the~$i \in \{1 \ldots n\}$ such that~$\accessible(a_i)$ holds
  in~$I_0'$.  For every $P' \subseteq P$ of size at most~$w$, 
  add to~$I_0^\lift$ the fact $R_{P'}(a_1 \ldots a_n)$. Further, in the case
    where
    $\accessible(a_i)$ holds for each $1 \leq i \leq n$, then we add the fact $R'(a_1
    \ldots a_n)$.
\end{itemize}

It is now easy to see that $\Sigma^\lift$ and $I_0^\lift$
satisfy the required conditions: for every set of primed facts~$I$ derivable
from $I_0$ by chasing with $\Sigma$ and~$\Delta$, we can derive the same set of
primed facts from~$I_0^\lift$ by chasing with $\Sigma^\lift$.
Indeed, chasing with the Lift rules creates a tree of facts that
corresponds exactly to a short-cut chase proof: when we create an $R_P$-fact,
the $P$ subscript denotes exactly the set of positions of the new facts that
contains exported elements that are accessible. It is then easy to see that the
(Transfer) rules creates primed facts exactly for facts that can be transferred
by applying some method.

The only thing left to do is to notice that $\Sigma^\lift$ has bounded semi-width,
but this is because the rules (Lift) have bounded width
and the rules (Transfer) clearly have an acyclic position graph. This
concludes the proof.

\subsection{Proof of Proposition~\ref{prop:semiwidthclassic}: Containment under
IDs of Bounded Semi-Width}
\label{apx:semiwidthclassic}

We recall the statement of 
Proposition~\ref{prop:semiwidthclassic}.

\medskip

\semiwidthclassic

\medskip

We recall the definition of semi-width from the body, generalizing it slightly from $\incd$s.
The \emph{basic position graph} of a set of TGDs
$\Sigma$ is the directed graph whose nodes are the positions of relations
in~$\Sigma$,  with an edge from~$R[i]$ to~$S[j]$ if and only
if there is a dependency $\dep \in \Sigma$ with exported
variable $x$ occurring in position~$i$ of~$R$ in the body
of~$\dep$ and position~$j$ of~$S$ in the head of~$\dep$.
We say that a collection of TGDs $\Sigma$ has \emph{semi-width} bounded
by $w$ if it can be decomposed into~$\Sigma_1 \cup \Sigma_2$
where $\Sigma_1$ has width bounded by $w$ and
the position graph of~$\Sigma_2$ is acyclic.

Consider  a chase sequence based on the canonical database $I_0$ of a conjunctive query
$Q$, using  a collection of $\incd$s $\Sigma$.
The collection of facts generated by this sequence can be given the structure of a tree, where there is a root node
associated with~$I_0$, and one node~$n_F$ for each generated fact $F$.
If performing a chase step on fact $F$ produces fact $F'$ in the sequence, then
the node~$n_{F'}$ is a child of the node~$n_F$. We refer to this as the \emph{chase tree} of the sequence.

Consider nodes $n$ and $n'$ in the chase tree, with~$n$ a strict ancestor of~$n'$.
We say they $n$ and $n'$ are \emph{far apart}
 if there are distinct generated facts $F_1$ and $F_2$ such that the node $n_1$
 corresponding to~$F_1$ and the node~$n_2$ corresponding to~$F_2$ are both ancestors
 of~$n'$ and descendants of~$n$, if $n_1$ is an ancestor of~$n_2$,
if $F_1$ and $F_2$ were generated by the same rule of~$\Sigma$,
and if any value of~$F_1$ which occurs in~$F_2$ occurs
in the same positions within  $F_2$ as in~$F_1$.
If such an $n$ and $n'$ are not far apart, we say
that are \emph{near}.

Given a match $h$ of~$Q$ in the chase tree, its
\emph{augmented image} is the closure of its image
under least common ancestors. If $Q$ has size $k$ then
this has size $\leq 2k$.
For nodes $n_1$ and $n_2$ in the augmented image,
we call $n_1$ the \emph{image parent of} $n_2$ 
if $n_1$ is the lowest  ancestor
of~$n_2$ in the augmented image.

The analysis of Johnson and Klug is based on the following lemma:
\begin{lemma} If $Q$ has a match in the chase, then there is 
a match $h$ with the property that if
$n_1$  is the image parent of~$n_2$
then $n_1$ and $n_2$ are near.
\end{lemma}
\begin{proof}
We prove this by induction on the number of violating $n_2$'s
 and
the sum of the depths of the violations in the tree.
If $n_1$ is far apart from~$n_2$, then there are witnesses
$F_1$ and $F_2$ to this.
   We eliminate the interval between $F_1$
and $F_2$ (along with the subtrees hanging off of them, which
by assumption do not contain any match elements).
We  adjust  $h$ accordingly.
In doing this we reduce the sum of the depths, while no new violations 
are created, since the image parent relationships  are preserved.
Iterating this operation
we must achieve a tree where the nodes corresponding
to~$n_1$ and $n_2$ are  near 
 and thus the number of violations decreases.
\end{proof}

Call a match $h$ of~$Q$ in the chase \emph{tight} if it
has the property given in the lemma above. The \emph{depth}
of the match is the depth of the lowest node in its image.
The next observation, also due to Johnson and Klug, is that when the width
is bounded, tight matches can not occur  far down in the tree:

\begin{lemma} \label{lem:depthbound} If $\Sigma$ is a set of $\incd$s of width $w$
and the schema has arity bounded by $m$,
then a tight match  of size $k$  has depth at most
$k \cdot |\Sigma| \cdot m^{w+1}$.
\end{lemma}
\begin{proof}
We claim that the length of the path  between a match element $h(x)$
and its image parent $h(x')$
 must be at most $|\Sigma| \cdot m^{w+1}$.
At most $w$ values from~$h(x')$ are present
in any fact on the path,  and thus the number of configurations
that can occur is at most $m^{w+1}$.
  Thus after
 $|\Sigma| \cdot m^{w+1}$ there will be two elements which
repeat both the rule and the configuration of the values, which
would contradict tightness.
\end{proof}

Johnson and Klug's result follows from combining the previous two lemmas:
\begin{proposition}[\cite{johnsonklug}]
  \label{prop:jkwidth}
For any fixed  $w \in \NN$, 
there is an $\np$ algorithm for query containment under  $\incd$s of
width at most $w$.
\end{proposition}
\begin{proof}
We guess $k$ branches of depth at most
$k \cdot |\Sigma| \cdot m^{w+1}$ in the chase and
a match in them.
\end{proof}

We now give the  extension of this calculation for
bounded semi-width.

Recall from the body that a collection of $\incd$s $\Sigma$ has \emph{semi-width} bounded
by $w$ if it can be decomposed into~$\Sigma_1 \cup \Sigma_2$
where $\Sigma_1$ has width bounded by $w$ and
the position graph of~$\Sigma_2$ is acyclic.

An easy modification of Proposition~\ref{prop:jkwidth}
now completes the proof of Proposition
\ref{prop:semiwidthclassic}:
\begin{proof}
We revisit the argument of Lemma~\ref{lem:depthbound}.
As in that argument, it suffices
to show that the length of the path  between a match element $h(x)$
and its closest ancestor $h(x')$ in the image
 must be at most $|\Sigma| \cdot m^{w+1}$.
As soon as we apply a rule of~$\Sigma_1$ along
the path, at most $w$ values are exported, and so the
remaining path is bounded as before.
Since $\Sigma_2$ has an acyclic position graph, 
a value in~$h(x')$ can propagate for at most $|\Sigma_2|$ steps 
when using rules of~$\Sigma_2$ only. Thus after
at most $|\Sigma_2|$  edges in a path
we will either have no values propagated (if we used only
rules from~$\Sigma_2$) or at most $w$ values (if we used
  a rule from~$\Sigma_1$). Thus we can bound the path size
  by the previous bound plus a factor of~$|\Sigma_2|$.
\end{proof}

We will need a slight strengthening of
this result in Appendix~\ref{app:complexitygeneral},
which works with \emph{linear TGDs} rather than IDs: these are
TGDs with a single atom in the body, but allowing
repetition of variables in either body or head.

Our strengthened  result is:

\newcommand{\semiwidthclassicgeneral}{%
For fixed  $w$,
there is an $\np$ algorithm for query 
containment under linear TGDs of
semi-width at most~$w$.
}%
\begin{proposition} \label{prop:semiwidthclassic-general}
  \semiwidthclassicgeneral
\end{proposition}

The proposition is proven exactly as in the case of $\incd$s, defining the chase
tree for linear TGDs analogously as how we defined it for $\incd$s.

\subsection{Proof of Theorem~\ref{thm:npidsbounds}: Complexity of Monotone
Answerability for Bounded-Width IDs}
\label{apx:npidsbounds}

We now prove Theorem~\ref{thm:npidsbounds}. 
Recall the statement:

\begin{quote}
  \npidsbounds
\end{quote}

In the main text, we have only sketched the proof in the case without result
bounds. We first complete the proof in the case without result bounds, and then
extend it to support result bounds.

\subsubsection{Proving Theorem~\ref{thm:npidsbounds} without Result Bounds}

Recall from the body of the paper that, in the absence of result bounds, the containment for~$\amd$ is
$Q \subseteq_\Gamma Q'$, where 
$\Gamma$ consists of the bounded-width $\incd$s $\Sigma$, their primed copy
$\Sigma'$, and $\Delta$ which includes, for each access method $\mt$ on a relation
$R$ with input positions $\vec x$:
\begin{itemize}
  \item (Truncated Accessibility): $\left(\bigwedge_i \accessible(x_i)\right) \wedge
  R(\vec x, \vec y) \rightarrow \bigwedge_i
  \accessible(y_i)$
\item (Transfer): $\left(\bigwedge_i \accessible(x_i)\right) \wedge R(\vec x, \vec y)
  \rightarrow  R'(\vec x, \vec y)$
\end{itemize}

Recall the proof of Proposition~\ref{prop:linearizeaccessids}, and the
definition of $\Sigma^\lift$, which consists of IDs created by bullet point
(Lift) and of IDs created by bullet point (Transfer) in
Appendix~\ref{app:prflinearizeaccessids}. 
Let $\Gamma_{\bounded}$ consist of~$\Sigma'$ and of the IDs of~$\Sigma^\lift$
created by bullet point (Lift), and let
$\Gamma_{\acyclic}$ consist of the rules of~$\Sigma$ created by bullet point
(Transfer).

We now claim the following:

\begin{claim}
$\amd$ is equivalent to checking whether the chase of~$I_0^\lift$ by 
$\Gamma_{\bounded} \cup \Gamma_{\acyclic}$ satisfies~$Q'$,
  where the instance $I_0^\lift$ is obtained from~$I_0 \colonequals \canondb(Q)$ by applying derived truncated accessibility axioms
and the original axioms.
\end{claim}

\begin{proof}
We know that~$\amd$ is equivalent to the containment under~$\Gamma = \Sigma \cup
  \Sigma' \cup \Delta$.

It is easy to see that proofs formed from~$I_0^\lift$ using~$\Gamma_{\bounded}
  \cup \Gamma_{\acyclic}$
can be simulated by a proof formed from~$I_0$ using~$\Gamma$, so we focus on showing the converse.
We can observe that it suffices to consider chase proofs where
  we first fire rules of~$\Sigma$, (Truncated Accessibility) axioms, and
  (Transfer) axioms to get
a set of primed facts $I_1$,
and we then fire rules of~$\Sigma'$ to get $I_2$. From
  Proposition~\ref{prop:linearizeaccessids} 
we  know that using the axioms of the linearization,
which are in~$\Gamma_{\bounded} \cup \Gamma_{\acyclic}$, we can
derive a set of primed facts~$I'_1 = I_1$
Now we can apply the  rules of~$\Sigma'$ to~$I'_1$ to get a set $I'_2$
that is  a homomorphic image of~$I_2$. 
We conclude that~$I'_2$ also has a match of~$Q'$ as required.
\end{proof}

Now, the semi-width of~$\Gamma_{\bounded} \cup \Gamma_{\acyclic}$ is then $w$, since
 $\Gamma_{\bounded}$ consists of $\incd$s of width~$w$ and  $\Gamma_{\acyclic}$ of
 acyclic $\incd$s.
 We can therefore answer the problem in~$\np$ using
 Proposition~\ref{prop:semiwidthclassic}. This concludes the proof of
Theorem~\ref{thm:npidsbounds} in the case without result bounds.

\subsubsection{Proving Theorem~\ref{thm:npidsbounds} with Result Bounds}

We now conclude the proof of Theorem~\ref{thm:npidsbounds} by handling the case
with result bounds.
By Theorem~\ref{thm:simplifyidsexistence}, for any schema $\aschema$ whose constraints $\Sigma$ are $\incd$s,
we can reduce the monotone answerability problem to the same problem for the
existence-check simplification $\aschema'$ with no result bounds, by
replacing each result-bounded method $\mt$ on a  relation~$R$  with
a non-result bounded access method $\mt_1$ on a new relation  $\checkview_\mt$, and expanding
$\Sigma$ to a larger set of constraints $\Sigma_1$,
adding additional constraints capturing the semantics of the ``existence-check views'' $\checkview_\mt$:
\begin{align*}
\forall \vec x \vec y ~ R(\vec x, \vec y) \rightarrow \checkview_\mt(\vec x) \\
\forall \vec x ~ \checkview_\mt(\vec x) ~ \rightarrow ~ \exists \vec y ~ R(\vec x, \vec y)
\end{align*}
Let us denote IDs of the first form as ``relation-to-view'' and
of the second form as ``view-to-relation''.
Note that these IDs do not have bounded width, hence we cannot simply reduce to
the case without result bounds that we have just proved. We will explain how to
adapt the proof to handle these IDs, namely, 
linearizing using Proposition~\ref{prop:linearizeaccessids}, and then
partitioning
the results into two subsets, one of bounded width and the  other acyclic.

Let us consider the query containment problem for the monotone answerability problem of~$\Sigma_1$.
This problem is of the form
$Q \subseteq_\Gamma Q'$, where
$\Gamma$ contains $\Sigma_1$, its copy $\Sigma'_1$, and the accessibility
axioms. These axioms can again be rephrased: for each access method $\mt$ on a
relation~$S$, letting $\vec x$ denote the input positions of~$\mt$, we have the
following two axioms:
\begin{itemize}
  \item  (Truncated Accessibility): $\left(\bigwedge_i \accessible(x_i)\right) \wedge S(\vec x, \vec y) \rightarrow \bigwedge_i
  \accessible(y_i) $
\item (Transfer): $\left(\bigwedge_i \accessible(x_i)\right) \wedge S(\vec x, \vec y)
  \rightarrow  S'(\vec x, \vec y)$
\end{itemize}
In the above, the relation $S$ can be any of the relations of~$\Sigma_1$, including relations $R$ of the
original signature and relations $\checkview_\mt$. In the first case, this means
that~$\mt$ is an access method of~$\aschema$ that did not have a result
bound. In the second case, this means that~$\mt$ is a method of the form $\mt_1$ introduced in the
existence-check simplification $\aschema'$ for a result-bounded method
of~$\aschema$, so $\mt_1$ has no input positions: this means that, in this case,
the (Truncated Accessibility) axiom is vacuous and the (Transfer) axiom further
simplifies to:
\[
  \text{(Simpler Transfer):~}\checkview_\mt(\vec y) \rightarrow \checkview'_\mt(\vec y)
\]
We first observe that in~$\Gamma$ we do not need to include 
the view-to-relation constraints of~$\Sigma_1$: in the
chase, they will never fire, since facts over~$\checkview_\mt$ can only be formed  from the corresponding
$R$-fact, and we only fire active triggers.
Similarly, we do not need  to include the relation-to-view constraints
of~$\Sigma_1'$. These rules could fire to produce a new fact
$\checkview'_\mt(\vec y)$, but such a fact could only trigger the corresponding
view-to-relation constraint of $\Sigma_1'$, resulting in a state of the chase
that has a homomorphism to the one before the firing of the relation-to-view
constraint. Thus such firings can not lead to new matches.
Thus, $\Gamma$ consists now of~$\Sigma$, of~$\Sigma'$,
of (Truncated Accessibility) and
(Transfer) axioms for each method $\mt$ having no result bound in~$\aschema$,
and for each method $\mt$ with a result bound in~$\aschema$ we have
a relation-to-view
constraint from~$R$ to~$\checkview_\mt$ that comes from~$\Sigma_1$,
a view-to-relation constraint from~$\checkview_\mt'$ to~$R'$
that comes from $\Sigma_1'$, and a 
(Simpler Transfer) axiom.

We next note that we can normalize chase proofs with~$\Gamma$ so that the relation-to-view constraints
are applied only prior to (Simple Transfer). 
Thus, for each result-bounded method $\mt$ of~$\aschema'$,
we can merge the relation-to-view rule
from~$R$ to~$\checkview_\mt$, the
(Simpler Transfer) axiom from~$\checkview_\mt$ to~$\checkview_\mt'$, and
the view-to-relation rules from~$\checkview_\mt'$ to~$R'$, into an axiom of the
following form, where $\vec x$ denotes the input positions of~$\mt$:
\[
\mbox{(Result-bounded Fact Transfer) } R(\vec x, \vec y) \rightarrow  \exists \vec z ~ R'(\vec x, \vec z)
\]
To summarize, the resulting axioms $\Gamma'$ consist of:
\begin{itemize}
  \item The original constraints $\Sigma$ of the schema.
  \item Their primed copy $\Sigma'$
  \item The (Truncated Accessibility) and (Transfer) axioms for each access method
    without result bounds
  \item The (Result-bounded Fact Transfer) axioms for access methods with result
    bounds
\end{itemize}
In other words, the only difference with the setting without result bounds is
the last bullet point corresponding to (Result-bounded Fact Transfer). We can
then conclude with exactly the same proof as for the case without result bounds,
but modifying the proof of Proposition~\ref{prop:linearizeaccessids} to add the
following axiom:

\begin{itemize}
  \item (Result-bounded Fact Transfer): For each relation $R$ and subset $P$ of
    positions of~$R$ of size at most~$w$, for each access method $\mt$ on~$R$
    with a result bound, we add the $\incd$:
    \[R_P(\vec x, \vec y) \rightarrow \exists \vec z ~ R'(\vec x, \vec z)\]
    where $\vec x$ denotes the input positions of~$\mt$.
\end{itemize}

It is clear that adding this axiom ensures that the same primed facts are
generated than in the short-cut chase, and the resulting axioms still have
bounded semi-width: the (Result-bounded Fact Transfer) axioms are grouped in the
acyclic part together with the Transfer axioms, and they still have an acyclic
position graph.
This completes the proof of Theorem~\ref{thm:npidsbounds}.

\sectionarxiv{Proofs for Section~\lowercase{\ref{sec:simplifychoice}}: Schema
Simplification for Expressive\\Constraints}{Proofs for Section~\lowercase{\ref{sec:simplifychoice}}: Schema Simplification for Expressive Constraints}
\subsection{Proof of Theorem~\ref{thm:simplifychoice}: Choice Simplification for
Equality-Free FO}
\label{app:simplifychoice}

Recall the statement of Theorem~\ref{thm:simplifychoice}.

\begin{quote}
  \thmsimplifychoice
\end{quote}

Using our equivalence with~$\amd$, we see that it suffices to show:

\begin{quote} 
 Let schema $\aschema$ have  constraints  given by
equality-free first-order constraints,  and $Q$ be a CQ that is $\amd$ in~$\aschema$.
Then $Q$ is also $\amd$ in the choice simplification of~$\aschema$.
\end{quote}

We will again use the  
  ``blowing-up'' construction of Lemma~\ref{lem:enlarge}. Note that, this time,
  the schema of~$\aschema$ and~$\aschema'$ is the same, so we simply need to
  show that $I_p$ is a subinstance of~$I_p^+$ for each $p \in \{1, 2\}$.

  Consider a counterexample $I_1, I_2$ to~$\amd$ for~$Q$ in the choice
  simplification:
  we know that~$I_1$ satisfies $Q$, that $I_2$ violates $Q$, that
  $I_1$ and $I_2$ satisfy the equality-free first order constraints 
  of~$\aschema$,
  and that $I_1$ and $I_2$ have a common subinstance $I_\acc$ which is access-valid
  in~$I_1$ in the choice simplification of~$\aschema$.
  We will expand them to~$I_1^+$ and $I_2^+$ that have a common subinstance which
  is access-valid in~$I_1^+$ for~$\aschema$.

  For each element $a$ in the domain of~$I_1$, introduce infinitely many
  fresh elements $a_j$ for~$j \in \NN_{>0}$, and identify $a_0 \colonequals a$.
  Now, define $I_1^+ \colonequals \mathrm{Blowup}(I_1)$, where 
  $\mathrm{Blowup}(I_1)$ is the instance with facts $\{R(a^1_{i_1} \ldots a^n_{i_n}) \mid R(\vec
  a) \in I_1, \vec i \in \NN^n\}$. Define $I_2^+$ from~$I_2$ in the same way.

  We will now show correctness of this construction.
We claim that~$I_1$ and $I_1^+$ agree on all equality-free first-order
  constraints, which we show using a variant of the
  standard Ehrenfeucht-Fra\"iss\'e game without
  equality~\cite{casanovas1996elementary}.
  In this game there are pebbles on both structures;
 play proceeds by Spoiler placing a new pebble on some element in one structure, and Duplicator
must respond by placing a pebble with the same name in the other structure. Duplicator
loses if the mapping given by the pebbles does not preserve all relations
of the signature.  If Duplicator has a strategy  that never loses, then one can show by induction
that the two structures agree on all equality-free first-order sentences.

  Duplicator's strategy  will maintain the following invariants:
\begin{enumerate}
\item  if a pebble is on some element
  $a_j \in I_1^+$, then the corresponding pebble in~$I_1$ 
 is on~$a$;
\item if a pebble is on some element in~$I_1$, then the corresponding pebble
in~$I_1^+$ is on some element $a_j$ for~$j \in \NN$.
\end{enumerate}
These invariants will guarantee that the strategy is winning.
Duplicator's  response to a move by Spoiler in~$I_1^+$ is determined by the strategy 
above. In response to a move by Spoiler placing a pebble on~$b$ in~$I_1$, Duplicator
places the corresponding pebble on 
$b_0 = b$ in~$I_1^+$. 

Clearly the same claim can be shown for~$I_2$ and $I_2^+$.
In particular this shows that~$I_1$ still satisfies $Q$ and $I_2$ still
violates $Q$.

All that remains is to construct the common subinstance.
Let $I_\acc^+ \colonequals \mathrm{Blowup}(I_\acc)$. As $I_\acc$ is a common
subinstance of~$I_1$ and~$I_2$, clearly $I_\acc^+$ is a common subinstance
of~$I_1^+$ and~$I_2^+$. To see why $I_\acc^+$ is access-valid in~$I_1$, 
given an input tuple $\vec t'$ in
$I_\acc^+$, let $\vec t$ be the corresponding tuple in~$I_\acc$.
If $\vec t$ had no matching tuples in~$I_1$, then clearly the same is true  in
$I_1^+$.
If $\vec t$ had at least one matching tuple $\vec u$ in~$I_1$, then such a
tuple exists in~$I_\acc$ because it is access-valid in~$I_1$, and hence
sufficiently many copies exist in~$I_\acc^+$ to satisfy the original result
bounds, so that we can find a valid output for the access in~$I_\acc^+$.
Hence $I_\acc^+$ is access-valid in~$I_1^+$, which completes the proof.

\subsection{Proof of Theorem~\ref{thm:simplifychoiceuidfd}: Choice Simplification for
UIDs and FDs}
\label{app:simplifychoiceuidfd}

Recall the statement of Theorem~\ref{thm:simplifychoiceuidfd}:

\begin{quote}
\thmsimplifychoiceuidfd
\end{quote}

Our high-level strategy to prove Theorem~\ref{thm:simplifychoiceuidfd} is to
use a ``progressive'' variant of the
process of Lemma~\ref{lem:enlarge},
a variant where we ``fix'' one access at a time.
Remember that Lemma~\ref{lem:enlarge} said
that, if a counterexample to~$\amd$ in~$\aschema'$ can be
expanded to a counterexample in~$\aschema$, then $Q$ being $\amd$  in
$\aschema$ implies the same in~$\aschema'$. The next lemma makes a weaker
hypothesis: it assumes that for any counterexample in~$\aschema'$ and for any
choice of access, we can expand to a counterexample in~$\aschema'$ in which there
is an output to this access which is valid for~$\aschema$. To ensure that we
make progress, we must also require that, for every choice of access to which
there was previously a valid output for~$\aschema$, then there is still such an
output to the access.
In other words, the assumption is
that we can repair the counterexample from~$\aschema'$ to~$\aschema$ by working one
access at a time. We show that this is sufficient to reach the same conclusion:

\begin{lemma}
  \label{lem:enlargeprog}
  Let $\aschema$ be a schema and $\aschema'$ be its choice simplification,
  and let $\Sigma$ be the constraints.

  Assume that, for any CQ~$Q$ not $\amd$ in~$\aschema'$, for any
  counterexample $I_1, I_2$ of~$\amd$ for~$Q$ and
  $\aschema'$ with a common subinstance $I_\acc$ which is access-valid
  in~$I_1$ for~$\aschema'$, for any access $\mt, \accbind$ in~$I_\acc$,
  the following holds:
  we can construct a counterexample $I_1^+, I_2^+$ of
  $\amd$ for~$Q$ and $\aschema'$,
  i.e., $I_1^+$ and $I_2^+$ satisfy $\Sigma$,
  $I_1 \subseteq I_1^+$, $I_2 \subseteq I_2^+$,
  $I_1^+$ has a homomorphism to~$I_1$, $I_2^+$ has a homomorphism to~$I_2$,
  and $I_1^+$ and $I_2^+$ have a common subinstance $I_\acc^+$ which is 
  access-valid in~$I_1^+$ for~$\aschema'$, and we can further impose that:
  \begin{enumerate}
    \item $I_\acc^+$ is a superset of~$I_\acc$;
    \item there is an output to the access $\mt, \accbind$ in~$I_\acc^+$ which is valid
in~$I_1$
  for~$\aschema$;
    \item for any access in~$I_\acc$ having an output
      in~$I_\acc$ which is valid
      for~$\aschema$ in~$I_1$, there is an output to this access in~$I_\acc^+$ 
      which is valid for~$\aschema$ in~$I_1^+$;
  \item for any
  access in~$I_\acc^+$ which is not an access in~$I_\acc$, there is an
      output in~$I_\acc^+$ which is valid for~$\aschema$ in~$I_1^+$;
  \end{enumerate}
  Then any query which is $\amd$ in~$\aschema$ is also
  $\amd$ in~$\aschema'$.
\end{lemma}

\begin{proof}
  We will again prove the contrapositive. Let $Q$ be a query which is not
  $\amd$ in~$\aschema'$, and let $I_1, I_2$ be a
  counterexample, with~$I_\acc$ the common subinstance of~$I_1$ and~$I_2$ which
  is access-valid in~$I_1$ for~$\aschema'$.
  Enumerate the
  accesses in~$I_\acc$ as a sequence $(\mt^1, \accbind^1), \ldots,
  (\mt^n, \accbind^n), \ldots$: by definition of~$I_\acc$, all of them have an
  output in~$I_\acc$ which is valid in~$I_1$ for $\aschema'$, 
  but initially we do not assume that any of these outputs are valid
  for~$\aschema$ as well.
  We then build
  an infinite sequence $(I_1^1, I_2^1), \ldots, (I_1^n, I_2^n), \ldots$ along
  with all
  the corresponding common subinstances $I_\acc^1, \allowbreak \ldots, \allowbreak I_\acc^n, \ldots$, with each
  $I_\acc^i$ being a common subinstance of~$I^i_1$ and $I^i_2$ which is access-valid
  in~$I^i_1$, by 
  applying the process of the hypothesis of the lemma in succession to the
  accesses
  $(\mt^1, \accbind^1), \ldots, (\mt^n, \accbind^n), \ldots$. In particular, note
  that whenever $(\mt^i, \accbind^i)$ already has an output in~$I_\acc^i$ which
  is valid in~$I^i_1$ for~$\aschema$, then we
  can simply take $I^{i+1}_1$, $I^{i+1}_2$, $I^{i+1}_\acc$ to be respectively equal
  to~$I^i_1$, $I^i_2$, $I^i_\acc$, without even having to rely on the hypothesis of
  the lemma.

  It is now obvious by induction that, for all $i \in \NN$, $I^i_1$ and $I^i_2$
  satisfy the constraints $\Sigma$, we have $I_1 \subseteq I^i_1$, we have that $I^i_2$ has a
  homomorphism to~$I_2$, and $I^i_\acc$ is a common subinstance of~$I^i_1$ and
  $I^i_2$ which is access-valid in~$I^i_1$ for~$\aschema'$, where the
  accesses $(\mt^1, \accbind^1), \ldots, \allowbreak (\mt^i, \accbind^i)$ additionally 
  have an output in~$I_\acc^i$ which is valid in~$I_1^i$ 
  for~$\aschema$, and where all the accesses in~$I^i_\acc$ which are not accesses
  of~$I_\acc$ also have an output in~$I^i_\acc$ which is valid in~$I_1^i$ for~$\aschema$. Hence, 
  considering, the infinite result $(I_1^\infty, I_2^\infty), I_\acc^\infty$ of this
  process, we know that all accesses in~$I^\infty_\acc$ have an output
  in~$I^\infty_\acc$ which is valid in~$I_1^\infty$ for~$\aschema$. Hence,
  $I_\acc^\infty$ is actually a common subinstance of~$I_1^\infty$ and
  $I_2^\infty$ which is access-valid in~$I_1^\infty$ for~$\aschema$, so
  $I_1^\infty, I_2^\infty$ is a counterexample to 
  $\amd$ of~$Q$ in~$\aschema$, which concludes the proof.
\end{proof}

  Thanks to Lemma~\ref{lem:enlargeprog}, we can now prove
  Theorem~\ref{thm:simplifychoiceuidfd} by arguing that we can fix each
  individual access.
  Let $\aschema$ be the schema, let $\aschema'$ be its choice simplification, and
  let $\Sigma$ be the constraints.

  We now explain how we fulfill the requirements of  Lemma~\ref{lem:enlargeprog}. 
Let $Q$ be a CQ and assume that it is not
  $\amd$ in~$\aschema'$, and let $I_1, I_2$, be a 
  counterexample to~$\amd$, with~$I_\acc$ being a common
  subinstance of~$I_1$ and $I_2$ which is access-valid in~$I_1$ for~$\aschema'$.
  Let $(\mt, \accbind)$ be an access on relation~$R$ in~$I_\acc$:
  we know that there is an output to the access in~$I_\acc$ which is valid
  for~$\aschema'$ in~$I_1$, but this output is not necessarily valid
  for~$\aschema$.
  Our goal is to build $I_1^+$ and $I_2^+$ such that~$I_1^+$ is 
  a superinstance of~$I_1$ and $I_2^+$ homomorphically maps to~$I_2$,;
  we want both $I_1^+$ and $I_2^+$ to satisfy  $\Sigma$, and want $I_1^+$ and
  $I_2^+$ to have a
  common subinstance $I_\acc^+$ which is access-valid in~$I_1^+$,
  where $\accbind$ now has an output which is valid for~$\aschema$ (i.e., not only for the choice
  simplification), all new accesses also have an output which is valid for~$\aschema$, and
  no other accesses are affected.
  At a high level, we will do the same blow-up as in the proof of
  Theorem~\ref{thm:fdsimplify}, except that we will need to chase afterwards to
  argue that the $\uincd$s are true.

  First observe that, if there are no matching tuples in~$I_1$ for the access
  $(\mt, \accbind)$, then the empty set is already an output in~$I_\acc$ to the
  access which is valid in~$I_1$ for~$\aschema$
  so there is nothing to do, i.e., we can just take $I_1^+ \colonequals I_1$,
  $I_2^+ \colonequals I_2$, and $I_\acc^+ \colonequals I_\acc$.
  Further, note that if there
  is only one matching tuple in~$I_1$ for the access, as~$I_\acc$ is access-valid for
  the choice simplification, then this tuple is necessarily in
  $I_\acc$ also, so again there is nothing to do. Hence, we know that there is
  strictly more than one matching tuple in~$I_1$ for the access $(\mt,
  \accbind)$; as~$I_\acc$ is access-valid for~$\aschema'$, then it contains at least
  one of these tuples, say $\vec t_1$, and as~$I_\acc \subseteq I_2$,
  then $I_2$ also contains $\vec t_1$. Let $\vec t_2$ be a second matching tuple
  in~$I_1$ which is different from~$\vec t_1$. Let $C$ be 
  the non-empty set of  positions of~$R$ where $\vec t_1$ and $\vec t_2$
  disagree. Note that, since $I_1$ satisfies the constraints,
  the constraints cannot imply an FD from the complement of~$C$
  to a position~$j \in C$, as otherwise $\vec t_1$ and $\vec t_2$
  would witness that~$I_1$ violates this FD.

  We form an infinite collection of facts $R(\vec o_i)$
  where $\vec o_i$ is constructed from~$\vec t_1$
  by replacing the values at positions in~$C$ by fresh values (in particular distinct from
  values in   other positions in~$R$ and in other $\vec o_j$'s). 
  Let $N \colonequals  \{R(\vec o_1) \ldots R(\vec o_n), \ldots\}$.
  We claim that
  $I_1 \cup N$ does not violate any FDs implied by
  the schema. If there were a violation of a FD~$\phi$, the violation
  $F_1, F_2$ must involve some new fact $R(\vec o_i)$, as~$I_1$ on its own
  satisfies the constraints. We know that the left-hand-side of~$\phi$ cannot include a
  position of~$C$, as all elements in the new facts $R(\vec o_i)$ at these
  positions are fresh. Hence, the left-hand-side of~$\phi$ is included in the
  complement of~$C$, but recall that we argued above that then the
  right-hand-side of~$\phi$ cannot 
  be in~$C$. 
  Hence, both the left-hand-side and right-hand-side of~$\phi$ are in the complement of~$C$. But
  on this set of positions the facts of the violation~$F_1$ and $F_2$ agree
  with the existing fact $\vec t_1$ and~$\vec t_2$ of~$I_1$, a contradiction.
  So we know that~$I_1 \cup N$ does not violate the FDs.
  The same argument shows that~$I_2 \cup N$ does not violate the FDs.

  So far, the argument was essentially the same as in the proof of
  Theorem~\ref{thm:fdsimplify}, but now we explain the additional chasing step.
  Let $W$ be formed from chasing $N$ with the $\uincd$s,
  ignoring triggers whose exported element occurs in~$\vec t_1$.
  We have argued that~$I_1 \cup N$ and $I_2 \cup N$ satisfy the FDs. 
We want to show that both the $\uincd$s and FDs  hold of~$I_1 \cup W$ and $I_2 \cup W$.
 Note that as we
  have $\vec t_1$ in~$I_1$ and in~$I_2$ we know that any element of the domain
  of~$N$ which also occurs in~$I_1$ or in~$I_2$ must be an element of~$\vec
  t_1$. Also note that any such element that  occurs at a certain position~$(R,i)$ in~$N$, then it also 
occurs at~$(R,i)$
  in~$I_1$.  We then conclude that
   that~$I_1 \cup W$ and $I_2 \cup W$
  satisfy the constraints, thanks to the following general lemma:

\begin{lemma} \label{lem:decomp}
Let $\ids$ be a set of~$\uincd$s
  and let $\fds$ be a set of FDs.
  Let $I$ and $N$ be instances, and let 
  $\Delta \colonequals \dom(I) \cap \dom(N)$.
  Assume that~$I$ satisfies $\fds \cup \ids$, that
  $I \cup N$ satisfies $\fds$, and that whenever $a \in \Delta$ occurs at a
  position~$(R, i)$ in~$N$ then it also occurs at~$(R,i)$ in~$I$.
  Let $W$ denote the chase of~$N$ by~$\ids$ 
  where we do not fire any triggers which map  an
  exported variable to an element of  $\Delta$.
  Then $I \cup W$ satisfies $\ids \cup \fds$.
\end{lemma}

Intuitively, the lemma applies to any instance $I$ satisfying the
constraints ($\uincd$s and FDs), to which we want to add a set $N$ of new facts, in a way which
still satisfies the constraints. We assume that the elements of~$I$ that occur
in~$N$ never do so at new positions relative to where they occur in~$I$, and we
assume that~$I \cup N$ satisfies the FDs. We then claim that
we can make $I \cup N$ satisfy the $\uincd$s simply
by chasing~$N$ by the $\uincd$s in a way which ignores some triggers, i.e., by
adding~$W$. (The triggers that we ignore are unnecessary in terms of satisfying
the $\uincd$s, and in fact we would possibly be introducing FD violations by firing
them, so it is important that we do not fire them.)

We now prove the lemma:

\begin{proof}
  We assume without loss of generality that the~$\uincd$s are closed under
  implication~\cite{cosm}. This allows us to assume that, whenever we chase by
  the~$\uincd$s, after each round of the chase, all remaining violations of the
  $\uincd$s are on facts involving some null created in the last round.
In particular, in
  $W$, all remaining violations of~$\ids$ are on facts of~$N$.

  We first show that~$I \cup W$ satisfies $\ids$. Assume by way of
  contradiction that it has an active trigger $\trig$ for a $\uincd$ $\dep$. 
   The range of~$\trig$ is either in~$I$ or in~$W$. The
  first case is impossible because $I$ satisfies $\ids$ so it cannot have an
  active trigger for~$\dep$. The second case is impossible also by definition of the chase,
  unless the active trigger maps an exported variable
  to an element  of~$\Delta$, i.e., it is a trigger which we did not fire in
  $W$. Let $R(\vec a)$ be the fact of~$W$ in the image of~$\trig$.
  By the above, as the~$\incd$s are closed under
  implication, $R(\vec a)$ is necessarily a fact of~$N$. Let $a_i$ be the image of the exported
  variable in~$\vec a$, with~$a_i \in \Delta$. Hence, $a_i$ occurs at position~$(R, i)$ in
  $N$, so by our assumption on~$N$ it also occurs at position~$(R,i)$ in~$I$. Let $R(\vec b)$ be a
  fact of~$I$ such that~$b_i = a_i$. As $I$ satisfies $\ids$, for the match
of the body of~$\dep$ to~$R(\vec b)$ there is a corresponding
fact $F$ in~$I$ extending the match to the head of~$\dep$.
  But $F$ also serves as a witness in~$I \cup W$ for the match of the body of
  $\dep$,
  so we have reached a contradiction.
  Hence, we have shown satisfaction of~$\ids$.
  
  We now show that~$I \cup W$ satisfies $\fds$.
  We begin by arguing that~$W$
  satisfies $\fds$. This is  because $N$ satisfies $\fds$; it is easy to show (and
is proven in \cite{cali2003decidability}) 
that  performing the chase with
active triggers of  $\uincd$s
  never creates violations of FDs, so this is also true of~$W$
  as it is a subset of the facts of the actual chase of~$N$ by~$\ids$.
  Now,
  assume by way of contradiction that there is an FD violation~$\{F, F'\}$ in~$I \cup
  W$. As $I$ and $W$ satisfy $\fds$ in isolation,
  it must be the case that one
  fact of the violation is in~$I$ and one is in~$W$: without loss of
  generality, assume that we have $F \in I$ and $F' \in W$. There are
  three possibilities: $F'$ is a fact of~$N$, $F'$ is a fact created in the
  first round of the chase (so one of its elements, the exported element, is in
  $\dom(N)$, and the others are not), or $F'$ is a fact created in later rounds
  of the chase. 
The first case is ruled out by the
  hypothesis that~$I \cup N$ satisfies $\fds$.
In the
  second case, by definition of~$W$, the element from~$\dom(N)$ in
  $F'$ cannot be from~$\dom(I)$, as otherwise we would not have exported this
  element (i.e., it would be a trigger that we would not have fired); hence $F'$ contains only fresh elements and one element in
  $\dom(N) \setminus \dom(I)$, so  $F$ and $F'$ are on disjoint elements so they cannot be a
  violation. 
In the third case, $F'$ contains only fresh elements, so again  $F$
  and $F'$ cannot form an FD violation as they have no common element.
\end{proof}

So we now know that~$I_1^+ \colonequals I_1 \cup W$
and $I_2^+ \colonequals I_2 \cup W$ satisfy the constraints.
Let
us then conclude our proof of
Theorem~\ref{thm:simplifychoiceuidfd} using the process
of Lemma~\ref{lem:enlargeprog}. We first show that~$(I_1^+, I_2^+)$ is a
counterexample of~$\amd$ for~$Q$ and $\aschema'$:
\begin{itemize}
    \item We have just shown that~$I_1^+$
and $I_2^+$ satisfy the constraints.
    \item We clearly have 
$I_1 \subseteq I_1^+$ and $I_2 \subseteq I_2^+$.
\item We now argue that~$I_1^+$ has a homomorphism to~$I_1$ (the proof for
  $I_2^+$
  and $I_2$ is analogous). This point is reminiscent of the proof of
    Theorem~\ref{thm:simplifyidsexistence}.
 We first define the homomorphism from~$I_1 \cup N$
to~$I_1$ by mapping $I_1$ to itself, and mapping each fact of~$N$ to~$R(\vec
t_1)$ (which is consistent with what precedes);
it is clear that this
is a homomorphism. We then extend this homomorphism inductively on each fact
created in~$W$ in the following way. Whenever a
fact $S(\vec b)$ is created by firing an active trigger $R(\vec a)$ 
    for a $\uincd$ $R(\vec x) \rightarrow S(\vec y)$ where $x_p = y_q$ is the
    exported variable,
    (so we have $a_p = b_q$), consider the
fact $R(h(\vec a))$ of~$I_1$ (with~$h$ defined on~$\vec a$ by induction
    hypothesis). As $I_1$ satisfies
$\Sigma$, we can find a fact $S(\vec c)$ with~$c_q = h(a_p)$, so we can define
$h(\vec b)$ to be $\vec c$, and this is consistent with the existing image
    of~$a_p$.

\item We can define $I_\acc^+ \colonequals I_\acc \cup
  W$ as a common subinstance of~$I_1^+$ and $I_2^+$. We now show that~$I_\acc^+$
    is access-valid
    for~$I_1^+$ and $\aschema'$. Let  $(\mt', \accbind')$ be an access in
    $I_\acc^+$. The first case is when $(\mt', \accbind')$ includes an element
    of~$\dom(I_\acc^+)
    \setminus \dom(I_\acc)$, namely, an element of~$\dom(W) \setminus
    \dom(I_1)$. In this case, clearly all matching facts must be facts that were created in the
    chase, i.e., they are facts of~$W$. Hence, we can construct a 
    valid output from~$W \subseteq I_\acc^+$. The second case is when $(\mt', \accbind')$
    is only on elements of~$\dom(I_\acc)$, then it is actually an access
    on~$I_\acc$, so, letting $U \subseteq I_\acc$ be the set of matching tuples
    which is the valid output to~$(\mt', \accbind')$ in~$I_1$, we can 
    construct a valid output to~$(\mt', \accbind')$ in~$I_1^+$ from~$U \cup W
    \subseteq I_\acc^+$, because any matching tuples for this access in~$I_1^+$
    must clearly be either matching tuples of~$I_1$ or they must be matching
    tuples of~$W$.
\end{itemize}
We now show the four additional conditions:
\begin{enumerate}
  \item It is clear by definition that~$I_\acc^+ \supseteq I_\acc$.
\item We must show that the access $(\mt, \accbind)$ is valid for~$\aschema$ in
  $I_\acc^+$. Indeed, there are now
infinitely many matching tuples in~$I_\acc^+$, namely, those of~$N$. Thus
this access is valid for~$\aschema$ in~$I_1$: we can choose as many tuples as
the value of the bound to obtain an output which is valid in~$I_1$.

\item We must verify that, for any access $(\mt', \accbind')$ 
  of~$I_\acc$ that has an output which is valid in~$I_1$
    for~$\aschema$, then we can construct such an output in~$I_\acc^+$ which is
    valid in~$I_1^+$
    for~$\aschema$.
    The argument is
    the same as in the second case of the fourth bullet point above: from the
    valid output to the access~$(\mt', \accbind')$ in~$I_1$ for~$\aschema$, we construct
    a valid output to~$(\mt', \accbind')$ in~$I_1^+$ for~$\aschema$.
 \item  Let us consider any access in~$I_\acc^+$ 
    which is not an access in~$I_\acc$.
The binding for this access must include
some element of~$\dom(W)$, so its matching tuples must be in~$W$, which are all
    in~$I_\acc^+$. Hence, 
    by construction any such accesses are valid for~$\aschema$. 
\end{enumerate}
So we conclude the proof of Theorem~\ref{thm:simplifychoiceuidfd} using
Lemma~\ref{lem:enlargeprog}, fixing each access according to the above process.

\needspace{5em}
\section{Proofs for Section~\lowercase{\ref{sec:complexitychoice}}: Decidability
using Choice Simplification}
\subsection{Further example of decidability using choice simplification}
\label{app:complexitychoice}

In the body of the paper we proved decidability for
monotone answerability with FGTGDs.
We claimed that it applies to  extensions with disjunction and negation.
We now substantiate this.
We will use the Guarded Negation Fragment (GNF) of \cite{gnfj}. We will need
to know only that GNF satisfiability is decidable, and that GNF
contains  CQs, contains extensions of frontier-guarded TGDs with disjunction and negation,
and is closed under Boolean combinations of sentences.

\begin{theorem} \label{thm:decidegnf}
  We can decide whether a CQ is
monotone answerable with respect to a schema having result bounds, with 
equality-free constraints
that are  in the Guarded Negation Fragment. In particular, this holds when
  constraints are extensions of FGTGDs with disjunction and negation~\cite{bourhis2016guarded,gnfj}.
\end{theorem}
\begin{proof}
Again, by  Theorem~\ref{thm:simplifychoice} we can assume that all result bounds
are one, and by Proposition~\ref{prop:elimupper} we can replace
the schema with the relaxed version containing
 only result lower bounds. Now, a result lower bound of~$1$ can be expressed
  as an $\incd$. Thus, Proposition~\ref{prop:reduce} allows us to reduce
  monotone answerability to a query containment problem
  with constraints $\Gamma$ consisting of $\Sigma$, a copy
$\Sigma'$ on primed relations, and additional frontier-guarded TGDs.
  The constraints $\Gamma$ are still in GNF. The resulting containment problem
$Q \subseteq_\Gamma Q'$ can be restated as satisfiability of $Q \wedge \Gamma \wedge \neg Q'$, which
is a GNF satisfiability problem. Hence, we conclude because GNF satisfiability
  is decidable.
\end{proof}

\subsection{Proof of Theorem~\ref{thm:deciduidfd}: Complexity of Monotone
Answerability for UIDs and FDs}
\label{app:deciduidfd}

In this appendix, we show Theorem~\ref{thm:deciduidfd} with the $\exptime$
complexity bound. Recall the statement:

\begin{quote}
  \deciduidfd
\end{quote}

To show this theorem, we will introduce
a general linearization result that extends
Proposition~\ref{prop:linearizeaccessids}.
The result will reduce query containment under IDs
and full GTGDs to query containment under \emph{linear TGDs}, that is, TGDs
whose body and head consist of one single atom (but allowing variable
repetitions). 
IDs and full GTGDs can simulate arbitrary GTGDs,
for which containment is $\twoexp$-complete, and
$\exptime$-complete for constant signature arity~\cite{taming}.
However, we will be
able to show an $\exptime$ bound without assuming constant signature arity, by just
bounding the arity of the signature used in \emph{side atoms} (i.e., non-guard
atoms): this can handle, e.g.,
accessibility facts
in truncated accessibility axioms. We will also show an $\np$ bound under
some additional assumptions: this bound is not used in the main text, but we state it for
completeness because it generalizes 
Proposition~\ref{prop:linearizeaccessids}.
Our linearization technique will resemble that
of~\cite{gmp}, but their results only apply when bounding the
arity and number of relations in the whole signature, whereas we do not make
such assumptions.

In this appendix, we first state the generalized linearization result and its
consequences in Appendix~\ref{app:complexitygeneral}. We
then prove Theorem~\ref{thm:deciduidfd} using this result, in
Appendix~\ref{apx:separable}. Last, we give the
proof of the generalized linearization result, which is somewhat technical, in
Appendix~\ref{apx:prfcomplexitygeneral}.

\subsubsection{Statement of Generalized Linearization Result}
\label{app:complexitygeneral}
We consider constraints that consist of non-full $\incd$s and full GTGDs on a
specific side signature, and measure the arity of head relations:

\begin{definition}
  Let $\gdep$ be a full GTGD on signature $\sign$.
  The \emph{head arity} of~$\gdep$ is the number of variables used in the head
  of~$\gdep$. Given a sub-signature~$\sidesign \subseteq \sign$, we say that~$\gdep$ has
  \emph{side signature $\sidesign$} if there is a choice of guard atom in the body
  of~$\gdep$ such that all other body atoms are relations of~$\sidesign$.
\end{definition}

The result below uses the notion of semi-width, defined in
Appendix~\ref{apx:semiwidthclassic}.
For a set of constraints $\Sigma$, we write $\card{\Sigma}$ for their size (e.g. in a string
representation), and extend the notions of head arity and side signature in the
expected way.

\newcommand{\thmidreduce}{
  For any $a' \in \NN$, there are polynomials $P_1, P_2$ such that the
  following is true. Given:
  \begin{compactitem}
  \item A signature $\sign$ of arity $a$;
  \item A subsignature $\sidesign \subseteq \sign$ with~$n'$ relations and
    arity~$\leq a'$;
  \item A CQ~$Q$ on~$\sign$;
  \item A set $\Sigma$ of non-full $\incd$s of width $w$ and full GTGDs with
    side signature $\sidesign$ and head arity $h$;
  \end{compactitem}
  We can compute the following:
  \begin{compactitem}
  \item A set $\Sigma'$
  of linear TGDs of semi-width $\leq w$ and arity~$\leq a$, in time
    $P_1(\card{\Sigma}, 2^{P_2(w, h, n')})$, independently from~$Q$;
  \item A CQ~$Q^\lift$, in time 
  $P_1(\card{\Sigma},
  \card{Q}^{P_2(w, h, n')})$.
  \end{compactitem}
  The constraints $\Sigma'$ and the CQ~$Q^\lift$ ensure that
  for any CQ~$Q'$,
  we have $Q \subseteq_\Sigma Q'$ iff $Q^\lift \subseteq_{\Sigma'} Q'$.
}

\begin{theorem}
  \label{thm:idreduce}
  \thmidreduce
\end{theorem}

Note that we assume that $\incd$s are \emph{non-full}, i.e., they must create at
least one null. Of course, full $\incd$s can be seen as full GTGDs with empty side
signature, so they are also covered by this result, but they may
make the head arity increase if included in the full GTGDs.

Theorem~\ref{thm:idreduce} is proven in Appendix~\ref{apx:prfcomplexitygeneral}
by generalizing
the linearization argument of
Proposition~\ref{prop:linearizeaccessids}.
We compute derived axioms of a limited breadth (generalizing the notion of 
Appendix~\ref{apx:linearizeaccessids}), and then use them in a short-cut
chase, which avoids passing facts up and down.
We then show how to simulate the short-cut chase
by linear TGDs. The main 
difference with~\cite{gmp} is that we exploit the width
and side signature arity bounds to compute only a portion of the derived
axioms, without bounding the overall signature.

We can use Theorem~\ref{thm:idreduce} by fixing the head arity~$h$, the
width~$w$, and the entire side signature~$\sidesign$, to deduce the following.
This result is only given for completeness, as we do not use it in the main
text; but it is the result that one would use to generalize
Proposition~\ref{prop:linearizeaccessids}:

\begin{corollary}
  \label{cor:np}
  There is an $\np$ algorithm for query containment under bounded-width non-full $\incd$s and
  full GTGDs of bounded head arity on a fixed side signature.
\end{corollary}

While the side signature $\sidesign$ is constant,
the arity of~$\sign$ is not constant above; however, relations in
$\sign \setminus \sidesign$ can only be used in the bounded-width $\incd$s and as
guards in the full GTGDs.

\begin{proof}[\myproof of Corollary~\ref{cor:np}]
  Apply the reduction of Theorem~\ref{thm:idreduce}, which computes in
  $\ptime$ an equivalent set of linear TGDs of constant semi-width and a
  rewriting of the left-hand-side query. Then, conclude
  by Proposition~\ref{prop:semiwidthclassic-general}.
\end{proof}

The result also implies an $\exptime$ bound for query containment with a more
general language of $\incd$s and GTGDs, which we will use to prove
Theorem~\ref{thm:deciduidfd}:

\begin{corollary}
  \label{cor:exp}
  There is an $\exptime$ algorithm for query containment under $\incd$s and
  GTGDs on a bounded arity side signature.
\end{corollary}

\begin{proof}
  One can simulate  GTGDs by $\incd$s and full GTGDs, via additional relations.
Thus we can assume the GTGDs are full. Likewise, we can assume that the
  $\incd$s are non-full, by making the full $\incd$s part of the full GTGDs.

  Now, apply the reduction of Theorem~\ref{thm:idreduce}, which computes in
  $\exptime$ an equivalent set of linear TGDs and computes a rewriting of the
  left-hand-side query.
  Consider each one of the exponentially many possible first-order
rewritings of the
  right-hand-side query under these linear TGDs (see \cite{calirewriting,datalogpmj}),
  and for each of them, check whether it holds in the closure.
\end{proof}

\subsubsection{Proving Theorem~\ref{thm:deciduidfd} using Generalized
Linearization}
\label{apx:separable}
We can now complete the proof of Theorem~\ref{thm:deciduidfd} from the proof
given in the body. We must first explain why FD violations do not happen, and
second explain how to obtain the $\exptime$ bound from Corollary~\ref{cor:exp}.

\myparaskip
\myparagraph{FD violations}
Recall that, in the proof, we had to argue that after applying the FDs to the canonical database
and pre-processing the constraints slightly, we could drop the FDs in~$\Sigma$ and in~$\Sigma'$ without impacting the entailment.
We argued in the body that the constraints for~$\amd$ would consist of~$\Sigma$,
$\Sigma'$, and the
following:
\begin{itemize}
  \separableconstraints
\end{itemize}
We then modified the second set of axioms so that, in going
from~$R$ to~$R'$, they
preserve not only the input positions of~$\mt$, but also the positions of~$R$ that are determined by
input positions of~$\mt$ (i.e., appear as the right-hand-side of an FD whose
left-hand-side is included in the input positions). 
As mentioned in the body,
the use of these ``expanded result-bounded constraints''
does not impact the soundness of the chase, since chase step with these
  constraints can
be mimicked by a step with an original constraint followed by FD applications.

We now complete the argument  to show
that after this rewriting, and after applying the FDs to the initial
instance, we can apply  the TGD constraints while \emph{ignoring the FDs}.

To argue this, we note that it suffices to consider chase proofs where
the primed copies of the $\uincd$s in~$\Sigma'$ are never fired prior
to constraints in~$\Sigma$ or prior to expanded result-bounded constraints.
This is because the primed copies of $\uincd$s can not create triggers for any of those
constraints.

We show that in a chase with this additional property, the FDs will never fire.
We prove this by induction on the rule firing in the chase.

Observe that the $\uincd$s
of~$\Sigma$ cannot introduce FD violations when we perform the
chase, because we fire only active triggers. The same is true of the $\uincd$s
of~$\Sigma'$ when we apply them. So it suffices to consider the expanded
result-bounded constraints. 
Assume by contradiction that firing these rules creates a violation, and
consider the first violation that is created. Either the violation is on a
primed relation, or it is on an unprimed relation. If it is on a primed
relation, it consists of 
a first fact $F_1' = R'(\vec c, \vec d)$, and of a second fact $F_2' = R'(\vec f,
\vec g)$ which was just generated by firing an expanded result-bounded
constraint on some fact $F_2 = R(\vec f, \vec h)$. The constraint may be of
the form of the first bullet point above, in which case $\vec g$ and $\vec h$
are empty tuples; or it may be of the form of the second bullet
point above, modified to also export determined positions as we explained, in
which case all values in $\vec g$ are fresh. Our
additional property on the chase ensures that we do not fire~$\Sigma'$, so $F_1$ must also
have been generated by firing an expanded result-bounded constraint on some fact
$F_1 = R(\vec c, \vec e)$, and again $\vec d$ is either empty or only consists of
fresh values. Now, we know that the determiner of the violated FD must be within
the intersection of the positions of $\vec c$ and of $\vec f$, because it cannot
contain fresh values in any of the two facts $F_1'$ and $F_2'$. Hence, by the
modification that we did on the axioms, the determined position of the violated
FD must also be within the intersection of the positions of~$\vec c$ and
of~$\vec f$. This means that $F_1$ and $F_2$ are already a violation of the FD,
which contradicts minimality of the violation.

Now, if the violation is on an unprimed relation, it consists of a first fact $F_1' = R(\vec
c, \vec d)$, and of a second fact $F_2' = R(\vec f, \vec g)$ which was just generated by a
constraint of the form of the second bullet point above, modified to also export
determined positions as we explained. In this case, let $F_2 = R(\vec f, \vec
h)$ be the fact that triggered the rule application. Because the elements of
$\vec g$ are fresh, the determiner of the violated FD must be within positions
of~$\vec f$, hence, by the modification that we did on the axioms, the determined position of the violated FD must also be
within positions of~$\vec f$, but this means that $F_1'$ and $F_2$ are already a
violation of the FD, contradicting minimality.

\myparaskip
\myparagraph{Complexity}
The improved complexity bound is simply by using Corollary~\ref{cor:exp} on
$\gammasep$: the side signature is fixed, because it only consists of
$\accessible$. This shows the desired complexity bound and concludes the proof.

\subsubsection{Proof of Generalized Linearization Result}
\label{apx:prfcomplexitygeneral}
In this appendix, we prove Theorem~\ref{thm:idreduce},
which implies an $\np$ bound for query containment
under a class of guarded TGDs (Corollary~\ref{cor:np}), and an $\exptime$ bound
for query containment under a larger
class (Corollary~\ref{cor:exp}).
The first bound generalizes Johnson and Klug's result on query
containment under bounded-width $\incd$s \cite{johnsonklug}. The second
result generalizes a result of Cal\`{\i}, Gottlob and Kifer \cite{taming} 
that query containment under guarded TGDs of bounded arity
is in~$\exptime$. The construction we use is a refinement of the linearization
method given in Section~4.2 of
Gottlob, Manna, and Pieris \cite{gmp}.

Of the two corollaries mentioned above, the first one generalizes the technique
presented for accessibility axioms in 
Appendix~\ref{apx:linearizeaccessids}, and the
second one is used in the body of the paper to give bounds on the monotone
answerability problem. These two results, and the more general
Theorem~\ref{thm:idreduce}, are completely 
independent from access methods or result bounds, and may be of independent
interest.

We recall the statement of Theorem~\ref{thm:idreduce}:

\begin{quote}
  \thmidreduce
\end{quote}
Before proving Theorem~\ref{thm:idreduce}, we comment on the intuition of why
the hardness results of~\cite{taming} do not apply to the languages described in
Corollaries~\ref{cor:np} and~\ref{cor:exp}. For Corollary~\ref{cor:exp},
it is shown in~\cite[Theorem~6.2]{tamingjournal} that deciding the containment of a fixed
query into an atomic query under GTGDs is $\twoexp$-hard when the arity is
unbounded, even when the number of relations in the signature is bounded. The
proof works by devising a GTGD theory that simulates an $\expspace$ alternating
Turing machine, by coding the state of the Turing machine as facts on tuples of
elements: specifically, a fact $\mathit{zero}(\mathbf{V}, X)$ codes that there
is a zero in the cell indexed by the binary vector $\mathbf{V}$ in
configuration~$X$. The arity of such relations is unbounded, so they cannot be
part of the side signature~$\sidesign$. However, in the simulation of the Turing
machine, the GTGDs in the proof use another relation as guard (the~$g$
relation), and the bodies contain other high-arity relations, so there is no
choice of~$\sidesign$ for which the GTGD theory defined in the hardness proof
can satisfy the
definition of a side signature.

For Corollary~\ref{cor:np}, the proof in~\cite[Theorem~6.2]{tamingjournal}
explicitly writes the state of the~$i$-th tape cell of a configuration~$X$ as,
e.g., $\mathit{zero_i}(X)$. These relations occur in rule bodies where they are
not guards, but as Corollary~\ref{cor:np} assumes that the side signature is
fixed, they cannot be part of the side signature. A variant of the construction
of the proof (to show $\exptime$-hardness on an unbounded signature arity) would be
to code configurations as tuples of elements $X_1 \ldots X_n$ and write, e.g.,
$\mathit{zero}(X_i)$. However, the constant width bound on $\incd$s would then mean
that the proof construction can only look at a constant number of cells when
creating one configuration from the previous one.

\medskip

We now turn to the proof of  Theorem~\ref{thm:idreduce}.

We say that a full GTGD is \emph{single-headed} if it has only one head atom,
and we will preprocess the input full GTGDs in~$\ptime$ to ensure this condition.
For every full GTGD, we introduce a new relation of arity at most~$h$ to stand for its
head, and we add the full GTGDs which assert that the new head relation creates
every fact in the original head. This is $\ptime$, and the resulting set of full
GTGDs is single-headed and still satisfies the constraints, it does not change
the side signature, and the width cannot increase too much: the bound $w$ on
width is changed by $\max(w, h)$, which is not a problem for the complexity
bounds that we are claiming. 
Hence, we perform this transformation, and throughout the appendix, when we refer to \emph{full GTGDs}, we always
assume that they are single-headed, and when we refer to the bound on width we
assume that it has been changed accordingly. However, we will still need to refer to the
head arity bound in the sequel, to bound the arity of the
single relation in the head of full GTGDs: note that the transformation
described in this paragraph cannot change the head arity bound.

Now, the intuition of the proof is essentially to follow the process given
in  Appendix~\ref{apx:linearizeaccessids}
for linearizing bounded-width $\incd$s and  truncated accessibility axioms.
Specifically, our proof strategy consists of three steps.
We first show that we can compute a form of
\emph{closure} of our $\incd$s and full GTGDs for a bounded domain of side
atoms, which we call bounded \emph{breadth} (generalizing the notion in
Appendix~\ref{apx:linearizeaccessids}). This intuitively ensures
that, whenever a full GTGD generates a fact about earlier elements, then this
generation could already have been performed when these earlier elements had
been generated, using an implied full GTGD. This first step is the main part of
the proof, and its correctness relies of a notion of \emph{well-ordered} chase
that we introduce, generalizing the analogous notion presented in
Appendix~\ref{app:orderability-accessids} for the specific case of
accessibility axioms.

Once this closure has been
done, the second step is to structure the chase further, by enforcing that we only fire the
full GTGDs and their small-breadth closure just after having fired an $\incd$. As in the
earlier proof, we call this 
the \emph{short-cut chase}, since we short-cut certain derivations that go up and down
the chase tree via the firing of derived axioms.
The third step is to argue that the short-cut chase can be
linearized with $\incd$s, again generalizing the previous constructions.

We now embark on the proof of Theorem~\ref{thm:idreduce}, which will conclude at
the end of this section of the appendix.

\myparaskip
\myparagraph{Bounded breadth closure}
The first step of our proof is to show how to compute a closure of the constraints
$\Sigma$.
We will now consider the
side signature $\sidesign$ with its fixed arity bound~$a'$, and will consider the bounds $w, h$ on
width and head arity respectively. We will reason about full GTGDs with
side signature $\sidesign$ that obey a certain \emph{breadth} restriction:

\begin{definition}
  Let $b \in \NN$, let $\sidesign$ be a side signature, and
  let $\gdep$ be a full GTGD on side signature~$\sidesign$.
  We say that~$\gdep$ has
  \emph{breadth $\leq b$} if there exists a guard atom $A$ in the body of~$\gdep$
  and a subset $P$ of at most $b$ positions of~$A$ such that, letting $X$ be the
  $\leq b$ variables that occur at the positions of~$P$ in~$A$, the following
  \emph{variable occurrences condition} holds:
  \begin{itemize}
    \item the other body atoms than~$A$, which use relations of~$\sidesign$,
      only use variables of~$X$;
    \item the other positions of~$A$ use variables that are not in~$X$ and occur
      only once in the body (i.e., there are no variable repetitions)
  \end{itemize}
  In other words, $\gdep$ can then be written:
  \[
    A(\vec x, \vec y) \wedge \phi(\vec x) \rightarrow \exists \vec z ~ H(\vec
    x, \vec y, \vec z)
  \]
  where
  $A$ denotes the guard atom, 
  $\phi$ denotes the other body atoms (which use relations
  of~$\sidesign$), $H$ denotes the head atom,
  $\vec x$ denotes the variables of~$X$, the variables of~$\vec y$ in the
  body are not repeated, the variables $\vec x$ occur in at most $b$
  positions in~$A$.
\end{definition}

It will be useful to reason about the
possible full GTGDs on the side signature $\sidesign$ that satisfy
the head arity bound~$h$ and have breadth at most 
the width $w$ of the $\incd$s.
We will call such full GTGDs the 
\emph{suitable full GTGDs}.

\begin{lemma}
  \label{lem:numbergtgds}
  The  number of suitable
full GTGDs of breadth $\leq b$ and of head arity $\leq h$
  is at most 
\[ 
  n \cdot (a+1)^b \cdot b^b \cdot (2^{n' \cdot b^{a'}}) \cdot n \cdot a^{h}
\]
 where: $n$ is the number of relations
in the full signature~$\sign$, $g$ is the number of atoms in bodies
  of~$\Sigma$,
  $a$ is the maximal arity of any relation
in~$\sign$,  $n'$ is the number of
relations in~$\sidesign$, and
$a'$ is the maximal arity of the relations of~$\sidesign$.
\end{lemma}

\begin{proof}
  To compute this, we illustrate how to choose a full GTGD satisfying the
  condition. Several choices will lead to the same GTGD, but this is not a
  problem as we are trying to derive an upper bound.

  Here are the choices that we have to make:

  \begin{itemize}
    \item Choose the relation for the guard atom: factor of~$n$.
    \item Choose the subset $P$ of size $\leq b$: factor of $(a+1)^b$ because an
      upper bound on this is to choose $b$ positions, with replacement, and with
      the option of choosing a dummy position for subsets of size $< b$.
    \item Choose the pattern of variables on these $P$ positions: factor of
      $b^b$ because an upper bound on this is to use a set of $b$ variables and
      decide which variable to put at each of the (at most) $b$ positions
      of~$P$.
    \item Choose the side facts, i.e., a subset of the possible facts: the
      corresponding factor is $2^N$  where $N$ is the number of possible facts,
      which are obtained by choosing a relation of the side signature and
      choosing a variable (among $b$) for each position (among $a'$ at most).
    \item Choose the head relation: factor of~$n$ for the choice of relation,
      and $a^h$ to decide, for each of the $h$ positions of this relation, which
      body variable to use (of which there are at most~$a$).
  \end{itemize}

  Hence, the result is indeed the quantity given in the lemma statement.
\end{proof}

Observe that, when $b \colonequals w$, when $h$, $n'$, and $a'$ are bounded, then the above quantity
is polynomial in the input signature~$\sign$. Further, when only $a'$
is bounded, then the quantity is singly exponential in the input.

Letting $\Sigma$ be our set of constraints (with non-full $\incd$s of bounded
width and full GTGDs), we say that a full GTGD $\gdep$ is
a \emph{derived suitable} full GTGD if it is suitable and 
if $\gdep$ is entailed by~$\Sigma$:
that is, any instance
that satisfies $\Sigma$ also satisfies $\gdep$.
Note that the derived suitable GTGDs do \emph{not} include all the full GTGDs that we
started with, because some of them may have breadth larger than~$w$, so they are
not suitable. The same was true in the previous proof: some original truncated
accessibility axioms were not completely reflected in the derived truncated
accessibility axioms, but the width bound ensured that this did not matter
except on the initial instance: this will also be the case here.
The second step of our proof of 
Theorem~\ref{thm:idreduce} is to show that the set of derived suitable full GTGDs can be
computed efficiently.

\begin{definition}
  We say that a suitable GTGD $\gdep$ is \emph{trivial} if its head atom
  already occurs in its body.
  Given a set of non-full $\incd$s and 
full GTGDs $\Sigma$, the
  \emph{$b$-closure} $\wclo{\Sigma}{b}$ is obtained by starting with
  the suitable GTGDs in~$\Sigma$ plus the
  trivial suitable GTGDs, and applying the following inference rules until we
  reach a fixpoint:
\begin{itemize}
\item (Transitivity): Suppose that there is a GTGD body $\beta: 
  R(\vec x, \vec y) \wedge \bigwedge_i A_i(\vec x)$ and heads
    $B_1(\vec z_1) \ldots B_n(\vec z_n)$ such that,
    for each $1 \leq j \leq n$,
    the 
    GTGD $\beta \rightarrow B_j(\vec z_j)$ is suitable and
    is in~$\Sigma \cup \wclo{\Sigma}{b}$.
    Suppose that there is a GTGD $\beta' \rightarrow \rho$ 
    in~$\Sigma \cup \wclo{\Sigma}{b}$, and 
    that there is a unifier $\upsilon$ mapping 
    $\beta'$ to~$\beta \wedge \bigwedge_j B_j(\vec z_j)$.
    Then add to~$\wclo{\Sigma}{b}$ the following:
    \[ \beta \rightarrow \upsilon(\rho).
\]
  \item ($\incd$): Suppose we have an $\incd$ $\dep$ in~$\Sigma$ from~$R(\vec x)$ to~$S(\vec y)$
    of width $w' \leq b$, which exports $x_{j_i}$ to~$y_{k_i}$ for~$1 \leq i
    \leq w'$.
     Suppose that~$S(\vec z, \vec w) \wedge \phi(\vec w) \rightarrow
     H(\vec w)$ is
     a GTGD in~$\Sigma \cup \wclo{\Sigma}{b}$, with the $S$-atom being
     a guard, such that the variables of~$\vec w$ in the~$S$-atom occur only at
     exported positions of the head of~$\dep$, and the variables of the $S$-atom at
     positions not exported in the head of~$\dep$ contain variables 
     with no repetitions.
     Then add to~$\wclo{\Sigma}{b}$ the following:
     \[R(\vec z, \vec w) \wedge \phi(\vec w) \rightarrow H(\vec w)\]
     where the $R$-atom is obtained from the $S$-atom by backwards rewriting
     via~$\dep$, i.e., its positions at exported positions in the body of~$\dep$
     contain the variables of the original GTGD at the corresponding exported
     position in the head of~$\dep$, and the other positions of the $R$-atom
     contain fresh variables with no repetitions.
\end{itemize}
\end{definition}

We claim that if this procedure is performed with~$b$ set to the
width bound~$w$, then all the GTGDs of~$\wclo{\Sigma}{b}$ are suitable (this is
clear), and that in fact
$\wclo{\Sigma}{b}$ contains all derived GTGDs:

\begin{proposition} \label{prop:derived}
  For any set $\Sigma$ of non-full $\incd$s of width $\leq w$ and full GTGDs on side
  signature $\sidesign$ with head arity $\leq h$,
  then $\wclo{\Sigma}{w}$ is
  the set of derived suitable full GTGDs.
\end{proposition}

Before we prove this claim, we state and prove that the computation of the
$b$-closure can be performed efficiently:

\begin{lemma}
  For any set $\Sigma$ of non-full $\incd$s of width $\leq w$ and full GTGDs 
  on side signature $\sidesign$ with head arity $\leq h$,
  letting $b\colonequals w$,
  then $\wclo{\Sigma}{b}$ is
  we can compute
  $\wclo{\Sigma}{b}$ in polynomial time in~$\card{\Sigma} \times
  2^{\textrm{polynomial}(w,h)}$.
\end{lemma}

\begin{proof}
  From our bound in Lemma~\ref{lem:numbergtgds},
  we know that the maximal size of~$\wclo{\Sigma}{b}$ satisfies our running time
  bound. We can compute it by iterating the possible production of rules until
  we reach a fixpoint, so it suffices that at each intermediate state of~$\wclo{\Sigma}{b}$,
  testing every possible rule application is in PTIME in the current
  $\wclo{\Sigma}{b} \cup \Sigma$.

  For ($\incd$), this is straightforward: we simply try every $\incd$ from~$\Sigma$ and
  every rule from the current $\wclo{\Sigma}{b}$ unioned with~$\Sigma$,
  and we check whether we can
  perform the substitution (which is clearly PTIME), in which case we add the
  result to~$\wclo{\Sigma}{b}$.

  For (Transitivity), we enumerate all possible bodies $\beta$ of a suitable
  GTGD: there are polynomially many, up to variable renamings. For each~$\beta$,
  we then find all rules in~$\Sigma$ and in~$\wclo{\Sigma}{b}$ that have body~$\beta$
  up to variable renaming: if~$\beta$ only contains $\sidesign$-atoms, this is
  easy, because it is guarded and the arity of~$\sidesign$ is constant so the body is
  on a constant-size domain and we can just test the homomorphism; if~$\beta$ contains a
  guard atom not in~$\sidesign$, by assumption there is only one, and we can just
  consider the GTGDs whose body contains this atom, try to unify it with~$\beta$
  in PTIME, and check if the candidate body achieves exactly the side atoms
  of~$\beta$. Once all suitable rules are identified, clearly we can take $1
  \leq j \leq n$ to range over all such GTGDs, and consider the union~$H$ of their
  heads. Now, we enumerate all GTGDs in~$\Sigma \cup \wclo{\Sigma}{b}$ and we
  must argue that we can test in PTIME whether their body $\beta'$ unifies to
  $\beta \cup H$. If $\beta'$ contains only atoms
  from~$\sidesign$, then as it is guarded and the arity of~$\sidesign$ is fixed,
  its domain size is
  constant, so we can simply test in PTIME all possible mappings of~$\beta'$ to
  see if they are homomorphisms. If $\beta'$ contains an atom $A$ not in~$\sidesign$,
  by assumption there is only one and it is a guard, so we can simply consider
  all atoms in~$\beta \cup H$ which are not in~$\sidesign$: for each of them, we
  test in PTIME whether $A$ unifies with it, and if yes we test whether the
  mapping thus defined is a homomorphism from~$\beta'$ to~$\beta\cup H$. If yes,
  we add the new full GTGD to~$\wclo{\Sigma}{b}$. This concludes the proof.
\end{proof}

There remains to prove Proposition~\ref{prop:derived}. For this, we will need
additional machinery.

\myparaskip
\myparagraph{Well-ordered chase}
To prove Proposition~\ref{prop:derived}, we will need to study the chase by $\incd$s
and full GTGDs. We will first define a notion of \emph{well-ordered} chase
(which generalizes the notion studied in Appendix~\ref{app:orderability-accessids}), show
that we can ensure that the chase satisfies this condition, and conclude the proof of the proposition.

In the chase, we will 
distinguish between the \emph{$\incd$-facts}, which are the facts created in the
chase by firing an $\incd$, and the \emph{full facts}, the ones created by
firing a
full GTGD.
We further observe by an immediate induction that, for each full fact $F$ generated in the chase,
the guard of the trigger $\trig$ used to generate $F$ must
be guarded by some $\incd$-fact: this is vacuously preserved when firing an $\incd$, and
it is preserved when firing a full GTGD $\dep$ because the $\incd$-fact that guards the
guard of~$\dep$ also guards the generated head fact.

This allows us to define  a tree structure on the $\incd$-facts created in the
chase. Initially the tree consists of just a root, which is a special node
containing all facts $I_0$ of the original instance. Whenever we fire a trigger
$\trig$ 
in the chase to create a new fact $F$, we define the \emph{$\incd$-guard} of the
new fact~$F$
as the $\incd$-fact $G$ that guards $\trig$ and is the
topmost one in the chase. Observe that~$G$ is uniquely defined, because whenever
two $\incd$-facts guard $\trig$ then their lowest common ancestor also does.
If $F$ is an $\incd$-fact, then we create a new node in the chase tree that
contains $F$, with its $\incd$-guard $G$ being its \emph{parent}. If $F$ is a full
fact, we do not represent it in the chase tree.

For a chase that is structured into a tree as described above, 
we further say that it is  \emph{well-ordered}
if it satisfies the following
condition:

\begin{quote}
Whenever we create an $\incd$-fact $F = R(\vec c, \vec d)$, where $\vec c$
are the elements shared between $F$ and its parent fact,
there is at most one fact $H$ that uses only elements of~$\vec c$ such that the
following is true: $H$ will be created later in the chase by firing a trigger
whose $\incd$-guard is a descendant of~$F$.
\end{quote}

We can now show the analogue of Lemma~\ref{lem:wellaccessids}. The specific proof
technique is different (we create multiple child facts at once instead of
re-firing dependencies as needed), but the spirit is the same.

\begin{lemma} \label{lem:well}
  For any instance $I_0$, we can perform the chase in a well-ordered way, obtaining
an instance that satisfies the constraints.
\end{lemma}

\begin{proof}
  Fix $I_0$. We will perform the chase under full GTGDs and $\incd$s, in a
  way which chooses which triggers to fire in a special way, and instantiates
  the heads of violations for $\incd$s multiple times. Further, in this chase variant, 
  for all $\incd$-facts $F = R(\vec c, \vec d)$ generated in the chase, where we let
  $\vec c$ the elements shared between $F$ and its parent fact, the fact $F$
  will carry a \emph{subtree constraint label}, which is some fact on domain~$\vec c$ (not
  necessarily a fact which holds in the chase). Intuitively, this fact will be
  the one additional fact on~$\vec c$ which is allowed to be created in the
  subtree. Further, for any $\incd$-fact in the chase, there will be an equivalence
  relation on its children. We will inductively impose that, for any two sibling
  $\incd$-facts $F_1$ and $F_2$ that are equivalent, 
  at any state of the chase, 
  there is an isomorphism between the restriction of the chase to the domain of
  the subtree rooted at~$F_1$, and that of~$F_2$, which is the identity on the
  elements shared between $F_1, F_2$ and their parent fact $F$. In particular,
  $F_1$ and $F_2$ are facts of the same relation and share the same elements at
  the same positions with~$F$.

  We now explain how to perform the chase in a way which does not violate
  well-orderedness and satisfies this inductive invariant.

  Whenever we fire an $\incd$ in a chase proof, we cannot violate the
  well-orderedness property
  because all $\incd$s are non-full. We will explain how we change the usual
  definition of chase step to create multiple equivalent facts by instantiating 
  the heads of violations multiple times. Suppose we want to fire an
  $\incd$ $\dep$
  to create a fact $R(\vec c, \vec d)$, where $\vec c$ are the elements that are
  not fresh. Let $F_0$ be the $\incd$-guard of the new fact. The fact $F_0$ shares
  precisely~$\vec c$ with the new fact, and it be the parent of the new fact in
  the chase tree.
  We consider all
  possible facts over~$\vec c, d$ which include an element that $F_0$ does not
  share with its own parent, this requirement being vacuous if $F_0$ has no
  parent. We call these the \emph{$F_0$-native facts over~$\vec c$}.
  For every such fact~$F'$, we create a copy
  of the fact~$R(\vec c, \vec d)$ as a child of $F_0$, with subtree constraint label $F'$.
  We call these facts the \emph{alternative copies}.
  Now, if $F_0$ has a parent, we let $F''$ be the subtree constraint label
  on~$F_0$. If all elements of~the fact $F''$  still occur in the fact $R(\vec c, \vec d)$, we
  also create another copy of $R(\vec c, \vec d)$  as a child of $F_0$, whose
subtree constraint label is ~$F''$. We call
  this the \emph{propagating copy}.
  The intuition for the propagating copy is that it is the root of the child
  subtree where we can still create the fact $F'$ which is the subtree constrain label of~$F_0$;
  and the intuition for the alternative copies is that they are the child
  subtrees where we can generate different facts: the $F_0$-native facts
  over~$\vec c$ are defined in a way that ensures that generating these facts
  will not violate the well-orderedness requirement.

  It is clear that 
  creating these multiple copies does not impact soundness of the chase, and their addition
will not impact completeness provided that we also handle triggers on these new facts.
Further, for any pair of these one-fact
  subtrees, there is an isomorphism of the chase at the time they are fired that sends one to another.
 We must also argue that by firing further rules following this ``$\incd$  duplication step''
  we can avoid breaking the isomorphism between subtrees
  rooted at equivalent $\incd$-facts elsewhere in the chase. But  these
  subtrees were  isomorphic prior to this step, by induction; so we can simply perform an
  analogous collection of chase steps in all of these subtrees. This ``follow-up''
does not cause any new violations of our invariant, and also does not break well-orderedness,
again since the IDs are all non-full.

  Whenever we fire a full GTGD $\dep$ in a chase proof on a trigger~$\trig$ with $\incd$-guard $G$ to create a
  fact $H$, we consider the topmost ancestor $F$ of~$G$ that contains all
  elements of~$H$. Note that the topmost ancestor exists because the set of suitable
  ancestors is non-empty: indeed, $G$ is itself a
  suitable choice. Let $C$ be the child of~$F$ which is an ancestor of~$G$. We
  claim that $F$ has a child $C'$ which is equivalent to~$C$ and which has subtree constraint
  label $H$. This is because $F$ is the topmost ancestor containing all
  elements of~$H$, so there is an element in~$H$ which does not occur in the parent of~$F$ (or
  $F$ is the root).
Thus
 when we created~$C$ by firing an $\incd$ trigger on~$F$,  we know that
  $H$ was an $F$-native fact over the elements shared between $F$ and~$C$,
  which include all elements of~$H$. Hence, 
 we have created among the alternative copies a sibling $C'$ of $C$ which is
 equivalent to~$C$ and which has subtree constraint label $H$.
Now, as the subtrees rooted
  at~$C$ and~$C'$ are isomorphic, we can consider the path from~$C$ to~$G$
  and follow the image of this path starting at~$C'$ to reach  $\incd$-facts which
  are isomorphic to~$G$. When doing this we have to choose between equivalent
  children, and we always choose a child with subtree constraint label~$H$. Since we create a propagating
copy in our modified $\incd$ step above,  such children must exist. Thus we have
a path from $C$ to some $G'$ in which every node has subtree constraint label $H$.

  The isomorphism between the subtrees ensures that 
  $G'$ is the $\incd$-guard to a trigger $\trig'$ for $\dep$ which is isomorphic
  to the original trigger~$\trig$,
Hence, instead of firing the full GTGD~$\dep$ on the trigger~$\trig$, we fire
$\dep$ on the trigger $\trig'$.
  This will also create $H$, as the elements of~$H$ are shared between the subtrees.
This firing does not
  violate the well-orderedness condition: all possible choices for a
  counterexample~$F$
  are on the path from~$G'$ to~$F$, so they have subtree constraint  label $H$, and 
indeed $H$ is the only fact that we have created in the firing. 
  Further, given the domain of~$H$, creating it does not break
  the inductive isomorphism condition for any subtrees rooted at an $\incd$-fact
  below~$F$. 

We will now fire additional rules to preserve the isomorphism for subtrees
  rooted at ancestors of~$F$. 
  Consider any ancestor~$F'$ of~$F$, let $F'_1$ be the child of~$F'$ which is an
  ancestor of~$F$, and let $F'_2$ be another child of~$F'$ which is equivalent
  to~$F'_1$. We know that the subtrees rooted at~$F'_1$ and~$F'_2$ were
  isomorphic before the firing, so there is a trigger in~$F'_2$ which is the
  image by the isomorphism of the trigger
  that we are firing. Firing this trigger creates a fact which allows us to
  extend the isomorphism to account for the firing. We do this for every choice
  of ancestor~$F'$ and every choice of child $F'_2$ which is equivalent to the
  child~$F'_1$.
This ensures that we can extend
  the isomorphisms between subtrees rooted at equivalent nodes: this does not
  violate the well-orderedness, thanks to the isomorphism.
Note that
all the triggers that we fire in this way, and all the facts that we create, are pairwise distinct,
   because at least one element of
  the new fact~$H$ does not appear outside of the
  descendants of~$F$, i.e., none of the isomorphisms that we consider can be the
  identity on this element.

  We have described a chase variant that produces a well-ordered chase proof, so
  we have established the desired result.
\end{proof}

We are now ready to prove Proposition~\ref{prop:derived}:

\begin{proof}
  One direction is straightforward: we can immediately show by induction on the
  derivation that any GTGD produced in the closure is indeed a derived suitable
  full GTGD. Hence, we focus on the converse direction.

  We prove that every derived suitable full GTGD is produced in the closure, by induction on
  the length of a well-ordered chase proof of its head. Specifically,
  let the derived suitable full GTGD be $\gdep: R(\vec x, \vec y) \wedge \left(\bigwedge_i
  A_i(\vec x)\right)
  \rightarrow \lambda(\vec z)$, with~$\vec z \subseteq \vec x \cup \vec y$. Let
  $K_0$ be a set of constants, and $I_0$ be the initial instance which consists
  of the instantiation of the body of~$\gdep$ on~$K_0$: we let $\vec a_0, \vec
  b_0, \vec c_0$ be the tuples of~$K_0$ corresponding to~$\vec x$, $\vec y$,
  $\vec z$, and call $F_0$ the instantiation of the guard atom (which we will
  see as the root $\incd$-fact  in the chase). We show that~$\gdep$ is
  produced in the closure by induction on the number of chase steps required in
  a well-ordered chase proof to produce $\lambda_0 \colonequals \lambda(\vec
  c_0)$ from~$I_0$. The fact that we can assume a well-ordered chase proof is
  thanks to Lemma~\ref{lem:well}.

  The base case is when there is a well-ordered chase proof of length~$0$, i.e,.
  $\lambda_0$ is in~$I_0$. In this case, $\gdep$ is a trivial suitable GTGD, so it is
  in the closure by construction.

  We now show the induction case. We consider a well-ordered
  chase proof that produces $\lambda_0$ in as little chase steps as possible.
  The firing that produces $\lambda_0$ cannot be the firing of an $\incd$, because
  the $\incd$s are non-full, so they produce facts that contain some null whereas
  $\lambda_0$ is a fact on~$K_0$. Hence, $\lambda_0$ is produced by firing a
  full GTGD $\gdep'$ on a trigger~$\trig$.
  Remember that the $\incd$-guard of this firing is the $\incd$-fact that
  guards~$\gdep'$ and is the topmost one in the chase: let $F$ be the
  $\incd$-guard.
  Either $F = F_0$ or $F$ is a
  strict descendant of~$F_0$.

  If $F$ is $F_0$, then it means that~$F_0$ guards this firing. Hence, all facts
  of~$\trig$ are facts over~$K_0$. Hence, those which are $\incd$-facts 
  cannot be
  another fact than~$F_0$, as the other $\incd$-facts  contain nulls not in~$K_0$.
  As for those that are full facts,
  for each such fact $\phi(K_0)$, we know that
  $\phi(K_0)$ was derived in the chase
  from~$I_0$, which means that the GTGD
  $\gdep_\phi: R(\vec x, \vec y) \wedge \left(\bigwedge_i A_i(\vec x)\right) \rightarrow \phi(\vec z)$ is
  a derived full suitable GTGD, in particular it satisfies the variable
  occurrences condition because $\gamma$ does. Now, for each such $\phi(K_0)$,
  as it is produced earlier than~$\lambda_0$ in
  the chase, it means that there is a well-ordered chase that produces
  it in strictly less steps. Hence, by induction hypothesis, each $\gdep_\phi$ is
  in the closure. We can now see that, as~$\gdep' \in \Gamma$, as the
  $\gdep_\phi$ are in the closure when~$\phi(K_0)$ is a full fact, and as the
  other $\phi(K_0)$ must be $F_0$,
  we can apply the (Transitivity) rule and
  conclude that~$\gdep$ is also in the closure.

  Now, if the ID-guard $F$ of the firing is a strict descendant of~$F_0$ in the
  chase tree, we consider the path in the chase from~$F_0$ to~$F$.
  Let $F'$ be the $\incd$-fact which is the first element of
  this path after~$F_0$: it is a child of~$F_0$ and an ancestor of~$F$.
  Let $\vec c$ be the elements shared between $F'$ and $F_0$. Observe that
  $\lambda_0$ is a fact on~$\vec c$, because it is a fact on~$K_0$ with
  ID-guard~$F$, so the elements of~$\lambda_0$ are shared between
  $F$ and~$F_0$, hence between~$F'$ and~$F_0$. Now, the deduction of~$\lambda_0$
  creates a new fact on~$\vec c$ with ID-guard in the subtree rooted at~$F'$: by
  definition of the well-ordered chase, it is the only such firing for~$F'$.
  Thus, we know that, when we create~$F'$, we had already derived all facts on
  $\vec c$ that are derived in the chase, except for
  $\lambda$. Let $\Phi$ be the set of these facts.

  We now observe that, if we had started the chase with the $\incd$-fact $F'$ plus
  the set $\Phi$, then we would also have deduced $\lambda$. Indeed, we can
  reproduce all chase steps that happened in the subtree rooted at~$F'$,
  specifically, all steps where we applied IDs to a descendent of~$F'$, and 
  all full GTGD steps with ID-guard in the subtree rooted at~$F'$. We show this
  by induction: the base case corresponds to the facts of~$\{F'\} \cup \Phi$,
  the induction step is trivial for ID applications, and for full GTGD
  applications we know that all hypotheses to the firing are in guarded tuples
  of the subtree rooted at~$F'$, so they were all generated previously in that subtree or were part of~$\Phi$.
  Thus, letting $\beta'$ be the result of renaming the constants of~$\{F'\}\cup
  \Phi$ by variables in a manner compatible with the mapping from~$\gdep$ to~$I_0$,
  this shows that~$\gdep' : \beta' \rightarrow \lambda$ is entailed by~$\Sigma$. Now,
  this is a GTGD with breadth at most~$w$
  because the width of IDs (and hence the width of~$\Phi$) is at most~$w$.
  Hence, it is a derived suitable full GTGD, and 
  the proof for this derived suitable full GTGD is shorter than that of~$\gdep$.
  Hence, by induction
  hypothesis, we have $\gdep' \in \wclo{\Sigma}{b}$.
  We can now apply ($\incd$) to~$\gdep'$ with the IDs that generated $F'$
  from~$F_0$, thanks to the fact that~$\Phi$ and $\lambda_0$ are on~$K_0$ so
  they are on exported positions of the $\incd$,
  and there cannot be repeated variables in what corresponds to~$F'$ in~$\beta'$
  except at positions that were exported between $F'$ and~$F_0$ (because these
  other elements are fresh in~$F'$).
  This yields $\gdep'': R(\vec x,
  \vec y) \wedge \Phi'(\vec z) \rightarrow \lambda(\vec z)$, where $\Phi'$  is
  the result of renaming the elements of~$\Phi$ to variables as in the
  definition of~$\beta'$.

  We now argue as in the base case that, as each fact of~$\Phi$ was derived from
  $I_0$ with a shorter proof than the proof of~$\gdep$, by induction
  hypothesis, for each $\phi \in \Phi$, the derived suitable full GTGD $R(\vec x, \vec
  y)\wedge \left(\bigwedge_i A_i(\vec x)\right) \rightarrow \phi$ is in~$\wclo{\Sigma}{b}$.
  We conclude, by applying (Transitivity) to~$\gdep''$ and these derived suitable full GTGDs,
  that~$\gdep \in \wclo{\Sigma}{b}$.

  Hence, we have shown that~$\gdep$ was derived, which concludes the induction
  and finishes the completeness proof.
\end{proof}

\myparagraph{Normalization}
We are now ready for the second stage of our proof: normalizing the chase to add \emph{short-cuts}.
The \emph{short-cut chase} works by constructing \emph{bags}:
a \emph{bag} is a set of facts consisting of one fact generated by an $\incd$ (called
an \emph{$\incd$-fact})
and facts on the domain of the $\incd$-fact
(called \emph{full facts}). We will have a tree structure
on bags that corresponds to how they are created.

The short-cut chase then consists of two alternating kinds of steps:
\begin{itemize}
  \item The \emph{$\incd$ steps}, where we fire an $\incd$ on an $\incd$-fact $F$ where it is
    applicable. Let $g$ be the bag of~$F$. The $\incd$ step creates a new bag~$g'$
    which is a child of~$g$, which contains the result $F'$ of firing the $\incd$,
    along with a copy of the full facts of~$g$ which only use elements shared
    between $F$ and~$F'$.
  \item The \emph{full saturation steps}, which apply to a bag~$g$, only once
    per bag, precisely at the moment where it is created by an $\incd$ step (or on
    the root bag). In this step, we apply all the full GTGDs of
    $\wclo{\Sigma}{w}$
    to the facts of~$g$, and add the consequences to~$g$ (they are still on the domain of~$g$
    because the rules are full).
\end{itemize}

Our goal is to argue that the short-cut chase is equivalent to the usual chase.
The short-cut chase consists of chase steps in the usual sense, so it is still
universal. What is not obvious is that the infinite result of the short-cut chase
satisfies $\Sigma$: indeed, we must argue that all violations are solved. This
is not straightforward, because the short-cut chase does not consider
all triggers: specifically, for full GTGDs, it only considers triggers that are
entirely contained in a bag, and only for full GTGDs in~$\wclo{\Sigma}{w}$ (not
those in~$\Sigma$). So what we must do is argue that the short-cut
chase does not leave any violations unsolved. The intuition for this is that the
closure of~$\wclo{\Sigma}{w}$ suffices to ensure that all violations can be seen
within a bag.

To show this formally, we will rely on the following observation, which uses
the closure~$\wclo{\Sigma}{w}$. In the statement of this lemma, we talk of the
\emph{topmost bag} that contains a guarded tuple: it is obvious by considering
the domains of the bags that this is well-defined:
\begin{lemma} \label{lem:normalizedclosure}
  Consider the short-cut chase on an instance $I_0$ which is closed under the
  full GTGDs of~$\Sigma$ and of~~$\wclo{\Sigma}{w}$ (we see $I_0$ as a root bag).
  Assume that a full saturation step on some bag~$g$ creates a fact $F$. Then
  $g$ is the topmost bag of the chase tree that contains the elements of~$F$.
\end{lemma}

\begin{proof}
  Let $g'$ be the topmost bag~$g'$ of the chase tree that contains all
  elements of~$F$, and let us show that $F$ was created in~$g'$.
  Consider the moment in the chase where $g'$ was created in
  an $\incd$ step. At this moment, $g'$ consists of its $\incd$-fact plus full facts
  copied over from the parent of~$g'$ (or none, in the case where $g'$ is the
  root): as the $\incd$s have width $\leq w$, these full facts are on a domain~$\vec
  c$ of size at most~$w$. We let $\beta$ be the set of these facts.
  As the short-cut chase proceeds entirely downwards in the tree,
  and it constructs a subset of the usual chase, by starting the chase
  with a root bag containing $\beta$, 
  we know that~$F$ is deduced. Hence, letting $\beta'$ and $F'$ be the result of
  renaming the elements of~$\beta$ and $F$ to variables, we know that the
  GTGD $\gdep : \beta' \rightarrow F'$ is entailed by $\wclo{\Sigma}{w}$. 
  Hence, $\gdep$ is a derived suitable full GTGD, so we must have $\gdep \in
  \wclo{\Sigma}{w}$. Thus, we have also applied $\gdep$ in the full saturation
  step just after the moment where we created $g'$; or, if~$g'$ is the root
  bag, it is closed under $\Sigma$ and under~$\wclo{\Sigma}{w}$ by hypothesis.
  Hence, we have shown that $F$ was indeed created in~$g'$, which concludes.
\end{proof}

This immediately implies the following:

\begin{corollary}
  \label{cor:allbags}
  For any fact $F$ created in the short-cut chase, for any bag~$g$ containing
  all elements of~$F$, then $F$ appears in the bag~$g$.
\end{corollary}

\begin{proof}
  If $F$ is an $\incd$-fact, it contains a null, the topmost bag~$g$ containing all
  elements of~$g$ is the bag where this null was introduced, so it also
  contains~$F$. Now, the bags that contain the elements of~$F$ form a subtree of
  the tree on bags rooted at~$g$, and $F$ is copied in all these bags.

  If $F$ is a full fact, we use  Lemma~\ref{lem:normalizedclosure} to argue
  that~$F$ occurs in the topmost bag containing all its elements, and again~$F$
  is copied in all other bags.
\end{proof}

This allows us to show that the short-cut chase is equivalent to the full
chase. Specifically, let $I_0$ be an arbitrary set of facts, which we consider
as a root bag, and which has been closed under the full GTGDs of~$\Sigma$ and
of~$\wclo{\Sigma}{w}$.
Let $I$ be the result of the short-cut chase of~$I_0$ by $\wclo{\Sigma}{w}$. We claim:

\begin{lemma}
  $I$ satisfies $\Sigma$.
\end{lemma}

\begin{proof}
  Consider a trigger $\trig$ in~$I$, and show that it is not active. All rules in~$\Sigma$ are guarded, so the
  domain of~$\trig$ is guarded, and there is a topmost bag~$g$ where all elements
  of~$g$ appear. By Corollary~\ref{cor:allbags}, all facts of~$\trig$ are
  reflected in all bags, in particular they are all reflected in~$g$. Hence,
  $\trig$ is included in a single bag~$g$.

  It is clear that any trigger for an $\incd$ would have been solved by an $\incd$ step,
  so we can assume that~$\trig$ is a trigger for a full GTGD of~$\Sigma$.
  The bag $g$ cannot the root bag, because we assumed that $I_0$ is closed.
  Hence, $g$ is not the
  root bag. We need to argue that~$\trig$ is a trigger for a full GTGD
  of~$\wclo{\Sigma}{w}$. Indeed, when we created~$g$ by an $\incd$
  step, $g$ contained only an $\incd$-fact plus $\sidesign$-facts on a domain of size
  at most~$w$, thanks to the fact that the $\incd$s have width at most~$w$: and
  all further facts created in~$g$ are created by the full saturation step. Hence,
  if a trigger for~$\Sigma \cup \wclo{\Sigma}{w}$ is active, it means that its
  head is entailed by the initial contents of~$g$, so the corresponding full
  GTGD is suitable because its breadth is bounded by~$b$ 
  (again using the fact that the elements of the ID-fact of~$g$ cannot be
  repeated outside of the positions that contain elements exported from the
  parent fact of~$g$).
  Hence, the trigger
  is also a trigger for the corresponding derived suitable full GTGD in~$\wclo{\Sigma}{w}$.

  Now, as we have applied a full saturation step on~$g$, any remaining trigger
  there for a full TGD of~$\wclo{\Sigma}{w}$ would have been solved in this
  step, because $\wclo{\Sigma}{w}$ is closed so it does not leave any trigger by
  full TGDs unsatisfied in~$g$. Hence, $\trig$ is no longer active in~$I$.
\end{proof}

We now know that the result $I$ of the short-cut chase satisfies~$\Sigma$, and
as the short-cut chase only applies chase steps, it is actually equivalent to the chase, i.e.,
for any CQ~$Q$, we have that~$Q$ is satisfied in~$I$ iff $I_0, \Sigma \models
Q$. All that remains now is to translate the short-cut chase to a set of $\incd$s.

\myparaskip
\myparagraph{Linearization}
We now describe the third and last stage of the proof of
Theorem~\ref{thm:idreduce}, by describing the translation.
Fix a tuple $x_1, \ldots, x_n$ of variables.
For
every relation~$R$ of arity~$l$, for every subset $P = \{p_1 \ldots p_k\}$ of its positions of size at most~$w$,
for every instance $\chi$ of the relations of~$\sidesign$
on~$x_{p_1} \ldots x_{p_k}$,
we create a copy $R_{P, \chi}$ of relation~$R$.
Observe that this
creates a singly exponential number of relations when
the arity~$a'$ of~$\sidesign$ is fixed, and it creates 
only polynomially many relations when we further fix $w$ and
$\sidesign$.
We let $\Theta$ consist of the following $\incd$s:

\begin{itemize}
  \item Forget: for every relation~$R_{P,h,\chi})$, the full~$\incd$:
    \[
      R_{P,\chi}(\vec x) \rightarrow R(\vec x)
    \]
  \item Instantiate: 
    for every relation $R_{P,h,\chi}$, 
    for every homomorphism $h$ from~$x_1, \ldots, x_l$ to itself which maps~$P$
    to~$P$ and is the identity outside of~$P$, 
    letting $\chi'$ be the instance
    on~$x_1 \ldots x_n$
    obtained by computing the closure of~$R(h(\vec x)) \cup h(\chi)$ by the full TGDs
    of~$\wclo{\Sigma}{w}$ (a full saturation step), for every fact $S(\vec y)$
    of~$\chi'$ (with~$\vec y \subseteq h(\vec x)$), we add the full GTGD:
    \[
      R_{P,\chi}(h(\vec x)) \rightarrow S(\vec y)
    \]
  \item Lift: 
    for every relation $R_{P,\chi}$,
    for every homomorphism~$h$ as above,
    letting $\chi'$ be as above,
    for every $\incd$~$\dep$ $S(\vec
    y) \rightarrow \exists \vec z ~ T(\vec y, \vec z)$,
    for every match $h'$ of~$S(\vec y)$ in~$\chi'$,
    letting $P'' = \{p_1'', \ldots, p_{k''}''\}$ be the exported positions
    of~$\dep$
    in~$\vec y$, letting $\chi''$ be the restriction of~$\chi'$ to
    $h'(y_{p_1''}), \ldots, h'(y_{p_{k'}''})$,
    letting $P'''$ be the corresponding exported positions in the head
    of~$\dep$,
    we add the full GTGD:
    \[
      R_{P,\chi}(h(\vec x)) \rightarrow \exists \vec z ~ T_{P'',\chi'''}(h'(\vec
      y), \vec z)
    \]
\end{itemize}
The result of this transformation clearly consists of linear TGDs. Further, they are
of semi-width~$w$: indeed, the Lift rules have width bounded by~$w$, and the
other rules have an acyclic position graph.
Further, these rules are clearly computed in polynomial time in their number,
and this number indeed satisfies the required bound. Indeed, letting $n$ be the
number of relations of the signature, $a$ the maximal arity of the signature,
$n'$ the number of relations in the side signature, and $a'$ the maximal arity
of the side signature, then the number of relations $R_{P,\chi}$ is no greater
than $n \cdot (a+1)^w \cdot 2^{n' \cdot w^{a'}}$, the number of Forget rules is
bounded by this number, the number of homomorphisms in Instantiate rules given
the choice of $R_{P,\chi}$ is bounded by $w^w$, so the number of Instantiate
rules is bounded by $n \cdot (a+1)^w \cdot 2^{n' \cdot w^{a'}}$ times $w^w$
times $n' \cdot a^{a'}$, and the number of Lift rules is bounded by $n \cdot
(a+1)^w \cdot 2^{n' \cdot w^{a'}}$ times $w^w$, times $(a+1)^w$ for the choice of
$P'$ times $n \cdot (a+1)^w \cdot 2^{n' \cdot w^{a'}} \cdot w^w$ for the choice
of head. Hence, having fixed $a'$,
there are indeed fixed polynomials $P_1$ and $P_2$ such that the number of
rules, and the time to construct them, is in $P_1(\card{\Sigma}, 2^{P_2(w, h,
n')}$.

We have described the construction of~$\Sigma'$. Now, to construct the rewriting
$Q^\lift$
of~$Q$, we first construct the query $Q_1$ whose canonical database is obtained
by closing the canonical database of~$Q$ under the full GTGDs of~$\Sigma$ and
of~$\wclo{\Sigma}{w}$. To see why $Q_1$ can be computed in the
prescribed bound,
we can assume that we have computed $\wclo{\Sigma}{w}$ as we already know that
it can be computed in the given time bound, so we need only reason about the
complexity of applying the rules. Now,
note that the domain size does not increase, and the number of possible new
facts is bounded by $\card{\dom(\canondb(Q))}^{h}$. Now, testing each possible
rule application is in PTIME. Indeed, it amounts to homomorphism testing between
the body of a full GTGD and~$I$, for which it suffices to consider the guard atom,
trying to map it to every fact of~$I$, and then check whether the function that
this defines is a homomorphism from the entire body to~$I$. Now we construct
$Q_1$ from $Q^\lift$ by considering every fact $R(\vec a)$ of~$Q_1$, considering
every subset of positions $P$ of~$\vec a$ of size at most~$w$, and adding the
fact $R_{P,\chi}(\vec a)$, where $\chi$ is constructed from the restriction
of~$Q_1$ to $a_{p_1}, \ldots, a_{p_k}$ for $P = \{p_1, \ldots, p_k\}$. Formally,
letting $h$ be a homomorphism from $x_{p_1}, \ldots, x_{p_k}$ to $a_{p_1},
\ldots, a_{p_k}$, we add to $\chi$ every fact $R(\vec y)$ such that $R(h(\vec
y))$ is in~$Q_1$. Note that, if $\vec a$ contains duplicate elements, then $h$
is not injective, so we add multiple copies in~$\chi$ for every fact of~$Q_1$
with domain in $\{a_{p_1}, \ldots, a_{p_k}\}$. This still respects the time
bounds, because it is doable in PTIME in~$Q_1$.

The last thing to show is that, on the resulting $Q^\lift$,
the short-cut chase is equivalent to the chase by~$\Theta$.
To see this, consider the short-cut chase where each bag is annotated by the relation for
the $\incd$-fact that created it, the subset of positions of the elements that it
shared with its parent, and the subinstance that was copied by the parent; and
consider the $\incd$-chase by the rules of the form Lift in~$\Theta$. We can observe
by a straightforward induction that the tree structure on the bags of the short-cut chase
with the indicated labels is isomorphic to the tree of facts of the form
$R_{P,\chi}(\vec x)$ created in the $\incd$ chase by the Lift rules. Now, the
application of the rules Forget and Instantiate create precisely the facts
contained in these bags, so this shows that the $\incd$ chase by~$\Theta$ and the
well-ordered chase create precisely the same facts.

This shows that~$\Theta$ satisfies indeed the hypotheses of
Theorem~\ref{thm:idreduce}, and concludes the proof.

\subsection{Proof of Corollary~\ref{cor:uidfdfinite}: Complexity of Monotone
Answerability for UIDs and FDs in the Finite}
\label{app:uidfdfinite}
Recall Corollary~\ref{cor:uidfdfinite}:
\begin{quote}
  \uidfdfinite
\end{quote}

In the body we gave an argument for decidability, but with two gaps.
The first gap is that it relied on Theorem~\ref{thm:finiteopenucq}, finite
controllability modulo the finite closure.
This result is not explicitly stated in \cite{amarilli2015finite}, which deals not with containment
but only with certain answers for CQs. So let us prove
Theorem~\ref{thm:finiteopenucq}. Recall its statement:

\begin{quote}
  \thmfiniteopenucq
\end{quote}

\begin{proof}
   It is immediate that (ii.) implies (i.). Indeed, assuming (ii.), let $I$ be a
   finite instance satisfying~$\Sigma$, then it satisfies $\Sigma^*$ by
   definition of the finite closure, and we conclude by~(ii.).

   Conversely, let us show that (i.) implies (ii.). 
  The work \cite{amarilli2015finite} shows that the finite closure of a set of UIDs and FDs admits
  \emph{finite universal models}: for each set of UIDs and FDs $\Sigma$, each
  $k \in \NN$,  and each finite instance $I$, there is a finite instance $J$
  that satisfies~$\Sigma^*$ and such that
  for every Boolean CQ~$Q'$ of size at most $k$, the following equivalence holds:
   \[I \wedge \Sigma^* \text{~implies~} Q' \text{~over all instances\quad if and
   only if\quad}
  Q' \text{~holds on~} J.\]
  Applying this to canonical databases rather than instances, we have that for all Boolean
  CQs $Q$, for all $k \in \NN$, there is a finite $J$ satisfying~$\Sigma^*$
  such that for  all Boolean CQs $Q'$ of size at most $k$:
  \[Q \subseteq_{\Sigma^*} Q' \quad \text{if and
   only if\quad} Q' \text{~holds on~} J.\]
   where, in the left-hand-side of the equivalence, the containment is over all
   instances.
   
   So let us now assume (i.) and show (ii.). Fix the UCQs $Q \colonequals
   \bigvee_i Q_i$ and $Q' \colonequals \bigvee_j Q_j'$, 
   let $k$ be $\max(\max_i \card{Q_i}, \max_j \card{Q'_j})$, and 
   consider the finite $J_i$ satisfying~$\Sigma^*$ given by the above for each~$Q_i$.
   We know that $Q_i$ holds on~$J_i$, because vacuously 
  $Q_i \subseteq_{\Sigma^*} Q_i$. Hence, $Q$ holds on every~$J_i$. Further,
  every~$J_i$ is
  finite and it satisfies~$\Sigma^*$. By point (i.), we deduce that $Q'$ holds
   on every~$J_i$, so for each $J_i$ there is a disjunct $Q'_{j_i}$ of~$Q'$ that
   holds on~$J_i$. By the
   equivalence above, we know that, for each~$i$, we have $Q_i \subseteq_{\Sigma^*}
   Q_{j_i}'$, where the containment is over all instances. Thus we can show
   (ii.): for any instance~$I$ satisfying~$\Sigma^*$, if $Q$ holds on~$I$ then
   some disjunct $Q_i$ of~$Q$ holds on~$I$, so some disjunct~$Q'_{j_i}$ of~$Q'$
   holds on~$I$, so~$Q'$ holds on~$I$. This establishes (ii.) and concludes the
   proof.
\end{proof}

A second gap in the proof is that, in the body of the paper, we argue only for decidability. We now sketch how to
 obtain the $\exptime$ bound. The naive algorithm would be to
construct the finite closure $\Sigma^*$ explicitly, and then applying the $\exptime$ algorithm for the unrestricted
case. Since the closure  is exponential in $\Sigma$, this would give a $\twoexp$ algorithm.
However, we do not need the entire closure, but only a subset $\Sigma'$ such that its closure under
unrestricted entailment is the same as its closure under finite entailment. It is known
that such  a set can be built
in polynomial time \cite{cosm,amarilli2015finite}, which establishes our
$\exptime$ bound.

\section{Proofs for Section~\lowercase{\ref{sec:general}}: General First-Order Constraints}
\subsection{Decidability of Answerability for Two-Variable Logic with Counting} \label{app:aritytwo}
As mentioned in the body, for certain classes of constraints the reduction to $\amd$ and
the formalization of $\amd$
as a query containment problem will give decidability of answerability,  even without any schema simplification results.
An example is the guarded two-variable logic with counting quantifiers, $\gctwo$.
This is a logic over relations with arity at most two, which allows assertions
such as ``for any $x$, there are at least $7$ $y$'s such that $R(x,y)$''.
The only thing the reader needs to know about $\gctwo$ is that
query containment under $\gctwo$ constraints is decidable \cite{pratt2009data}, and that if we start with $\gctwo$ constraints and
perform the reduction given by the prior results, we still remain
in $\gctwo$. 

Thus the reduction to containment immediately gives:
\newcommand{\thmdecidgctwo}{
  We can decide if a CQ $Q$ is answerable with respect to
  a schema $\aschema$ where
all relations  have arity at most~$2$ and whose constraints are expressible in
  $\gctwo$.
}
\begin{theorem}
  \label{thm:decidgctwo}
  \thmdecidgctwo
\end{theorem}

\begin{proof}
We apply Theorem~\ref{thm:equiv} to this schema
to reduce answerability to deciding
query containment with constraints. In the entailment problem for $\amd$ we have 
two
copies of the above constraints, and also the additional axioms. The additional
axioms are easily seen to be in $\gctwo$ as well.
Access monotonic-determinacy in the schema  is equivalent
to containment of~$Q'$ by~$Q$ \wrt\ the constraints, where $Q'$ is a primed copy of~$Q$.
This containment problem involves only $\gctwo$ constraints,  and thus we have
  decidability by~\cite{pratt2009data}.
\end{proof}

\subsection{Undecidability of Monotone Answerability for Equality-Free First-Order Logic and Related Languages}
\label{apx:undecid}

In this appendix, we prove Proposition~\ref{prp:undec}. Recall the statement:

\begin{quote}
\foundecid
\end{quote}

This result is true even without result bounds, and follows
from results in \cite{thebook}: we give a self-contained argument here.
Satisfiability for equality-free first-order constraints is undecidable  \cite{AHV}.
We will reduce from this to show undecidability of monotone answerability. So
let us prove
Proposition~\ref{prp:undec}:

\begin{proof}[\myproof of Proposition~\ref{prp:undec}]
Assume that we are given a satisfiability problem 
consisting of  equality-free first-order constraints $\Sigma$.
We produce from this an answerability problem where the schema has no access methods and has
constraints $\Sigma$, and we have a CQ~$Q$ consisting of a single $0$-ary relation~$A$ not mentioned
in~$\Sigma$.

We claim that this gives  a reduction from unsatisfiability to answerability, and thus
shows that the latter problem is undecidable for equality-free first-order constraints.

If $\Sigma$ is unsatisfiable, then vacuously any plan answers $Q$: since answerability
is a condition where we quantify over all instances satisfying the constraints, this
is vacuously true when the constraints are unsatisfiable because we are
quantifying over the empty set.

Conversely, if there is some instance $I$ satisfying $\Sigma$, then we let $I_1$ be formed from~$I$
by setting $A$ to be true and $I_2$ be formed by setting $A$ to be false.
$I_1$ and $I_2$ both satisfy $\Sigma$ and have the same accessible part, so they form
a counterexample  to~$\amd$. Thus, there cannot be any  monotone plan  for~$Q$. This
establishes the correctness of our reduction, and concludes the proof of
Proposition~\ref{prp:undec}.
\end{proof}

As mentioned in the body, undecidability holds also for  other logics
for which query containment is undecidable, such as general TGDs. We illustrate this
by reducing query containment with TGDs to monotone answerability.
Given an instance of the containment problem $Q \subseteq_\Sigma Q'$, where $\Sigma$ consists of TGDs,
we  can reduce it to monotone answerability of $Q$ with respect
to a schema with:
\begin{itemize}
\item constraints $\Sigma_1$ that contain $\Sigma$ 
and also $A \rightarrow Q$ as well as $Q' \rightarrow A$. 
\item only one access method,
providing input-free access to $A$.
\end{itemize}
If $Q \subseteq_\Sigma Q'$ then $Q$ can be answered just by  accessing  $A$.
Conversely, suppose that containment fails with a counterexample instance $I$ that satisfies
$\Sigma \wedge Q$ but does not 
satisfy $Q'$. Then, letting $I'$ be the empty instance,
we see that $I$ and $I'$ have the same accessible part, namely, the empty set.
But $I$ and~$I'$ disagree on $Q$. Thus $I$ and $I'$ are a counterexample to $\amd$, and hence
$Q$ is not monotone answerable.

\section{Generalization of Results to RA-Plans}
\label{apx:ra}
In the body of the paper we dealt with monotone answerability. However, at the
end of
Section~\ref{sec:prelims} and in Section~\ref{sec:conc}, we claimed that many of the  results in
the paper, including the reduction to query containment and the schema
simplification results, generalize in the ``obvious way''
to answerability where general relational algebra expressions are allowed. In addition, the results on complexity
for monotone answerability that are shown in the body extend to answerability with RA-plans, with one exception and
one caveat.
The exception is that
 we do not have a decidability result for $\uincd$s and FDs analogous to
 Theorem~\ref{thm:deciduidfd}, because
the containment problem is more complex. 
The caveat concerns answerability over finite instances. Remember that, for
monotone answerability,
all of the decidability and complexity results could be translated to the finite variant with simple arguments
based on finite controllability
(see, e.g., Proposition~\ref{prop:finitecontrol}). Doing the same for
RA-answerability would require more effort, because
we would need to verify that each construction can be adapted to preserve
finiteness: this is not obvious, e.g., for the blow-up construction.
We believe that all constructions could be adapted in this way, and we
conjecture that all
the results on RA-answerability stated in this appendix also hold for the
finite variant. Nevertheless, we leave the verification of this for future work, and in this
appendix \emph{we will only deal with answerability over unrestricted
instances}.

We explain in the rest of the appendix how to adapt our results in the
unrestricted setting from
monotone-answerability to RA-answerability, except
Proposition~\ref{prop:finitecontrol}. In the specific case of $\incd$
constraints, we will show (Proposition~\ref{prop:monotonecollapse}) that
RA-answerability and monotone answerability coincide for $\incd$s: this
generalizes a result known for views, and extends it to the setting with result
bounds.
\subsection{Variant of Reduction Results for RA-Answerability} \label{subsec:detdeffull}

We first formally define the analog of~$\amd$ for the notion of RA-answerability
that we study in this appendix. 
In the absence of result bounds, this corresponds to the notion of
\emph{access-determinacy} \cite{thebook,ustods},
which states that two instances with the
same accessible part must agree on the query result. Here we
generalize this to the setting with result bounds, where the accessible instance
is not uniquely defined.

Given a schema $\aschema$ with constraints and result-bounded methods, a query $Q$ is said to be \emph{access-determined}
 if for any two instances $I_1$, $I_2$ satisfying the constraints of~$\aschema$,
 if there is a valid access selection $\sigma_1$ for $I_1$
 and a valid access selection $\sigma_w$ for $I_2$ such that
$\accpart(\sigma_1, I_1)=\accpart(\sigma_2, I_2)$,
then  $Q(I_1)=Q(I_2)$.

As we did with $\amd$,  it will be convenient to give an alternative definition
of access-determinacy that talks only about a subinstance of a single instance.

For a schema $\aschema$ a 
 common subinstance
$I_\acc$ of~$I_1$ and $I_2$ is \emph{jointly access-valid} if, for any access performed with a method of~$\aschema$ in
$I_\acc$, there is  a set of matching
tuples in~$I_\acc$ which is a valid output to the access in~$I_1$
and in~$I_2$. In other words, there is an access selection~$\aselect$
for~$I_\acc$
whose outputs are valid in~$I_1$ and in~$I_2$.

We now claim the analogue of
Proposition~\ref{prop:altdef}, namely:

\begin{proposition}
  \label{prop:altdefra}
 For any schema $\aschema$ with constraints $\Sigma$ and result-bounded methods,
  a CQ $Q$ is access-determined if and only if the following implication holds:
for any two instances $I_1, I_2$ satisfying $\Sigma$, if $I_1$ and $I_2$ have a
  common subinstance~$I_\acc$ that is jointly
  access-valid, then $Q(I_1)=Q(I_2)$.
\end{proposition}

This result gives the alternative definition of access-determinacy that we will
use in our proofs.
The equivalence with the definition via accessible parts follows from this result:

\begin{proposition}  \label{prop:commonaccessiblera}
The following are equivalent:
  \begin{enumerate}[(i)]
\item $I_1$ and $I_2$ have a common subinstance $I_\acc$ that is jointly access-valid.
\item There is a common accessible part $A$ 
of~$I_1$ and for~$I_2$.
\end{enumerate}
\end{proposition}
\begin{proof}
Suppose  $I_1$ and $I_2$ have a common subinstance $I_\acc$ that is jointly access-valid. 
 This means that we can define an access selection
  $\aselect$ that takes any access performed with values of~$I_\acc$ and a
  method of~$\aschema$, and maps it to a set of matching tuples in~$I_\acc$ that
  is valid in~$I_1$ and in~$I_2$.
We can see that~$\aselect$ can be used as a valid access selection in~$I_1$ and $I_2$
  by extending it to return an arbitrary valid output to accesses in~$I_1$ that
  are not accesses in~$I_\acc$, and likewise to accesses in~$I_2$ that are not
  accesses in~$I_\acc$; we then have
  $\accpart(\aselect,I_1)=\accpart(\aselect,I_2)$ so we can define the
  accessible part $A$
  accordingly, noting that we have $A \subseteq I_\acc$.
Thus the first item implies the second.

Conversely, suppose that $I_1$ and $I_2$ have a common accessible part $A$,
  and let $\aselect_1$ and~$\aselect_2$ be the witnessing valid access selections
  for~$I_1$ and~$I_2$, i.e.,
$A = \accpart(\aselect_1, I_1)=\accpart(\aselect_2, I_2)$.
Let $I_\acc \colonequals A$, and let us show that $I_\acc$ is a common
  subinstance of~$I_1$ and~$I_2$ that is jointly access-valid. By definition we
  have $I_\acc \subseteq I_1$ and $I_\acc \subseteq I_2$. Now, to show that it
  is jointly access-valid in~$I_1$ and~$I_2$, consider any access
$\abind, \mt$ with values in~$I_\acc$. We know that there is $i$ such
that~$\abind$ is in~$\accpart_i(\aselect_1, I_1)$, therefore by definition of
  the fixpoint process and of the access selection~$\aselect_1$ there is a valid
  output to the access in~$\accpart_{i+1}(\aselect_1, I_1)$, hence in~$I_\acc$. 
  Thus we can choose an output in~$I_\acc$ which is valid
  in~$I_1$.
But this output must also be in~$\accpart(\aselect_2, I_2)$, and thus it is valid in
$I_2$ as well.
Thus, $I_\acc$ is jointly access-valid. This shows the converse implication and
concludes the proof.
\end{proof}

Given a schema $\aschema$ with constraints and result-bounded methods, a query $Q$ is said to be \emph{access-determined}
 if for any two instances $I_1$, $I_2$ satisfying the constraints of~$\aschema$, if $I_1$
 and $I_2$ have a common subinstance that is jointly access-valid,
then  $Q(I_1)=Q(I_2)$.

The following analogue of Proposition~\ref{prp:plantoproof} justifies the definition:
\begin{proposition}
  \label{prop:plantoproofra}
If $Q$ has a plan $\aplan$ that answers it \wrt\ $\aschema$, then
$Q$ is access-determined over~$\aschema$.
\end{proposition}
\begin{proof}
Consider instances $I_1$ and $I_2$ with a common accessible subinstance $I_\acc$
  that is jointly access-valid. Let us show that $Q(I_1) = Q(I_2)$.
We  argue that there are valid access selections $\aselect_1$ on~$I_1$,
  $\aselect_2$ on~$I_2$
and $\aselect$ on~$I_\acc$
 such that the plan~$\aplan$ evaluated with~$\aselect_1, I_1$, the plan~$\aplan$ evaluated with
  $\aselect_2, I_2$, and the plan~$\aplan$
evaluated with~$\aselect, I_\acc$ all yield the same output for each temporary
table of~$\aplan$. We prove this by induction on~$\aplan$. Inductively, it suffices to look at an
access command  $T \Leftarrow \mt \Leftarrow E$ with~$\mt$ an access method on
some relation. We can assume by induction hypothesis that $E$ 
evaluates to the same set of tuples $E_0$ on~$I_\acc$ as on~$I_1$ and $I_2$.
Given a tuple $\vec t$ in~$E_0$, consider the set $M_{\vec t}$ of ``matching tuples''
(tuples for the relation~$R$ extending $\vec t$) in 
$I_\acc$. Suppose that this set has cardinality $j$ where $j$ is strictly
smaller than the result bound of~$\mt$.
Then we can see that the set of matching tuples in~$I_1$ and in~$I_2$  must be exactly $M_{\vec t}$, and
we can take $M_{\vec t}$ to be the output of the access on~$\vec t$ in all three structures.
Suppose now $M_{\vec t}$ has size at least that of the result bound. Then the other structures
may have additional matching tuples, but we are again free
to take a subset of~$M_{\vec t}$ of the appropriate size to be the output of the
  access to~$\mt$ with~$\vec t$ in all three structures.
Unioning the tuples for all $\vec t$ in~$E_0$ completes the induction. Hence, we
know that the output of~$\aplan$ on~$I_1$ and on~$I_2$ must be equal. As we have
  assumed that $\aplan$ answers~$Q$ on~$\aschema$, this means that $Q(I_1) =
  Q(I_2)$, which is what we wanted to show.
\end{proof}

Analogously to Theorem~\ref{thm:equiv},  we can show access-determinacy is equivalent
to RA-answerability.
The proof starts the same way as that of Theorem~\ref{thm:equiv}, noting
that in  the absence of  result bounds, this equivalence was shown in prior work:
\begin{theorem}[\cite{thebook,ustods}]
  \label{thm:determandplansclassic}
  For any CQ~$Q$ and schema $\aschema$ (with no result bounds)
  whose constraints $\Sigma$ are expressible
  in active-domain first-order logic, the following are equivalent:
\begin{enumerate}
\item $Q$ has an RA  plan that answers it over~$\aschema$
\item
$Q$ is access-determined over~$\aschema$.
\end{enumerate}
\end{theorem}
The extension to result bounds is shown using the same reduction as
for Theorem~\ref{thm:equiv}, by just ``axiomatizing'' the result bounds as
additional constraints. This gives the immediate generalization of
Theorem~\ref{thm:determandplansclassic} to schemas that may include result
bounds:

\begin{theorem} \label{thm:equivra}
  For any CQ~$Q$ and schema $\aschema$ 
  whose constraints~$\Sigma$ are expressible
  in active-domain first-order logic, the following are equivalent:
\begin{enumerate}
\item $Q$ has an RA-plan that answers it over~$\aschema$
\item
$Q$ is access-determined over~$\aschema$.
\end{enumerate}
\end{theorem}

Hence, we have shown the analogue of Theorem~\ref{thm:equiv} for the setting of
answerability and RA-plans studied in this appendix.

\myparagraph{Reduction to query containment}
From Theorem~\ref{thm:equivra} we immediately get an analogous reduction
of RA-answerability to query containment. We simply need a ``more symmetrical'' version
of the auxiliary axioms.

Given a schema $\aschema$ with constraints
and access methods without result bounds, the 
 \emph{access-determinacy containment} for~$Q$ and $\aschema$ is
the CQ containment $Q \subseteq_\Gamma Q'$ where
the constraints $\Gamma$ are defined as follows: they include the original
constraints $\Sigma$, the constraints $\Sigma'$ on the relations~$R'$, and the following
\emph{bi-directional accessibility axioms} (with implicit universal
quantification):
\begin{itemize}
  \item For each method $\mt$ that is not result-bounded, letting $R$ be the
    relation accessed by~$\mt$:
    \begin{align*}
    \Big(\bigwedge_i \accessible(x_i)\Big) \wedge
      \phantom{'}R(\vec x, \vec y) \rightarrow & R_\acc(\vec x, \vec y)\\
      \Big(\bigwedge_i \accessible(x_i)\Big) \wedge
      R'(\vec x, \vec y) \rightarrow & R_\acc(\vec x, \vec y)
    \end{align*}
    where $\vec x$ denotes the input positions of~$\mt$ in~$R$.
  \item For each method $\mt$ with a result lower bound of~$k$, letting $R$ be
    the relation accessed by~$\mt$, for all $j \leq k$:
    \begin{align*}
      \Big(\bigwedge_i \accessible(x_i)\Big) \wedge \exists^{\geq j} \vec y ~
      \phantom{'}R(\vec x, \vec y)
      \rightarrow & \exists^{\geq j} \vec z ~ R_\acc(\vec x, \vec z)\\
      \Big(\bigwedge_i \accessible(x_i)\Big) \wedge \exists^{\geq j} \vec y ~
  R'(\vec x, \vec y)
      \rightarrow & \exists^{\geq j} \vec z ~ R_\acc(\vec x, \vec z)
    \end{align*}
    where $\vec x$ denotes the input positions of~$\mt$ in~$R$.
\item For every relation $R$ of the original signature:
  \[R_\acc(\vec w) \rightarrow R(\vec w) \wedge R'(\vec w) \wedge \bigwedge_i
    \accessible(w_i)\]
\end{itemize}

The only difference from the~$\amd$ containment is that the additional
constraints are now symmetric in the two signatures, primed and unprimed.
The following proposition follows immediately from Theorem~\ref{thm:equivra} and the definition of access-determinacy:

\begin{proposition} \label{prop:reducera}
For any conjunctive query $Q$ and schema $\aschema$ with constraints expressible
  in active-domain first-order logic (and possibly including result bounds), the following are equivalent:
\begin{compactitem}
\item $Q$ has an RA-plan that answers it over~$\aschema$ 
\item 
$Q$ is access-determined over~$\aschema$ 
\item 
The containment corresponding to access-determinacy holds
\end{compactitem}
\end{proposition}

\myparagraph{Elimination of result upper bounds for RA-plans}
As with monotone answerability, it suffices to consider only result lower bounds.

\begin{proposition} \label{prop:elimupperra} Let $\aschema$ be a schema with
arbitrary constraints and access methods which may be result-bounded.
A query $Q$ is answerable in~$\aschema$ if and only if it is answerable
in~$\relaxs(\aschema)$.
\end{proposition}

\begin{proof}  
  The proof follows that of Proposition~\ref{prop:elimupper}. 
  We show the result for access-determinacy instead of answerability, thanks to
  Theorem~\ref{thm:equivra}, and we use Proposition~\ref{prop:altdefra}.
  Consider arbitrary instances $I_1$ and
  $I_2$ that satisfy the constraints, and let us show that any common subinstance
  $I_\acc$ of~$I_1$ and $I_2$  is jointly access-valid for~$\aschema$ iff
  it is jointly access-valid for~$\relaxs(\aschema)$: this implies
  the claimed result.

  In the forward direction, if $I_\acc$ is jointly access-valid for~$\aschema$,
  then clearly it is jointly access-valid for~$\relaxs(\aschema)$, as any output
  of an access on~$I_\acc$ which is valid in~$I_1$ and in~$I_2$ for~$\aschema$
  is also valid for~$\relaxs(\aschema)$.

  In the backward direction, assume $I_\acc$ is jointly access-valid
  for~$\relaxs(\aschema)$, and consider an access $(\mt, \abind)$ with values
  from~$I_\acc$. 
  If $\mt$ has no result lower bound, then there is only one possible output for the
  access, and it is valid also for~$\aschema$.
  Likewise, if $\mt$ has a result lower bound of~$k$ and there are $\leq k$ matching
  tuples for the access in~$I_1$ or in~$I_2$, then the definition of a result lower bound ensures
  that there is only one possible output which is valid for~$\relaxs(\aschema)$
  in~$I_1$ and~$I_2$, and it is again valid for~$\aschema$.
  Last, if there are $>k$ matching tuples for the access,
  we let $J$ be a set of tuples in~$I_\acc$ which is is a valid output
  to the access in~$I_1$ and~$I_2$ for~$\relax(\aschema)$, and take any subset $J'$ of~$J$ with
  $k$ tuples; it is clearly a valid output to the access for~$\aschema$
  in~$I_1$ and~$I_2$. This
  establishes the backward direction, concluding the proof.
\end{proof}

\subsection{Full Answerability and Monotone Answerability}
We show that there is no difference between full answerability and monotone answerability
when constraints consist of $\incd$s only.
This is a generalization of an observation that is known for views
(see, e.g. Proposition~2.15 
in~\cite{thebook}):

\begin{proposition} \label{prop:monotonecollapse}
Let $\aschema$ be a schema with access methods and constraints $\Sigma$ consisting
of inclusion dependencies, and $Q$ be a CQ that is access-determined.
Then $Q$ is $\amd$.
\end{proposition}
\begin{proof}
  We know by Propositions~\ref{prop:elimupper} and~\ref{prop:elimupperra} that
  we can work with $\relaxs(\aschema)$ which has only 
  result lower bounds, so we do so throughout this proof.

Towards proving $\amd$, assume by way of contradiction that
  we have:
\begin{compactitem}
\item  instances $I_1$ and $I_2$ satisfying
$\Sigma$;
\item    an accessible part $A_1$ of~$I_1$ with valid access selection $\aselect_1$,
and an accessible part $A_2$ of~$I_2$ with valid access selection
  $\aselect_2$;
\item $A_1 \subseteq A_2$;
\item $Q$ holding in~$I_1$
but not in~$I_2$
\end{compactitem}
 We first modify $I_2$ and $A_2$ to~$I_2'$ and $A_2'$ by replacing each element that is in~$I_1$ but not in~$A_1$ by
a copy that is not in~$I_1$; we modify the access selection from $\aselect_2$ to
  $\aselect_2'$ accordingly. Since $I_2'$
is isomorphic to~$I_2$,
it is clearly still true that $\aselect_2$ is valid, that  $A_2'$ is an accessible part of~$I_2'$ with access
selection $\aselect_2'$, that that~$I_2$ satisfies $\Sigma$
and that~$Q$ fails in~$I_2$. Further we still have $A_1 \subseteq A_2'$ by
  construction. What we have ensured at this step is that values of $I_2'$ that
  are in~$I_1$ must be in~$A_1$.

Consider now $I'_1\colonequals I_1 \cup I_2'$. It is clear that~$Q$ holds in~$I'_1$, and $I_1'$ also
  satisfies~$\Sigma$ because $\incd$s are preserved under taking unions.  We will show
that $I'_1$ have a common accessible part $A'_2$, which will contradict
the assumption that $Q$ is access-determined.

  Towards this goal, define an access selection $\aselect_1'$ on~$I_1'$ as follows:
  \begin{itemize}
    \item For any access $(\mt, \accbind)$ made with a binding where all values are
  in~$A_1$, we let $\aselect_1'(\mt, \accbind) \colonequals \aselect_1(\mt,
  \accbind) \cup \aselect_2(\mt, \accbind)$: note that all returned tuples are in $A_2'$ because
  the first member of the union is in~$A_2'$ and the other is in $A_1$ which is
  a subset of~$A_2'$.
    \item For any access $(\mt, \accbind)$ made with a binding where all values are
  in~$A_2'$ and some value is not in~$A_1$, we let $\aselect_1'(\mt, \accbind)
      \colonequals \aselect_2(\mt, \accbind)$: again, all the tuples returned here
      are in $A_2'$.
\item For any access $(\mt, \accbind)$ made with a binding where
  some value is not in~$A_2'$, we choose an arbitrary set of tuples of~$I_1'$ to
      form a valid output.
  \end{itemize}
  We claim that
  $\aselect_1'$ is a valid access selection and that performing the fixpoint
  process with this access selection yields $A_2'$ as an accessible part
  of~$I_1'$. To show this, first notice that performing the fixpoint
  process with $\aselect_2'$ indeed returns~$A_2'$: all facts of~$A_2'$ are
  returned because this was already the
  case in~$I_2'$, and no other facts are returned because it is clear by induction that the fixpoint process will
  only consider bindings in~$A_2'$, so that the choices made in the third point of
  the list above have no impact on the accessible part that we obtain.
  So it suffices to show that $\aselect_1'$ is valid,
  i.e., that for any access $(\mt, \abind)$ with a binding $\abind$
  in~$A_2'$, the access selection $\aselect_1'$ returns a set of tuples which
  is a valid output to the access. For the first point in the list, we know that
  the selected tuples are the union of a valid result to the access in~$I_1$
  and of a valid result to the access in~$I_2'$, so it is clear that it
  consists only of matching tuples in~$I_1'$. We then argue that it is valid by
  distinguishing two cases. If~$\mt$ is not result-bounded, then the output is
  clearly valid, because it contains all matching tuples of~$I_1$ and all
  matching tuples of~$I_2'$, hence all matching tuples of~$I_1'$. Now suppose $\mt$ has a
  result lower bound of~$k$. Suppose that for  $j \leq k$
there are $\geq j$ matching tuples in~$I_1'$. We will show that the output of the access contains 
  $\geq j$ tuples. There are two sub-cases. The first sub-case is when there are $\geq j$ matching tuples in~$I_1$.
In this sub-case we can 
  conclude because $\aselect_1(\mt, \accbind)$ must return $\geq j$
  tuples. The second sub-case is when there are $<j$ matching tuples in~$I_1$.
  In this sub-case, 
  $\aselect_1(\mt, \accbind)$ must return all of them, so these matching tuples
  are all in~$A_1$.  Hence they are all in $A_2'$ because $A_1 \subseteq A_2'$.
  Thus the returned matching tuples are in $I_2'$. Thus, in the second sub-case, all
  matching tuples in~$I_1'$ for the access are actually in~$I_2'$, so we
  conclude because $\aselect_2(\mt, \accbind)$ must return $\geq j$ tuples. This concludes
  the argument that the outputs of accesses defined in the first point are valid.
 
  For accesses corresponding to the second point in the list, by the construction used to create $I_2'$
  from~$I_2$, we know that the value in~$\accbind$ which is not in~$A_1$ cannot
  be in~$I_1$ either. Thus  all matching tuples of the access are in~$I_2'$.
  So we conclude because $\aselect_2'$ is a valid access selection of~$I_2'$.
  For accesses corresponding to the third point, the output is always valid by definition.
  Hence, we have
  established that $\aselect_1'$ is valid, and that it
  yields $A_2'$ as an accessible part of~$I_1'$.

  We have thus shown that~$I'_1$ and $I_2'$ both have $A_2'$ as an accessible
  part.
  Since $Q$ holds in~$I'_1$, by access-determinacy
$Q$ holds in~$I_2$, and this contradicts our initial assumption, concluding the
  proof.
\end{proof}

From Proposition~\ref{prop:monotonecollapse} we immediately
see that \emph{in the case where the constraints
consist of $\incd$s only, 
all the results about monotone answerability with result bounds 
transfer to answerability}. This includes simplification results and
complexity bounds.

\subsection{Enlargement for RA-answerability}
We now explain how the method of ``blowing up counterexamples'' introduced
in the body extends to work with access-determinacy.  We consider a
counterexample to access-determinacy in the simplification (i.e., a pair of
instances that satisfy the constraints and have a common subinstance that is jointly access-valid but one
satisfy the query and one does not), and we show that it can be enlarged to a
counterexample to access-determinacy in the original schema. 

\begin{definition}
  \label{def:preenlargera}
  A \emph{counterexample to access-determinacy} for a CQ~$Q$ and a schema $\aschema$ is a pair of
  instances $I_1, I_2$ both satisfying the schema constraints,
  such that~$I_1$ satisfies $Q$ while $I_2$
  satisfies $\neg Q$, and $I_1$ and $I_2$ have a common  subinstance
  $I_\acc$ that is jointly access-valid.
\end{definition}

It is clear that, whenever there is a counterexample to access-determinacy for schema
$\aschema$ and query~$Q$, then $Q$ is \emph{not} access-determined \wrt\
$\aschema$.

We now state the 
enlargement lemma that we use, which is the direct analogue of
Lemma~\ref{lem:enlarge}:

\begin{lemma}
  \label{lem:enlargera}
  Let $\aschema$ and $\aschema'$ be schemas and $Q$ a CQ on the common relations
  of $\aschema$ and~$\aschema'$ such that $Q$ that is not
  access-determined in~$\aschema'$. Suppose that for some counterexample $I_1,
  I_2$ to access-determinacy for~$Q$ in~$\aschema'$ we can construct instances
  $I_1^+$ and $I_2^+$ that satisfy the constraints of
  $\aschema$, that have a common subinstance $I_\acc$ that is jointly access-valid for~$\aschema$,
  and such that for each $p \in \{1, 2\}$, , the instance $I_p^+$ has a homomorphism to~$I_p$, and
  the restriction of $I_p$ to the relations of~$\aschema$ is a subinstance
  of~$I_p^+$.
\end{lemma}

\begin{proof}
  We prove the contrapositive of the claim. Let $Q$ be a query which is not
  access-determined in~$\aschema'$, and let $\{I_1, I_2\}$ be a counterexample.
  Using the hypothesis, we construct $I_1^+$
  and $I_2^+$. It suffices to observe that they are a counterexample to
  access-determinacy for~$Q$ and $\aschema$, which we show. First, they satisfy the
  constraints of~$\aschema$ and have a common subinstance which is jointly
  access-valid. Second, as~$I_1$ satisfies $Q$, as all relations used in~$Q$ are
  on~$\aschema$, and as the restriction of $I_1$ is a subset of~$I_1^+$, we know
  that~$I_1^+$ satisfies $Q$. Last, as~$I_2$ does not satisfy $Q$ and $I_2^+$ has
  a homomorphism to~$I_2$, we know that~$I_2^+$ does not satisfy~$Q$. Hence,
  $I_1^+, I_2^+$ is a counterexample to access-determinacy of~$Q$ in~$\aschema$,
  which concludes the proof.
\end{proof}

\subsection{Choice Simplifiability for RA-answerability} 
We say that a schema $\aschema$ is \emph{RA choice simplifiable} 
if any CQ that has an RA-plan over~$\aschema$ has one over its choice simplification.
The following result is the counterpart to Theorem~\ref{thm:simplifychoice}:

\begin{theorem} \label{thm:simplifychoicera}
 Let $\aschema$ be a schema with constraints in
  equality-free first-order logic (e.g., TGDs), and let $Q$ be a CQ that is access-determined \wrt\ $\aschema$.
Then $Q$ is also access-determined in the choice simplification of~$\aschema$.
\end{theorem}

The proof follows that of Theorem~\ref{thm:simplifychoice} with no surprises,
using Lemma~\ref{lem:enlargera}.

\begin{proof}
 We fix a counterexample $I_1, I_2$ 
 to access-determinacy in the choice simplification:
  we know that $I_1$ satisfies the query, $I_2$ violates the query,
  $I_1$ and $I_2$ satisfy the equality-free first order constraints
  of~$\aschema$,
  and $I_1$ and $I_2$ have a common subinstance $I_\acc$ which is jointly
  access-valid for
the choice simplification of~$\aschema$.
  We  expand $I_1$ and~$I_2$ to~$I_1^+$ and $I_2^+$ that have a common subinstance that is jointly access-valid
  for~$\aschema$, to conclude using Lemma~\ref{lem:enlargera}. Our construction is identical to the blow-up used in Theorem~\ref{thm:simplifychoice}:
  for each element $a$ in the domain of~$I_1$, introduce infinitely many
  fresh elements $a_j$ for~$j \in \NN_{>0}$, and identify $a_0 \colonequals a$.
  Now, define $I_1^+ \colonequals \mathrm{Blowup}(I_1)$, where 
  $\mathrm{Blowup}(I_1)$ is the instance with facts $\{R(a^1_{i_1} \ldots a^n_{i_n}) \mid R(\vec
  a) \in I_1, \vec i \in \NN^n\}$. Define $I_2^+$ from~$I_2$ in the same way.

The proof of Theorem~\ref{thm:simplifychoice} already showed that~$I_1$ and $I_1^+$ agree on all equality-free first-order
  constraints, that
$I_1$ still satisfies the query, and that $I_2$ still
violates the query.
All that remains now is to construct a common subinstance 
that is jointly access-valid for~$\aschema$. We will do this
as in the proof of Theorem~\ref{thm:simplifychoice}, 
setting  $I_\acc^+ \colonequals \mathrm{Blowup}(I_\acc)$. 
To show that $I_\acc^+$ is jointly access-valid, consider any access $(\mt,
\accbind)$ with values from~$I_\acc^+$. If there are no matching tuples in~$I_1$ and in~$I_2$,
then there are no matching tuples in~$I_1^+$ and~$I_2^+$ either. Otherwise,
there must be some matching tuple in $I_\acc$ because it is
jointly access-valid in~$I_1$ and~$I_2$ for the choice approximation
of~$\aschema$. Hence, sufficiently many copies exist in~$I_\acc^+$ to satisfy
the original result bounds, so that we can find a valid response to the access
in~$I_\acc^+$. Hence, $I_\acc^+$ is indeed jointly access-valid, which completes
the proof.
\end{proof}

As with choice simplification for~$\amd$, this result can be applied immediately
to TGDs. In particular,
if we consider frontier-guarded TGDs, the above result says that we can assume any result bounds
are $1$, and thus the  query containment problem produced
by Proposition~\ref{prop:reducera} will  involve only frontier-guarded TGDs.
We thus get the following analog of Theorem~\ref{thm:decidegf}:

\begin{theorem} \label{thm:decidegfra}
  We can decide whether a CQ is answerable with respect to a schema with
  result bounds whose constraints are frontier-guarded TGDs.
  The problem is $\twoexp$-complete.
\end{theorem}

\subsection{FD Simplifiability for RA-plans} 

Recall the definition of FD simplification from Section~\ref{sec:simplify}.
A schema is \emph{FD simplifiable for RA-plans} if every CQ having a plan over the schema
has an RA-plan in its FD simplification. 

We now show that schemas whose constraints consist  only of FDs are FD
simplifiable, which is the analogue of Theorem~\ref{thm:fdsimplify}:

\begin{theorem} \label{thm:fdsimplifyra} 
  Let $\aschema$ be a schema whose constraints are FDs, and let
$Q$ be a CQ that is answerable in~$\aschema$. Then $Q$ is answerable in the FD simplification of~$\aschema$.
\end{theorem}
\begin{proof}
  We use Lemma~\ref{lem:enlargera} and assume that we have a counterexample
  $I_1, I_2$ to determinacy for the FD simplification of $\aschema$, with $Q$
  holding in~$I_1$, with $Q$ not holding in~$I_2$, and with $I_1$ and $I_2$
  having a common subinstance $I_\acc$ which is jointly access-valid in~$I_1$
  and~$I_2$ for the FD simplification of~$\aschema$. We will upgrade these
  to~$I_1^+$, $I_2^+$, $I_\acc^+$ having the same property for~$\aschema$, by
  blowing up accesses one after the other.
  
  We do so similarly to the proof of
  Theorem~\ref{thm:fdsimplify} in Appendix~\ref{app:simplifyfds}.
  Consider each access $(\mt, \abind)$ using a method $\mt$ on relation $R$ with
  binding $\abind$ having values in $I_\acc$.
  Let $M_1$ be the matching tuples for $(\mt, \abind)$ in~$I_1$,
  and $M_2$ the matching tuples in~$I_2$: the definition of~$I_\acc$ and the
  constraints added in the FD simplification still ensure that
  $M_1$ and $M_2$ must either intersect or be both empty. 
 If $M_1$ and $M_2$ are
  both empty or they are both singletons (which are then identical), then we do nothing for the access $(\mt,
  \abind)$: intuitively, we can already define a valid output to this access
  in~$I_1$ and~$I_2$ for~$\aschema$.
  Otherwise, we know that $M_1$ and $M_2$ are both non-empty and that one of them is not a singleton.
  Let $k$ be the result bound
  of~$\mt$.
Recall that $\detby(\mt)$ denotes the positions determined under the FDs by
  the input positions of~$\mt$: the tuples of~$M_1$ and
  of~$M_2$ must agree on $\detby(\mt)$. Let~$X$
  be the other positions of~$R$ that are \emph{not} in~$\detby(\mt)$:
  again the set $X$ must be non-empty, since otherwise $M_1$ and $M_2$ would both be
  singletons, contradicting our assumption.

  We then blow the access up exactly like in the proof of
  Theorem~\ref{thm:fdsimplify} in Appendix~\ref{app:simplifyfds}, and define
  $I_1^+$, $I_2^+$, and $I_\acc^+$ as the result of performing this process for
  all accesses in~$I_\acc$.

  As in the proof of Theorem~\ref{thm:fdsimplify}, it is still the case that
  $I_1 \subseteq I_1^+$, that $I_2 \subseteq I_2^+$, that $I_\acc \subseteq
  I_\acc^+$, that $I_\acc$ is a common subinstance of~$I_1^+$ and $I_2^+$,  and
  that for every $p \in \{1, 2\}$ the instance
  $I_p^+$ has a homomorphism back to~$I_p$. Further, it is still the case that
  $I_1^+$ and $I_2^+$ satisfy the FDs.

  The only point to verify is that $I_\acc^+$ is jointly access-valid
  in~$I_1^+$ and $I_2^+$. Consider a method $\mt$ and binding $\abind$. Like in 
  the proof of Theorem~\ref{thm:fdsimplify}, we can focus on the case where
  $\abind$ consists of values of $\dom(I_\acc)$. In this case, let $M_1$ and
  $M_2$ be the matching tuples for the access in $I_1$ and $I_2$ respectively.
  As in the proof of Theorem~\ref{thm:fdsimplify}, if we performed the blowup
  process for this access, then we can use the corresponding tuples to define an
  output to the access which is in~$I_\acc^+$ and is a valid output both
  in~$I_1^+$ and in~$I_2^+$. Observe now that the analogue of Claims~\ref{clm:extends1}
  and~\ref{clm:extends2} still hold, so letting $M_1$ and~$M_2$ be the matching
  tuples for the access in~$I_1$ and~$I_2$ respectively, we have $M_1 = M_1^+$
  and $M_2 = M_2^+$. We can then finish the proof in the same way that we
  finished the proof of Theorem~\ref{thm:fdsimplify}, in particular restricting
  the final counterexamples to the relations of~$\aschema$, and conclude using
  Lemma~\ref{lem:enlargera}.
\end{proof}

\subsection{Complexity of RA-answerability for FDs}
\label{apx:complexityfdspmr}
In Theorem~\ref{thm:decidfd} we showed that monotone answerability with FDs was decidable in the lowest possible
complexity, namely, $\np$. 

The argument involved first showing
\emph{FD-simplifiability}, which
allowed us to eliminate result bounds at the cost of adding additional $\incd$s.
We then simplified the resulting rules to
ensure that the chase would terminate.  This relied on the fact that the axioms for~$\amd$ would include
rules going from~$R$ to~$R'$, but not vice versa. Hence, the argument
does not generalize for the rules that axiomatize RA plans.

However,  we can repair the argument at the cost of adding an additional assumption.
A schema $\aschema$ with access methods is \emph{single method per relation}, abbreviated $\smpr$, if for
every relation there is at most one access method. This assumption was made
 in many works on access methods \cite{access1,access1}, although we do not make it by default
elsewhere in this work. We can then show the following analogue of
Theorem~\ref{thm:decidfd} with this additional assumption:

\begin{theorem} \label{thm:decidfdrasmpr}
  We can decide whether a CQ~$Q$ is answerable with respect to an $\smpr$ schema with
  result bounds whose constraints are FDs. The problem is $\np$-complete.
\end{theorem}

We will actually show something stronger:
for~$\smpr$ schemas with constraints consisting of FDs only,
 there is no difference between full answerability and monotone answerability.
Given Theorem~\ref{thm:decidfd}, this immediately implies
Theorem~\ref{thm:decidfdrasmpr}.

\begin{proposition} \label{prop:monotonecollapsefdsmpr}
Let $\aschema$ be a  schema with access methods satisfying $\smpr$ and constraints $\Sigma$ consisting
of functional dependencies, and $Q$ be a CQ that is access-determined.
Then $Q$ is $\amd$.
\end{proposition}
\begin{proof}
We know from Theorem~\ref{thm:fdsimplifyra} that the schema is FD simplifiable.
Thus we can eliminate result bounds as follows:

\mylistskip
\begin{compactitem}
  \simplifyfddef
\end{compactitem}

By Proposition~\ref{prop:reducera} we know that~$Q$ is access-determined exactly
when $Q \subseteq_\Gamma Q'$, where $\Gamma$ contains two copies of the above schema
and also axioms of the following form for each access method~$\mt$:
\begin{itemize}
  \item (Forward):
    \[\Big(\bigwedge_i \accessible(x_i)\Big) \wedge S(\vec x, \vec y) \rightarrow
    \Big(\bigwedge_i \accessible(y_i)\Big)
    \wedge S'(\vec x, \vec y)\]
\item (Backward): 
  \[\Big(\bigwedge_i \accessible(x_i)\Big) \wedge S'(\vec x, \vec y) \rightarrow
    \Big(\bigwedge_i \accessible(y_i)\Big)
    \wedge S(\vec x, \vec y)\]
\end{itemize}
where $\vec x$ denotes the input positions of~$\mt$.
Note that~$S$ may be one of the original relations, or one of the relations~$R_\mt$ produced by the transformation
above.

We now show that  chase proofs with~$\Gamma$ must in fact be very simple under
the 
$\smpr$ assumption:
\begin{claim} Assuming our schema is $\smpr$, consider any chase
sequence for~$\Gamma$. Then:
\begin{compactitem}
\item Rules of the form
$R_\mt(\vec x, \vec y)
\rightarrow \exists \vec z ~
  R(\vec x, \vec y, \vec z)$
will never fire.
\item  Rules of the form
$R'(\vec x, \vec y, \vec z) \rightarrow R'_\mt(\vec x, \vec y)$
will never fire.
\item FDs will never fire (assuming they were applied to the initial
instance)
\item (Backward) axioms will never fire.
\end{compactitem}
\end{claim}
  Note that the last item suffices to conclude that
  Proposition~\ref{prop:monotonecollapsefdsmpr} holds, so it suffices to prove the claim.
We do so by induction.
We consider the first item.
Consider a fact $R_\mt(\vec c, \vec d)$. Since the (Backward) axioms never
  fire (fourth point of the induction), the fact must have been produced from a fact $R(\vec c, \vec d, \vec e)$. Hence
the axiom can not fire on this fact, because we only fire
active triggers.

We move to the second item, considering a fact $R'(\vec c, \vec d, \vec e)$.
By $\smpr$ and the inductive assumption that FDs do not fire, this fact can only have been produced by $R'_\mt(\vec c, \vec d)$.
Thus the rule in question will not fire in the chase.

Turning to the third item, we first consider a potential violation of an FD $D
\determines r$ on an unprimed relation~$R$. This consists of facts $R(\vec c)$ and $R(\vec d)$ agreeing
on positions in~$D$ and disagreeing on position~$r$. As the initial instance is
always assumed to satisfy the FDs, these facts are not in the initial instance.
But they could not have been otherwise produced, as we know by induction (first
and fourth points) that
none of the rules with an unprimed relation $R$ in their head will fire.
Now let us turn to facts that are potential violations of the primed copies
of the FDs, for some relation~$R'$. The existence of the violation implies that
there is an access method on the corresponding relation~$R$ in the original
schema. By the SPMR assumption there is exactly one such method.

We first consider the case where this access method has result bounds.
We know that the facts in the violation must have been produced by
the rule going from~$R'_\mt$ to~$R'$ or by a Forward rule. Thus the facts are
$R'(\vec c_1, \vec d_1, \vec e_1)$ and $R'(\vec c_2, \vec d_2, \vec e_2)$. Let
us assume that
that  $R'(\vec c_2, \vec d_2, \vec e_2)$ was the latter of the two facts to be created,
then $\vec e_2$ would have been chosen fresh. Hence the violation must occur within
the positions corresponding to~$\vec c_1, \vec d_1$ and $\vec c_2, \vec d_2$.
But by induction (third point), and by the~$\smpr$ assumption, these facts must have been created from
facts $R'_\mt(\vec c_1, \vec d_1)$ and $R'_\mt(\vec c_2, \vec d_2)$ where $\mt$
is the only access method on~$R$, and in turn these must
have been created from facts $R_\mt(\vec c_1, \vec d_1)$ and $R_\mt(\vec c_2, \vec d_2)$. These last
must (again, by induction, using the third and fourth points) have been created from facts
$R(\vec c_1, \vec d_1, \vec f_1)$ and $R(\vec d_1, \vec d_1, \vec g_1)$. But then we have
an earlier violation of the FDs on these two facts, which is a contradiction.

We now consider the second case, where the access method on~$R$ has no result bounds in the
original schema. In this case there
is no relation~$R_\mt'$ and the facts of the violation must have been produced
by applying the Forward rule. But then the $R$-facts used to create them must
themselves be an earlier violation of the corresponding FD on~$R$, which is
again a contradiction. Hence, we have shown the third item.

Turning to the last item, there are two kinds of Backward rules to consider.
First, the ones involving a primed relation $R'$ and the original relation $R$,
where there is an access method without result bounds on~$R$ in the original
schema. Secondly, the ones
involving a primed relation $R'_\mt$ and the unprimed relation $R_\mt$ where
there is an access method with result bounds on~$R$ in the original schema. 
For the first kind of axiom, any $R'$-fact can only have been created from an
$R$-fact using the Forward axioms, and so the Backward axiom cannot fire.
For the second kind of axiom, we show the claim 
by considering a fact $R'_\mt(\vec c, \vec d)$. Using the second point of the
induction, it can
only have been generated by a fact $R_\mt(\vec c, \vec d)$, and thus (Backward)
could not fire, which establishes the desired result.
\end{proof}

Without $\smpr$, we can still argue that RA-answerability is decidable, and show
a singly exponential complexity upper bound:
\begin{theorem} \label{thm:decidfdraexp}
For general schemas with access methods and constraints $\Sigma$
consisting of FDs, the RA-answerability problem is decidable in~$\exptime$.
\end{theorem}
\begin{proof}
We consider again the query containment problem for answerability obtained after
  eliminating result bounds, and let
$\Gamma$ be the corresponding constraints as in
  Proposition~\ref{prop:monotonecollapsefdsmpr}.

Instead of claiming that neither the FDs nor the backward axioms will not fire, as in the
case of~$\smpr$, we argue only that  the FDs will not fire.
From this it follows that the constraints consist only of $\incd$s
and accessibility axioms, leading to an $\exptime$ complexity upper bound:
one can apply either  Corollary \ref{cor:exp} from the body of the paper, 
or the~$\exptime$
complexity result without result bounds from \cite{bbbicdt}.

We consider a chase proof with~$\Gamma$, and claim, for each
relation~$R$ and each result-bounded method $\mt$ on~$R$, the following
  invariant:
\begin{itemize}
\item Every $R_\mt$-fact and every $R'_\mt$-fact
 is a projection of some $R$-fact or some $R'$-fact.
\item All the FDs are satisfied in the chase instance,
and further  for any relation~$R$, $R \cup R'$ satisfies the FDs. 
That is: for any FD $D \determines r$,
we cannot have an $R$ and $R'$-fact that agree on positions in~$D$ and disagree
on~$r$.
\end{itemize}
The second item of the invariant implies that the FDs do not fire, which as we
  have argued is sufficient to conclude our complexity bound.

 The invariant is initially true, by assumption that
FDs are applied on the initial instance.
When firing an $R$-to-$R_\mt$ axiom or an $R'$-to-$R'_\mt$ axiom, the first
item is preserved by definition,  and the second is trivially preserved since
there are no FDs on~$R_\mt$ or $R'_\mt$.

When firing an accessibility axiom, either forward or backward, again the first  and the
  second  item are clearly preserved.

Now, consider the firing of an $R_\mt$-to-$R$ axiom.
 The first item is trivially  preserved, so we must only show the second.

 Consider the fact
$R_\mt(a_1 \ldots a_m)$ and the generated 
fact $F = R(a_1 \ldots a_m, b_1 \ldots b_n)$
   created by the rule firing. Assume that~$F$ is part of an FD violation with some
  other fact $F'$ which is of the form $R(a'_1 \ldots a'_m, b'_1 \ldots b'_m)$ or
$R'(a'_1 \ldots a'_m, b'_1 \ldots b'_m)$.

We know that the left-hand-side of the FD cannot contain any
 of the positions of the~$b_i$, because they are fresh nulls. Hence, the
  left-hand-side of the FD is
 included in the positions of~$a_1 \ldots a_m$. But now, by definition of the FD
  simplification, the right-hand-side of the FD cannot correspond
to one of  the~$b_1 \ldots b_n$, since otherwise that position
would have been  included in~$R_\mt$.
  So the right-hand-side is also one of the positions of~$a_1 \ldots a_m$, and
  in particular we must have $a_i \neq a_i'$ for some $1 \leq i \leq m$ in the
  right-hand-side of the FD..

 Now we use the first item of the inductive
  invariant on the fact $R_\mt(a_1 \ldots a_m)$: there was already a fact $F''$, either an $R$ or
$R'$-fact, with tuple of values $(a_1 \ldots a_m, b''_1 \ldots b''_m)$.
As there is $1 \leq i \leq m$ such that~$a'_i \neq a_i$, the tuples of values
of~$F'$ and $F''$ must be different. 
But now, as~$F$ and $F'$ are an FD violation on the positions
  $a_1 \ldots a_m$, then $F'$ and $F''$ are seen to also witness an FD
   violation in~$R \cup R'$ that existed before the firing. This
  contradicts the first point of the invariant, so we conclude that the second
  item is preserved when firing an $R_\mt$-to-$R$ axiom.

When firing $R'_\mt$-to-$R'$ rules, the symmetric argument applies.

This completes the proof of the invariant, and concludes the proof of
Theorem~\ref{thm:decidfdraexp}.
\end{proof}

\begin{table*}[th!]
  \centering
  {
\caption{Summary of results on simplifiability and complexity of RA-answerability}
\label{tab:resultsfull}
  \begin{tabular}{lll}
\toprule
{\bfseries Fragment} &
{\bfseries Simplification} & 
{\bfseries Complexity} \\
\midrule
  IDs & Existence-check  (Theorem~\ref{thm:simplifyidsexistence} and Prop.~\ref{prop:monotonecollapse}) & $\exptime$-complete  
(Theorem~\ref{thm:decidids} and Prop.~\ref{prop:monotonecollapse})\\
  Bounded-width IDs & Existence-check (see above) & $\np$-complete  (Theorem~\ref{thm:npidsbounds} and Prop. \ref{prop:monotonecollapse}) \\
  FDs & FD (Theorem~\ref{thm:fdsimplifyra}) & 
    In $\exptime$ (Theorem~\ref{thm:decidfdraexp}) \\
  FDs under $\smpr$ & FD (see above) &
    $\np$-complete (Theorem~\ref{thm:decidfdrasmpr})\\
  FDs and UIDs & Choice  (Theorem~\ref{thm:simplifychoiceuidfdra})
    &  \emph{Open}   \\
  Equality-free FO & Choice  (Theorem~\ref{thm:simplifychoicera})
    & Undecidable (same proof as Prop.~\ref{prp:undec}) \\
  Frontier-guarded TGDs & Choice  (see above) & $\twoexp$-complete  (Theorem~\ref{thm:decidegfra}) \\
\bottomrule
\end{tabular}
  }
\end{table*}

\subsection{Choice Simplifiability for RA-plans with $\uincd$s and FDs}

\begin{theorem} \label{thm:simplifychoiceuidfdra}
Let schema $\aschema$ have  constraints  given by
 $\uincd$s and arbitrary FDs,  and $Q$ be a CQ that is access-determined \wrt\ $\aschema$.
Then $Q$ is also access-determined in the choice simplification of~$\aschema$.
\end{theorem}

We will proceed in a similar fashion to Theorem~\ref{thm:simplifychoiceuidfd},
i.e., fixing one access at a time,
using the following enlargement lemma as the analogue to
Lemma~\ref{lem:enlargeprog}:

\begin{lemma}
  \label{lem:enlargeprogra}
  Let $\aschema$ be a schema and $\aschema'$ be its choice simplification,
  and let $\Sigma$ be the constraints.

  Assume that, for any CQ~$Q$ not access-determined in~$\aschema'$, for any
  counterexample $I_1, I_2$ of access-determinacy for~$Q$ and
  $\aschema'$ with witness a common subinstance $I_\acc$  that it is
  jointly access-valid in~$I_1$ and~$I_2$ for~$\aschema'$,
  for any access $\mt, \accbind$ in~$I_\acc$,
  the following holds:
  we can construct a counterexample $I_1^+, I_2^+$ to
  access-determinacy for~$Q$ and $\aschema'$,
  i.e., $I_1^+$ and $I_2^+$ satisfy $\Sigma$,
  $I_1 \subseteq I_1^+$, $I_2 \subseteq I_2^+$,
  $I_1^+$ has a homomorphism to~$I_1$ and $I_2^+$ has a homomorphism to~$I_2$,
  and $I_1^+$ and $I_2^+$ have a common subinstance $I_\acc^+$ (which is 
  again jointly access-valid in~$I_1^+$ and $I_2^+$ for~$\aschema'$),
  and we can further impose that:
  \begin{enumerate}
    \item $I_\acc^+$ is a superset of~$I_\acc$;
    \item there is an output to the access $\mt, \accbind$ in~$I_\acc^+$ which is valid
  in~$I_1^+$ and~$I_2^+$ for~$\aschema$;
    \item for any access in~$I_\acc$ having an output
      in~$I_\acc$ which is valid for~$\aschema$ in~$I_1$ and
      $I_2$, there is an output to this access in~$I_\acc^+$
  which is valid for~$\aschema$ in~$I_1^+$ and in~$I_2^+$;
  \item for any
  access in~$I_\acc^+$ which is not an access in~$I_\acc$, there is an output
      in~$I_\acc^+$ which is valid for
      $\aschema$ in~$I_1$ and in~$I_2$.
  \end{enumerate}
  Then any query which is access-determined in~$\aschema$ is also
  access-determined in~$\aschema'$.
\end{lemma}

The proof of this lemma is exactly like that of Lemma~\ref{lem:enlargeprog}.

We are now ready to prove Theorem~\ref{thm:simplifychoiceuidfdra} using the process
of Lemma~\ref{lem:enlargeprogra}. We proceed similarly to the proof of
Theorem~\ref{thm:simplifychoiceuidfdra}.

Let $Q$ be a CQ which is not access-determined in~$\aschema'$,
let $I_1, I_2$ be a counterexample to access-determinacy,
and let $I_\acc$ be a common subinstance of~$I_1$ and~$I_2$ for~$\aschema'$
which is jointly access-valid in~$I_1$ and~$I_2$ for~$\aschema'$.
Let $(\mt, \accbind)$ be
an access on relation~$R$ in~$I_\acc$ which does not necessarily have an output
which is valid for~$\aschema$. As in the proof of
Theorem~\ref{thm:simplifychoiceuidfdra}, if there are no matching tuples in~$I_1$
for~$(\mt, \accbind)$, then there are no matching tuples in~$I_\acc$ either,
so the access $(\mt, \accbind)$ already has a valid output for~$\aschema$ and there
is nothing to do. The same holds if there are no matching tuples in~$I_2$.
Now, if there is exactly one matching tuple in~$I_1$ and exactly one matching
tuple in~$I_2$, as~$I_\acc$ is jointly access-valid for~$\aschema'$, it
necessarily contains those matching tuples, so that, as~$I_\acc \subseteq I_1$
and $I_\acc \subseteq I_2$, the matching tuple in~$I_1$ and~$I_2$ is the same,
and again there is nothing to do: the access $(\mt, \accbind)$ already has a
valid output for~$\aschema$.

Hence, the only interesting case is when there is a matching tuple to the access
in~$I_1$ and in~$I_2$, and there is more than one matching tuple in one of the
two. As $I_1$ and $I_2$ play a symmetric role in the hypotheses of
Lemma~\ref{lem:enlargeprogra}, we assume without loss of generality that it is
$I_1$ which has multiple matching tuples for the access.

As $I_\acc$ is access-valid in~$I_1$ for~$\aschema'$, we know
that~$I_\acc$ contains at least one of these tuples, say $\vec t_1$. As $I_\acc
\subseteq I_2$, then $I_2$ also contains $\vec t_1$. As in the previous proof,
we take $\vec t_2$ a different matching tuple in~$I_1$, let $C$ be the non-empty
set of positions where $\vec t_1$ and $\vec t_2$ disagree, and observe that
there is no FD implied from the complement of~$C$ to a position of~$C$.

We construct $I_1^+ \colonequals I_1 \cup W$ and $I_2^+ \colonequals I_2 \cup W$ as in the previous proof, and we
show that $(I_1^+, I_2^+)$ is a counterexample to determinacy for~$Q$
and~$\aschema'$:

\begin{itemize}
  \item We show as in the previous proof that $I_1^+$ and $I_2^+$ satisfy the $\uincd$s and the
FDs of~$\Sigma$.
    \item We clearly have $I_1 \subseteq I_1^+$ and $I_2 \subseteq
I_2^+$.
    \item The homomorphism from~$I_1^+$ to~$I_1$ and from~$I_2^+$ to~$I_2$ is
defined as in the previous proof.
    \item We define $I_\acc^+ \colonequals I_\acc \cup W$
      a common subinstance of~$I_1^+$ and~$I_2^+$ and we 
      must show that $I_\acc^+$ is 
    jointly access-valid in~$I_1^+$ and~$I_2^+$ for~$\aschema'$.
    We do this as in the
    previous proof. First, for accesses that include an element of~$\dom(W) \setminus
    \dom(I_\acc)$, the matching tuples are all in~$W$ so they are in~$I_\acc^+$.
    Second,
    for accesses on~$\dom(I_\acc)$, the matching tuples include matching tuples
    of~$W$, which are in~$I_\acc^+$, and matching tuples for that access
    in~$I_\acc$ in~$I_1$ and~$I_2$, so we can construct a valid output
    to this access for~$\aschema'$ because $I_\acc$ is jointly access-valid
    in~$I_1$ and~$I_2$.
\end{itemize}
What remains to be able to use Lemma~\ref{lem:enlargeprogra} is to show the
four additional conditions:

\begin{enumerate}
  \item It is immediate that~$I_\acc^+ \supseteq I_\acc$.
  \item The access $(\mt, \accbind)$ has an output in~$I_\acc^+$ which is valid for~$\aschema$ 
    in~$I_1^+$ and~$I_2^+$. This is established as in the previous proof: there are now
    infinitely many matching tuples for the access in~$I_1^+$ and~$I_2^+$, so we can choose
    as many as we want in~$W$ to obtain an output in~$I_\acc^+$ which is valid
    for~$\aschema$ in~$I_1^+$ and~$I_2^+$.
  \item For every access of~$I_\acc$ that has an output which is valid for~$\aschema$ in~$I_1$
    in~$I_2$, then we can construct such an output in~$I_\acc^+$ which is 
valid for~$\aschema$ in~$I_1^+$ and~$I_2^+$.
    This is similar to the fourth bullet point above. From the 
    output $U$ to the access in~$I_\acc$ which is valid for~$I_1$ and $I_2$,
    we construct an 
    output to the access in~$I_\acc^+$ which is valid for~$I_1^+$ and~$I_2^+$, using the tuples
    of~$U$ and the matching tuples in~$W$.
  \item All accesses of~$I_\acc^+$ which are not accesses of~$I_\acc$ have an
    output which is valid
    for~$\aschema$ in~$I_1^+$ and~$I_2^+$. As before, such accesses must include
    an element of~$W$, so by the fourth bullet point all matching tuples are
    in~$W$, so they are all in~$I_\acc^+$.
\end{enumerate}
Hence, we have explained how to fix the access $(\mt, \accbind)$, so we can
conclude using Lemma~\ref{lem:enlargeprogra} that we obtain a counterexample to
access-determinacy of~$Q$ in~$\aschema$ by fixing all accesses. This concludes
the proof.

\subsection{Summary of Extensions to Answerability with RA-plans}
Table \ref{tab:resultsfull} summarizes the expressiveness and complexity results
for RA-plans. There are three differences with the corresponding table for
monotone answerability (Table~\ref{tab:results} in the body):
\begin{itemize}
\item For RA-plans, while we know that choice simplifiability holds with FDs and
  UIDs, we do not know whether answerability is decidable. Indeed, in the
    monotone case, when proving Theorem~\ref{thm:deciduidfd}, we had used a separability argument
    to show that FDs could be ignored for FDs and UIDs (see the sketch of the
    proof of Theorem~\ref{thm:deciduidfd} in Section~\ref{sec:complexitychoice}
    in the body, and Appendix~\ref{apx:separable}). We do not have such an
    argument for answerability with RA plans.
\item For RA-plans, our tight complexity bound for answerability with FDs in
  isolation holds only under the $\smpr$ assumption; see
    Appendix~\ref{apx:complexityfdspmr} for
    details.
  \item The results of
    Table~\ref{tab:resultsfull} are only claimed for \emph{unrestricted} RA-answerability,
    i.e., RA-answerability over all instances, finite or infinite. By contrast,
    the results of Table~\ref{tab:results} hold both for finite monotone
    answerability and for unrestricted monotone answerability.
\end{itemize}

\end{document}